\documentclass[preprint,12pt]{elsarticle}

\usepackage[a4paper, margin=1.7cm]{geometry}

\usepackage{geometry}
\usepackage{algorithm}
\usepackage{algpseudocode}
\usepackage{amsmath,amssymb}
\usepackage{color}
\usepackage{subcaption}
\usepackage{amssymb}
\usepackage{amsmath}\usepackage[scr=rsfs]{mathalfa}
\usepackage{enumerate}
\usepackage{amsthm}
\usepackage{array}
\usepackage{changepage}
\usepackage{soul}

\newtheorem{result}{{\bf Result}}
\newtheorem{observation}{{\bf Observation}}

\newtheorem{definition}{{\bf Definition}}
\newtheorem{theorem}{{\bf Theorem}}
\newtheorem{lemma}{{\bf Lemma}}

\usepackage{amssymb}
\usepackage{amsmath}

\journal{Nuclear Physics B}

\begin{document}

\begin{frontmatter}



\title{Min-Sum Uniform Coverage Problem by Autonomous Mobile Robots}

\author[label1]{Animesh Maiti}
\author[label2]{Abhinav Chakraborty}
\author[label3]{Bibhuti Das}
\author[label1]{Subhash Bhagat}
\author[label4]{Krishnendu Mukhopadhyaya}


\affiliation[label1]{organization={Department of Mathematics},
            addressline={Indian Institute of Technology Jodhpur}, 
            country={India}}

\affiliation[label2]{organization={Department of Mathematics},
            addressline={Birla Institute of Technology, Mesra}, 
            country={India}}


 \affiliation[label3]{organization={ D\'epartement d'informatique, Universit\'e du Qu\'ebec en Outaouais},
            country={Canada}}

 \affiliation[label4]{organization={Advanced Computing and Microelectronics Unit, Indian Statistical Institute},
            country={India}}




\begin{abstract}
We study the \textit{min-sum uniform coverage} problem for a swarm of $n$ mobile robots on a given finite line segment and on a circle having finite positive radius, where the circle is given as an input. The robots must coordinate their movements to reach a uniformly spaced configuration that minimizes the total distance traveled by all robots. The robots are autonomous, anonymous, identical, and homogeneous, and operate under the \textit{Look–Compute–Move} (LCM) model with \textit{non-rigid} motion controlled by a fair asynchronous scheduler. They are oblivious and silent, possessing neither persistent memory nor a means of explicit communication. In the \textbf{line-segment setting}, the \textit{min-sum uniform coverage} problem requires placing the robots at uniformly spaced points along the segment so as to minimize the total distance traveled by all robots. In the \textbf{circle setting} for this problem, the robots have to arrange themselves uniformly around the given circle to form a regular $n$-gon. There is no fixed orientation or designated starting vertex, and the goal is to minimize the total distance traveled by all the robots. We present a deterministic distributed algorithm that achieves uniform coverage in the line-segment setting with minimum total movement cost. For the circle setting, we characterize all initial configurations for which the \textit{min-sum uniform coverage} problem is deterministically unsolvable under the considered robot model. For all the other remaining configurations, we provide a deterministic distributed algorithm that achieves uniform coverage while minimizing the total distance traveled. These results characterize the deterministic solvability of min-sum coverage for oblivious robots and achieve optimal cost whenever solvable.

\end{abstract}


\end{frontmatter}


\section{Introduction}
  

\noindent Coordination problems in multi-robot systems have received significant attention over the past years \cite{santorobook2}. A \textit{robot swarm} is a distributed multi-robot system composed of autonomous (no centralized controller), anonymous (no unique identifier), homogeneous (execute the same algorithm) and identical mobile robots that cooperatively accomplish a given task. The study aims to identify the minimal capabilities required for robots to accomplish a given task. In the literature, a large volume of work studies gathering \cite{DBLP:conf/icalp/CieliebakFPS03, DBLP:journals/siamcomp/CohenP05,DBLP:journals/tcs/FlocchiniPSW05,DBLP:journals/siamcomp/SuzukiY99}, pattern formation \cite{DBLP:journals/tcs/FlocchiniPSW08,DBLP:conf/sirocco/SuzukiY96,DBLP:journals/siamcomp/SuzukiY99}, {\it circle formation problem} \cite{128452,DBLP:conf/pomc/DefagoK02,DBLP:journals/cacm/Debest95} etc. The circle formation problem asks for a final configuration in which all the robots lie on a circle centered at a point, that is not defined a priori, within a finite time. 
The design of such multi-robot systems is inspired by the coordinated behaviours exhibited by social insects in natural environments. These systems are especially effective in challenging situations, where explicit human intervention is not possible. Since these systems operate without centralized coordination, they can function effectively in adversarial and unexplored environments. A group of small robots is often more cost-efficient than using a single large robot. Such systems are applied in many areas, including search-and-rescue in disaster-affected regions, exploration and mapping of unknown or hazardous environments, security, surveillance,  mining, agriculture, and other large-scale autonomous tasks where human involvement is limited due to safety concerns \cite{1500631,TAN201318}.

\subsection{General Model}
\label{sb-1}

\noindent Robots are autonomous (no centralized controller), homogeneous (run the same distributed algorithm), and anonymous (no unique identifier). 

\begin{itemize}
 \item {\bf Deployment space:} Robots may be deployed in a continuous or a discrete space. In the continuous setting, the environment is typically modeled as a $d$-dimensional Euclidean space, for $d \ge 1$. In the discrete setting, the robots are usually placed at the nodes of a given graph and are allowed to move along its edges only. In this paper, we assume that the robots are deployed in both continuous one-dimensional and two-dimensional spaces.
 \item {\bf Coordinate system:} The robots do not share a common global coordinate system; instead, each robot is equipped with its own local coordinate system whose origin is located at its current position. The directions and orientations of the coordinate axes and the unit distances of two distinct robots may vary. However, in some computational models, it is assumed that robots have some agreement on the directions and orientations of one or both of the local coordinate axes. Robots may share a common sense of handedness (clockwise direction), i.e., {\it chirality}. In this paper, we assume that the robots are disoriented, meaning that they do not agree on either the orientation or the direction of their local coordinate axes, nor do they share a common chirality.
 \item {\bf Visibility of a robot:} The robots are equipped with sensing capability, referred to as their \textit{visibility}. They use this capability to detect the positions of the other robots in the system. The visibility range may be either limited or global. In the \textit{limited visibility} model, a robot can observe only those robots located within a fixed radius centered at its current position. Under the \textit{global visibility} model, each robot can sense the positions of all the other robots in the system. In the \textit{obstructed visibility} model, a robot $r_2$ blocks the visibility between robots $r_1$ and $r_3$, if $r_2$ lies on the line segment joining $r_1$ and $r_3$. In this paper, it is assumed that the robots have \textit{global visibility}.

 \item {\bf Communications:} Robots are {\it silent}, i.e., there is no direct explicit communication capability.

 \item {\bf Computational cycle:} At any given time, a robot is either active or inactive (idle). Once activated, it executes a computational cycle comprising three phases: \textit{Look–Compute–Move} (LCM). After completing a LCM cycle, a robot may either become inactive again or initiate a new computational cycle. The robots continue to repeat these cycles until the goal is achieved.

 \begin{itemize}

  \item \textit{Look phase:} Upon activation, a robot enters the Look phase, during which it observes the positions of all the other robots in its own local coordinate system. In this work, during the \textit{Look} phase, a robot observes the positions of all the robots on the circumference of the circle $\mathcal C$ given as an input.

  \item \textit{Compute phase:} In this phase, a robot computes a destination point based on the positions obtained in the \textit{Look} phase. The computed destination point may coincide with its current position.

  \item \textit{Move phase:} During the Move phase, a robot moves toward the destination point computed in the \textit{Compute} phase. Two types of movements are considered: \textit{rigid} and \textit{non-rigid}. In the case of rigid movement, a robot reaches its destination without interruption. In the case of non-rigid movement, the motion is controlled by an adversary, which may cause the robot to stop before it reaches its destination. However, to ensure finite-time progress, it is assumed that there exists a $\delta > 0$ such that, whenever a robot does not reach its destination, it moves at least a distance of $\delta$ toward it. The constant $\delta$ is called the \textit{rigidity constant}. Its value may be unknown to the robots, and each robot $r_i$ may have its own local rigidity constant $\delta_i$, with $\delta = \min\{\delta_1,\delta_2,\ldots,\delta_n\}$. If the destination coincides with the robot's current position, it performs a \textit{null movement}. In traditional models, the robots typically move in a straight line. However, in specific models, the robots can even perform guided movements, i.e., they can move along some specified curve \cite{PATTANAYAK2019145,DBLP:journals/dc/CiceroneSN19}. In this work, we consider \textit{non-rigid} movement, where the robots are restricted to moving along the circumference of the input circle $\mathcal{C}$.

\end{itemize}
\item {\bf Memory of a robot:} The
robots are oblivious in the sense that their memory is volatile. They do not have
any memory of past observations, computations and actions. At the end of each
Look-Compute-Move cycle, the memory is erased. 

\item {\bf Scheduler:} The activation of the robots and the duration of the phases in their computational cycles are determined by an adversarial scheduler. Three main types of schedulers are commonly considered:
 
 \begin{itemize}
   
  \item {\bf Asynchronous Scheduler (the {\it ASYNC} model):} The asynchronous scheduler is the most general model, and there is no shared notion of time. Under this scheduler, the robots are activated independently of one another, and the duration of each phase of the computational cycle is finite but unpredictable. The phases of different robots may overlap, and therefore, a robot may be observed by others while it is in motion. Robots do not possess motion-detection capability; in particular, a robot cannot determine whether another robot is currently moving or stationary. As a consequence of asynchrony, a robot may perform computations based on outdated observations.

  \item {\bf Semi-synchronous Scheduler (the {\it SSYNC} model):}  In the semi-synchronous scheduler, time is logically divided into non-overlapping global rounds. In each round, a subset of robots is activated simultaneously, and all activated robots complete their computational cycles within that round. The uncertainty arises from the selection of the active subset in each round. A robot is not visible to others while it is in motion.

  \item {\bf Fully synchronous Scheduler (the {\it FSYNC} model):}  The fully synchronous scheduler is a special case of the semi-synchronous scheduler, in which all the robots are activated in every round.

 \end{itemize}
  \item \textbf{Multiplicity Detection:} Robots may be endowed with some \textit{multiplicity detection capability}, which helps the robots to identify points occupied by multiple robots. If the robots have a \textit{weak multiplicity detection capability}, they can detect whether multiple robots occupy a point. They are unable to count the exact number of robots that make up the multiplicity. However, in this paper, we have assumed that the robots do not have any multiplicity detection capability. 
\end{itemize}
 

\noindent Motivated by the coordination problems in swarm robotics, this paper studies the \textit{min-sum uniform coverage} problem in both the \textbf{line-segment} and \textbf{circle} settings for a swarm of autonomous, anonymous, homogeneous, and identical mobile robots operating under the \textit{Look--Compute--Move} (LCM) model. The robots are oblivious, silent, and execute a deterministic distributed algorithm under a fair asynchronous scheduler. The objective is to reach a uniformly spaced configuration while minimizing the total distance traveled by all robots. We consider two deployment scenarios consistent with the general model described above: (i) the robots are initially deployed on a finite line segment and are restricted to move along the segment, and (ii) the robots are initially deployed on the circumference of a given input circle~$\mathcal{C}$ and are restricted to move only along its circumference. In the line-segment setting, the goal is to place the robots at equally spaced points along the segment, whereas in the circle setting, the robots are required to form a regular $n$-gon on the circumference of~$\mathcal{C}$, without any predefined orientation or distinguished starting position, while minimizing the total movement cost.\\

\noindent Earlier work (Bhattacharya et al.~\cite{bhattacharya2009optimal}) proposed centralized algorithms for this problem, where the movements of all the robots are determined in a centralized manner. In our work, we study the problem in a distributed setting, where the robots act autonomously and make independent decisions. To the best of our knowledge, no prior work has studied the min-sum objective in a deterministic distributed model for asynchronous robots.



.

\subsection{Problem Definition}

\noindent In this paper, we consider two deployment scenarios. First, we assume the robots are deployed on a line segment and study the \textit{min-sum uniform coverage} of the line segment. Next, we assume that the robots are deployed on the circumference of a circle. The robot movements are restricted on the circumference of the circle. In this scenario, we investigate the \textit{min-sum uniform coverage on a circle} problem. In both cases, the robots start from an arbitrary initial configuration such that they are located at distinct positions. The goal is to reach a configuration, where the positions of the robots are uniformly distributed over the deployment space while minimizing the total distance travelled.

\medskip

\noindent \textbf{Problem Definition 1 (\textit{Min-sum uniform coverage} on a line segment).}
Let us consider a set of $n$ autonomous robots $\mathcal R=\{r_1,r_2,\ldots,r_n\}$ deployed on a finite line segment $\mathcal L^*=[a,b]$ having non-zero length. $\mathcal{R}(t_0)=\{r_1(t_0), r_2(t_0), \ldots, r_n(t_0)\}$ be the initial configuration of robots in $\mathcal R$ on the line segment $\mathcal L^*$ at time $t_0$. Let $X^*=\{x_1^*,x_2^*,\ldots,x_n^*\}$ be a set of uniformly spaced points on $\mathcal L^*$ such that $x_i^* = a + \left(i-\tfrac{1}{2}\right)\frac{b-a}{n}, \quad 1 \le i \le n.$ The \textit{min-sum uniform coverage} problem on the line segment $\mathcal L^*$ requires the robots in $\mathcal R$ to coordinate their movements along $\mathcal L^*$ in such a way that, starting from the initial configuration $\mathcal R(t_0)$, the robots reach a configuration $\mathcal R(t)$ in which they occupy a unique point in $X^*$ and the sum of distances traversed by all the robots during this process is minimized. Formally, the final configuration satisfies
$
\mathcal R(t) = X^*
$ and $\sum_{i=1}^{n}
\left| r_i(t_0) - f\bigl(r_i\bigr) \right|$ is minimized, where $f:\mathcal R \rightarrow X^*$ is a bijection assigning each robot to a distinct target point in $X^*$, i.e.,  $f\bigl(r_i\bigr)\in X^*$, $f(r_i)\neq f(r_j)$ for $r_i\neq r_j$, and  $|x-y|$ denotes the Euclidean distance between two points $x$ and $y$ on the line.



\medskip

\noindent \textbf{Problem Definition 2 (\textit{Min-sum uniform coverage} on a circle).} Let us consider a set of $n$ autonomous robots $\mathcal R=\{r_1,r_2,\ldots,r_n\}$ deployed on the boundary (circumference) of a circle $C$ having finite positive radius. $\mathcal{R}(t_0)=\{r_1(t_0), r_2(t_0), \ldots, r_n(t_0)\}$ be the initial robot configuration of robots in $\mathcal R$ on the boundary of $C$ at time $t_0$. Let $\mathcal{P}^*=\{p^*_1,p^*_2,\ldots,p^*_n\}$ be a regular $n$-gon whose each vertex $p_i^*$ lies on the boundary of $C$. The \textit{min-sum uniform coverage} problem on a circle requires the robots in $\mathcal R$ to coordinate their movements along the circumference of $C$ in such a way that, starting from the initial configuration $\mathcal R(t_0)$ on $C$, the robots reach a configuration $\mathcal R(t)$ in which they occupy a unique point in $\mathcal P^*$ and the sum of distances traversed by all the robots during this process is minimized. Formally, the final configuration satisfies
$\mathcal R(t) = \mathcal P^*$ and $\sum_{i=1}^{n}
d_{\text{arc}}\big(r_i(t_0), f(r_i)\big)$ is minimized where  $f:\mathcal R \rightarrow \mathcal P^*$ is a bijection assigning each robot to a distinct point in $\mathcal P^*$ i.e., $f\bigl(r_i\bigr)\in \mathcal P^*$,  $f(r_i)\neq f(r_j)$ for $r_i\neq r_j$, and $d_{\text{arc}}(x,y)$ denotes the arc distance between two points $x$ and $y$ on the circumference of a circle. The distance is measured along the circumference of $\mathcal C$.

\subsection{Related Works}

\noindent Geometric formation problems for autonomous mobile robots have been extensively
studied in the distributed computing literature. These problems focus on
rearranging a set of anonymous robots from an arbitrary initial configuration
into a desired geometric pattern. Below, we briefly review the most relevant
results related to circle formation and its variants.

\begin{itemize}

\item {\bf Geometric Formation Problems:}
The general framework for geometric formation by autonomous mobile robots was
introduced by Suzuki and Yamashita~\cite{Suzuki1999}, laying the foundation for
subsequent work in distributed robot coordination. Later, Flocchini et
al.~\cite{DBLP:conf/isaac/FlocchiniPSW99} extended this framework to asynchronous
and oblivious robots, establishing several fundamental feasibility and
impossibility results. These works form the basis of most later studies on
distributed geometric pattern formation.

\item {\bf The Circle Formation Problem:}
Circle formation is one of the most basic geometric formation problems, where
robots are required to position themselves on the boundary of a common circle,
without enforcing uniform spacing. Early work by Sugihara and
Suzuki~\cite{128452} proposed a heuristic distributed solution; however, the
resulting configuration only approximates a circle, typically converging to a
Reuleaux triangle. A Reuleaux triangle is a curve of constant width formed by circular arcs centered at the vertices of an equilateral triangle. Debest~\cite{DBLP:journals/cacm/Debest95} studied the circle
formation problem in the context of self-stabilising systems. In a self-stabilising system, the robots must converge to a correct formation from any arbitrary initial configuration, regardless of faults, and remain correct thereafter without relying on initial assumptions. Exact circle formation has been studied under stronger synchronization
assumptions. Défago and Konagaya~\cite{DBLP:conf/pomc/DefagoK02} proposed an
\textit{SSYNC} algorithm that converges to a uniform circle, while Défago and
Souissi~\cite{DBLP:journals/tcs/DefagoS08} proposed an \textit{SSYNC} algorithm
for non-uniform circle formation assuming chirality and unobstructed
visibility. Flocchini et al.~\cite{DBLP:series/synthesis/2012Flocchini} were the
first to study circle formation under the fully asynchronous \textit{ASYNC}
model with global visibility, without any additional assumptions. Subsequent works
addressed exact circle formation under weaker assumptions, including models
without agreement on coordinate systems
\cite{DBLP:conf/opodis/FlocchiniPSV14} and solutions for all cases except
$n=4$~\cite{DBLP:journals/dc/FlocchiniPSV17}, with the remaining case solved
in~\cite{vigl2016}. Circle formation has also been investigated for fat robots in the limited visibility model and discrete environments
\cite{DBLP:conf/icdcit/DuttaCDM12,DBLP:journals/aghcs/AdhikaryKS21}. In \cite{DBLP:journals/aghcs/AdhikaryKS21}, a distributed algorithm for the circle formation problem under the infinite grid environment by asynchronous mobile opaque robots has been presented, where the robots have agreements on one of the coordinate axes.

\item {\bf The Uniform Circle Formation Problem:}
Uniform circle formation (UCF) is a stricter variant in which robots must occupy
equally spaced positions on the circumference of a circle. Early solutions were
developed under synchronous or semi-synchronous models. Flocchini et
al.~\cite{DBLP:journals/dc/FlocchiniPSV17} solved the UCF problem under the
\textit{ASYNC} model for configurations with at least five robots, and the
special case of four robots was later resolved by Mamino and
Viglietta~\cite{vigl2016}. To overcome impossibility results in the asynchronous
settings, several works introduced persistent visible lights. Feletti et
al.~\cite{feletti2018uniform,app13137991} studied the uniform circle formation problem with
lights, analyzing time--color tradeoffs. More recently, Feletti et
al.~\cite{feletti_et_al:LIPIcs.DISC.2024.46} proposed an asymptotically optimal
solution achieving $O(1)$ time and $O(1)$ colors under the \textit{ASYNC}
scheduler. Extensions to obstructed visibility and collision avoidance were
presented in~\cite{10579120}. In \cite{DBLP:conf/icdcn/MondalC18, DBLP:conf/icdcit/MondalC20}, the problem is solved for fat robots. Specifically, the authors of~\cite{DBLP:conf/icdcn/MondalC18} consider transparent fat robots that agree on one axis of their coordinate systems and operate under the \textit{SSYNC} model. The algorithm in~\cite{DBLP:conf/icdcit/MondalC20} solves the uniform circle formation problem for a swarm of fat, opaque, and oriented robots operating under the \textit{ASYNC} model.  

\item {\bf Distance-Optimal and Constrained Variants:}
Most existing work focuses on feasibility and symmetry breaking rather than
optimizing movement costs. Bhagat et al.~\cite{bhagat2018optimum} studied a
constrained circle formation problem with the objective of minimizing the
maximum distance traveled by any robot, using persistent lights under the
\textit{ASYNC} model. Centralized formulations for distance-optimal uniform
circle formation, including min-sum objectives, were studied in
\cite{bhattacharya2009optimal}, but these approaches do not apply to fully distributed systems of oblivious robots.
\end{itemize}

\begin{table}[ht]
\scriptsize
\renewcommand{\arraystretch}{1.7}
\setlength{\tabcolsep}{2.5pt}
\begin{tabular}{|c|c|c|c|}
\hline
\textbf{Problem} & \textbf{Light / Agreement} & \textbf{Optimization Criterion} & \textbf{Obstructed Visibility} \\
\hline
Circle Formation \cite{DBLP:series/synthesis/2012Flocchini} 
& None & No & No \\
\hline
Uniform Circle Formation \cite{feletti2018uniform}  
& Persistent Lights & No & No \\
\hline
Uniform Circle Formation \cite{app13137991}  
& $O(1)$ Persistent Lights & No & No \\
\hline
Uniform Circle Formation \cite{10579120} 
& Persistent Lights & No & Yes \\
\hline
Constrained Circle Formation \cite{bhagat2018optimum} 
& 2- Persistent Lights & Minimize maximum distance & No \\
\hline
\textbf{Constrained Uniform Circle Formation (This paper)} 
& \textbf{None} & \textbf{Minimize total distance} & No \\
\hline
\end{tabular}

\caption{Summary of related work on circle formation and uniform circle formation}
\label{tab:earlier-work}
\end{table}
\noindent The results are summarized in \textbf{Table} \ref{tab:earlier-work}.

\subsection{Motivation}

\noindent Uniform coverage of autonomous robots is a fundamental coordination task in distributed robotic systems, with applications ranging from environmental monitoring and area coverage to sensor deployment and formation control. In many real-world scenarios, robots are constrained to move along simple geometric structures, such as lines (e.g., roads, pipelines, or corridors) or closed curves (e.g., boundaries, circular tracks, or ring-shaped environments). Placing robots evenly along these structures ensures balanced coverage, reduces overlap in sensing, and improves robustness.

\noindent Among the various formulations of uniform coverage, the \textit{min-sum uniform coverage} problem is particularly relevant in energy-aware and resource-constrained settings. By minimizing the total distance traveled by all the robots, this objective reflects overall energy usage, total movement cost, and long-term mechanical effects on the robots, which are important factors in in large-scale and long-term robotic deployments.


\noindent While the problem admits efficient solutions in centralized settings~\cite{bhattacharya2009optimal,tan2010new}, designing algorithms for autonomous robots with limited capabilities remains challenging. In the classical swarm robotics model, the robots are typically anonymous, oblivious, unable to communicate explicitly, and operate asynchronously. Even in one-dimensional environments, the combined effects of geometric constraints, symmetry, and distributed decision-making can lead to subtle algorithmic and correctness issues.

\noindent The line segment and the circle are two basic but fundamentally different settings for studying the uniform coverage problem. A line segment introduces boundary effects that partially break symmetry and impose ordering constraints, influencing both feasibility and optimality. In contrast, a circle has no endpoints and is highly symmetric, making symmetry breaking and destination assignment significantly more challenging under limited sensing and coordination. Studying the \textit{min-sum uniform coverage} problem in both settings therefore provides a common framework for understanding how geometry and topology affect optimal distributed robot coordination.

\noindent Beyond their theoretical interest, these problems serve as foundational building blocks for more complex formation and coverage tasks in higher dimensions. Insights derived from the line and circle settings can be leveraged to design scalable and energy-efficient algorithms for robots constrained to curves, graphs, and manifolds. Consequently, the study of \textit{min-sum uniform coverage} on a line segment and on a circle is not only theoretically appealing but also practically significant for the development of efficient autonomous robotic systems.

\noindent In the swarm robotics literature, uniform coverage problems on a line and on a circle are fundamental instances of \emph{Geometric Pattern Formation} problems~\cite{Suzuki1999,DBLP:conf/isaac/FlocchiniPSW99}. Extending these classical formation problems to include a min-sum objective provides a new perspective on geometric optimization in distributed and resource-constrained robotic environments.



\subsection{Our Contributions}

\noindent This paper studies the \textit{min-sum uniform coverage} problem for autonomous mobile robots under a classical distributed robotics model. The main contributions are summarized as follows.

\begin{itemize}
    \item 
    We formalize the \textit{min-sum uniform coverage} problem on a line segment and on a circle for anonymous, oblivious, silent, and asynchronous robots with no multiplicity detection or explicit communication. We focus on the feasibility of min-sum uniform coverage with a a global optimality objective that minimizes the total distance traveled by all robots.
    

    \item 
    We derive structural properties of min-sum optimal assignments in one-dimensional geometric settings. For the line segment case, we show how boundary constraints induce a unique ordering of robots in any optimal solution. For the circle, we characterize all possible min-sum optimal assignments and relate their existence to the symmetry of the initial configuration.

    \item 
   We establish that, on a circle, there exist initial configurations for which the \textit{min-sum uniform coverage} problem admits no deterministic solution. These impossibility results arise from certain symmetric configurations.
   

    \item 
    For every solvable initial configuration, we present deterministic distributed algorithms that allow the robots to reach a uniform configuration while achieving min-sum optimality. Our proposed algorithms work under asynchrnous scheduler. We formally prove that the proposed algorithms ensure collision freedom, monotonically decrease the global movement cost, and guarantee finite-time convergence under a fully asynchronous scheduler.

\end{itemize}

\noindent These results demonstrate that strong global optimality guarantees can be achieved in geometric pattern formation problems, even under severely constrained distributed robot models, and provide a basis for energy-efficient coordination in one-dimensional environments. To the best of our knowledge, this paper presents the first distributed study of the min-sum objective under an asynchronous robot model.

\subsection{Technical Challenges}
\noindent Designing optimal algorithms for the \textit{min-sum uniform coverage} problem under autonomous robotic settings presents several fundamental technical challenges, even in one-dimensional environments.
\begin{itemize}
    \item The min-sum objective implicitly defines an optimal matching between the initial robot positions and the target locations. In a distributed setting particularly for symmetric configurations, it may not be possible to deterministically select one optimal assignment among multiple optimal assignments.
\item  Both the line segment and the circle allow highly symmetric robot configurations. On a line segment, symmetry appears as mirror images about the midpoint, while on a circle, symmetry also arises from arbitrary rotations. In such cases, the robots may perceive identical views and consequently make the same decisions, which can hinder deterministic progress and complicate both correctness and convergence. Resolving these symmetries without increasing the total movement cost is therefore a key technical challenge.
\item  Minimizing the sum of traveled distances requires not only computing an optimal final assignment but also ensuring that intermediate movements do not introduce unnecessary detours. Under asynchronous scheduling and non-rigid motion, the robots may experience interruptions or delays during movement, and poorly designed motion strategies can result in a total cost that exceeds the optimal value.
\end{itemize}

\subsection{Road Map}

\noindent This paper is organized as follows. \textbf{Section}~\ref{Notations and Definitions} introduces the basic notation, definitions, and concepts used throughout the paper, including the robot model, distance measures, symmetry notions, and configuration views. It also presents a classification of initial robot configurations based on their
symmetry properties and identifies the classes for which the min-sum objective is deterministically unsolvable. In \textbf{Section}~\ref{Line Segment}, we investigate the min-sum uniform coverage problem on a finite line segment. This simpler one-dimensional setting allows us to isolate and analyze the impact of the min-sum constraint. We characterize the optimal target positions, propose a deterministic distributed and collision-free algorithm, and establish its correctness and finite-time convergence. The structural insights and techniques developed in this section serve as a foundation for the circle case. \textbf{Section}~\ref{Oncircle} addresses the main problem of min-sum uniform coverage on a circle, where the robots are initially deployed on the circumference of a given circle. We analyze the structure of optimal assignments, study the role of symmetry in determining uniqueness and feasibility, and classify all initial configurations accordingly. For all solvable classes, we present deterministic distributed algorithms achieving min sum optimality, while proving impossibility results in \textbf{Subsection}~4.2.1 for the remaining configurations. Finally, \textbf{Section}~\ref{Conclusion and Future Work} concludes the paper with a summary of the main contributions and an outline of possible directions for future work, such as weaker visibility assumptions, alternative optimization criteria, and extensions to higher-dimensional settings.

\section{Notations and Definitions}
\label{Notations and Definitions}

 \begin{itemize}
 \item  Let $\mathcal R = \{r_1,r_2,\cdots,r_n\}$ be a set of $n$ autonomous, anonymous, homogeneous, identical robots deployed in the Euclidean plane. The robots are punctiform, i.e., the robots are modeled as points with no physical extent. We have assumed that the robots are deployed at distinct locations in the initial configuration. They operate in {\it Look-Compute-Move} (LCM) cycle under the \textit{ASYNC} model. Each robot is activated infinitely often by the scheduler, and upon each activation, it completes an entire LCM cycle.
 
 \item Let $\mathcal C$ denote the given circle as an input with center $\mathcal O$, and $\mathcal{C}_{out}$ denotes the circumference of $\mathcal C$. The robots are assumed to be deployed on $\mathcal{C}_{out}$. Let $r_i(t)$ denote the position of $r_i$ at time $t$. $\mathcal R(t)$ denotes the configuration of the robots on the circumference $\mathcal{C}_{out}$ of $\mathcal{C}$ at time $t$, and is defined to be the multi-set of positions occupied by the robots in $\mathcal R$, i.e., $\mathcal R(t) = \{r_1(t),r_2(t),\cdots,r_n(t)\}$. We denote by $\widetilde{\mathcal R}(t)$ the set of all such configurations at time $t$. Note that $\mathcal{R}(t_0)$ denotes the initial configuration of the robots deployed on $\mathcal{C}_{out}$ of  $\mathcal{C}$ at time $t_0$, and at time $t_0$, all robots in $\mathcal{R}(t_0)$ are stationary and occupy distinct locations.

\begin{figure*}[!ht]
    \centering
    \begin{subfigure}{0.48\textwidth}
        \includegraphics[width=\linewidth]{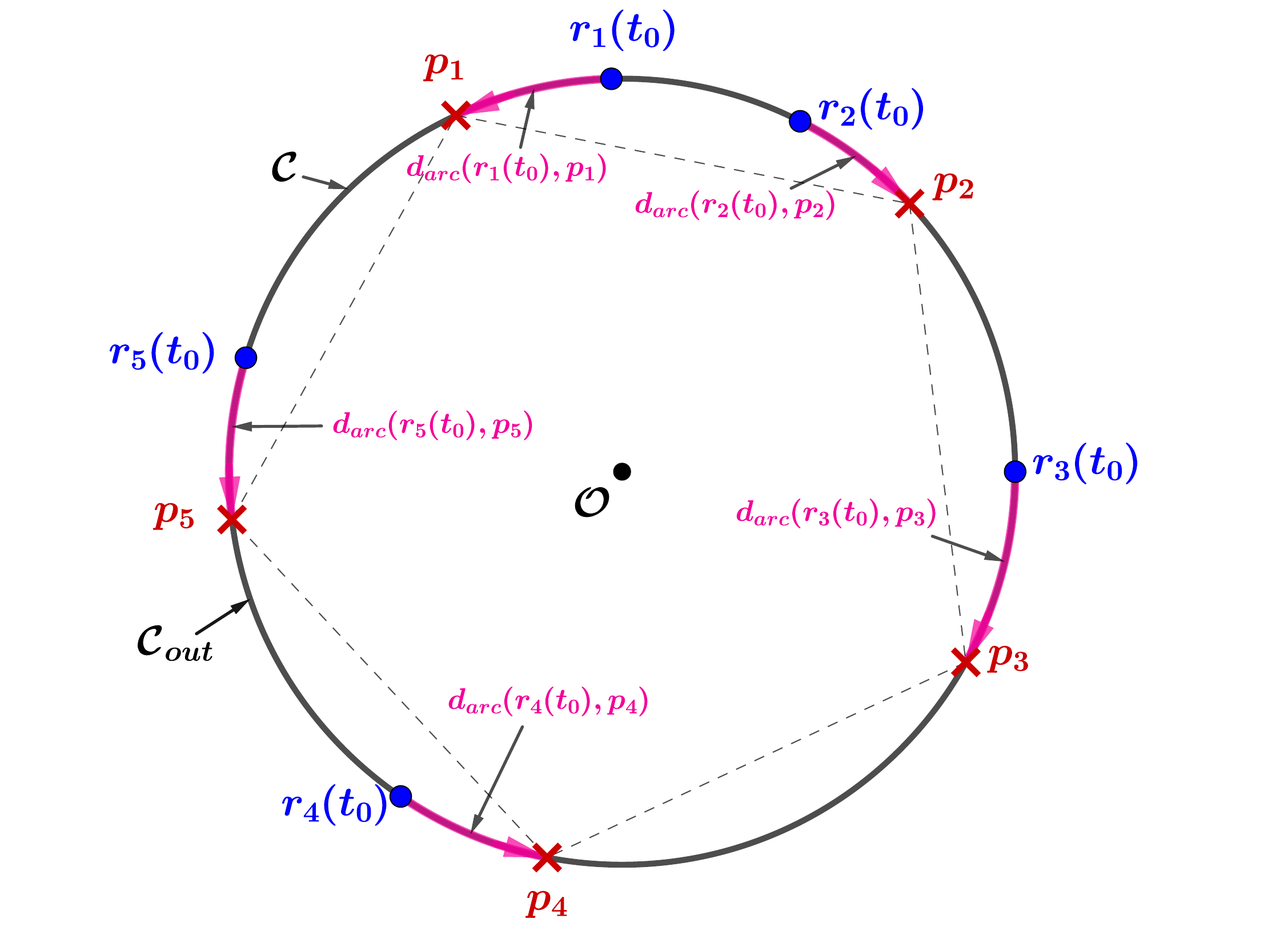}
        \caption*{(A)}
    \end{subfigure}
    \begin{subfigure}{0.48\textwidth}
        \includegraphics[width=\linewidth]{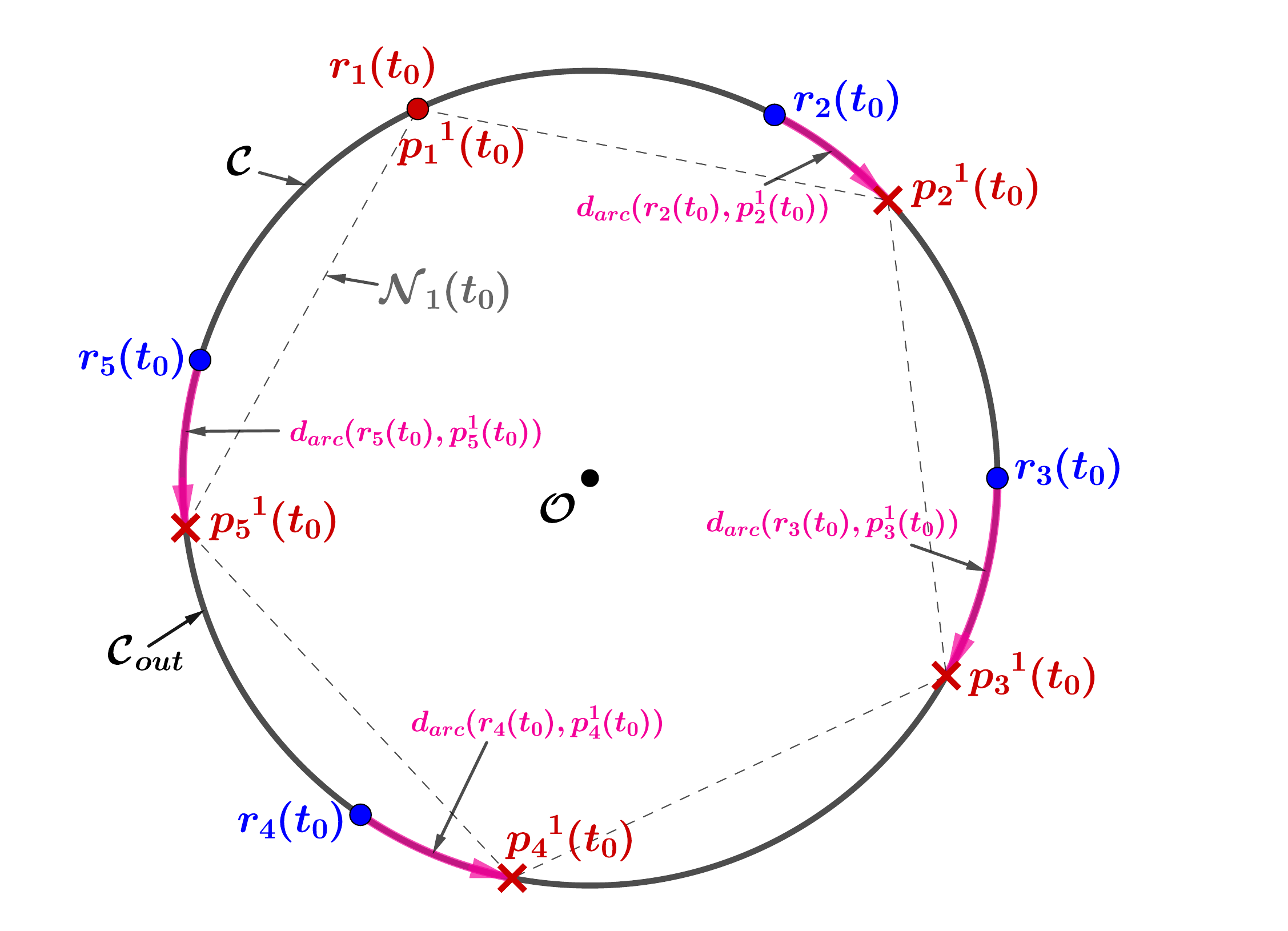}
        \caption*{(B)}
    \end{subfigure}
    \caption{\textit{An illustration of two possible configurations of the regular $n$-gon constructed on $\mathcal{C}_{out}$. \textbf{(A)} The robots  $r_1(t_0),\dots,r_5(t_0)$ (blue points) are placed on $\mathcal{C}_{out}$, and the corresponding destination points are $p_1,\dots,p_5$ (red crosses). None of the vertices of the regular $n$-gon lie on the initial positions of the robots. For each robot $r_i(t_0)$, the arc distance to its assigned destination point is denoted by $d_{arc}\Big(r_i(t_0), p_i\Big)$. \textbf{(B)} The robots  $r_1(t_0)(\textit{red point}),\dots,r_5(t_0)$ (blue points) are placed on $\mathcal{C}_{out}$, and the corresponding destination points $p^{1}_1(t_0),\dots,p^{1}_5(t_0)$ (red crosses) form the vertex set $\mathcal{P}^{1}(t_0)$ of $\mathcal{N}_{1}(t_0)$ by fixing robot $r_1(t_0)$ as a vertex. For each robot $r_i(t_0)$, the arc distance to its assigned destination point is denoted by $d_{\text{arc}}(r_i(t_0), p_i^{1}(t_0))$.}}
    \label{sum}
\end{figure*}

\item We consider the scenario in which the robots are deployed on the circumference $\mathcal{C}_{out}$ of a given circle $\mathcal{C}$ at distinct locations. Let $\widetilde{\mathcal C}$ denote the set of all such circles having finite positive radius and $\widetilde{\mathscr A}$ denote the set of all algorithms that can solve the {\it uniform circle formation} problem for any given circle $\mathcal C\in\widetilde{\mathcal C}$. Note that $\widetilde{\mathscr A}\neq\emptyset$ \cite{feletti2018uniform}. Suppose, we fix an algorithm $\mathcal{A}\in\widetilde{\mathscr A}$, a circle $\mathcal C$ and an initial robot configuration $\mathcal R(t_0)$ on $\mathcal C$. Let $p_i$ denote the final position of a robot $r_i \in \mathcal{R}$ on the circle $\mathcal{C}$ obtained after executing algorithm $\mathcal{A}$ on $\mathcal{C}$, starting from the initial configuration $\mathcal{R}(t_0)$. In other words, the set of points, say $\mathcal P=\{p_1, p_2, \ldots, p_n\}$ forms the vertices (lying on $\mathcal{C}_{out}$) of a regular $n$-gon. Thus, the problem definition reduces to show that $\exists$ $t^*\ge t_0$ such that $\mathcal R(t^*)=\{p_1,p_2,\ldots,p_n\}$.

\begin{figure*}[!ht]
    \centering
    \begin{subfigure}{0.48\textwidth}
        \includegraphics[width=\linewidth]{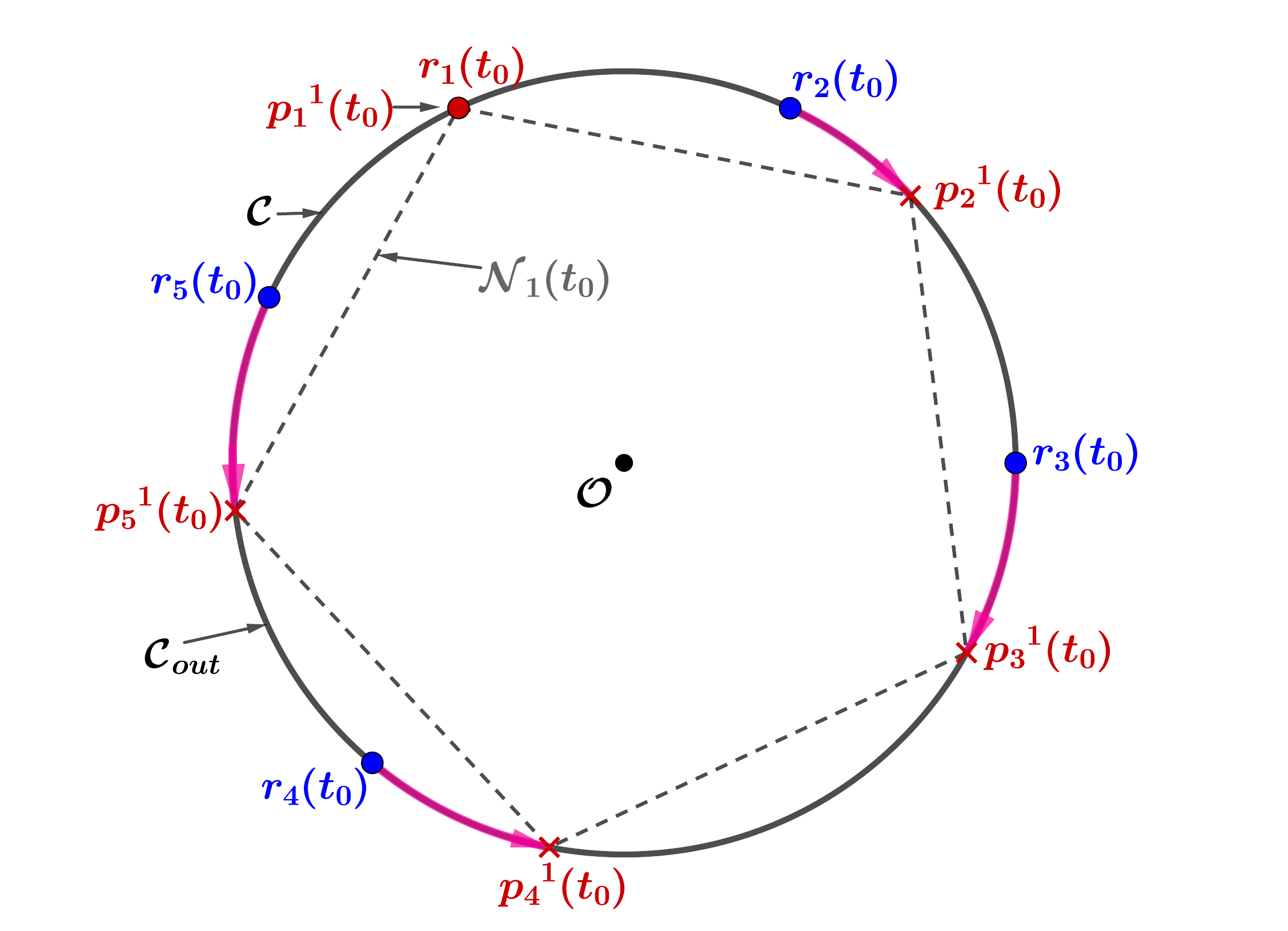}
        \caption*{(A)}
    \end{subfigure}
    \begin{subfigure}{0.48\textwidth}
        \includegraphics[width=\linewidth]{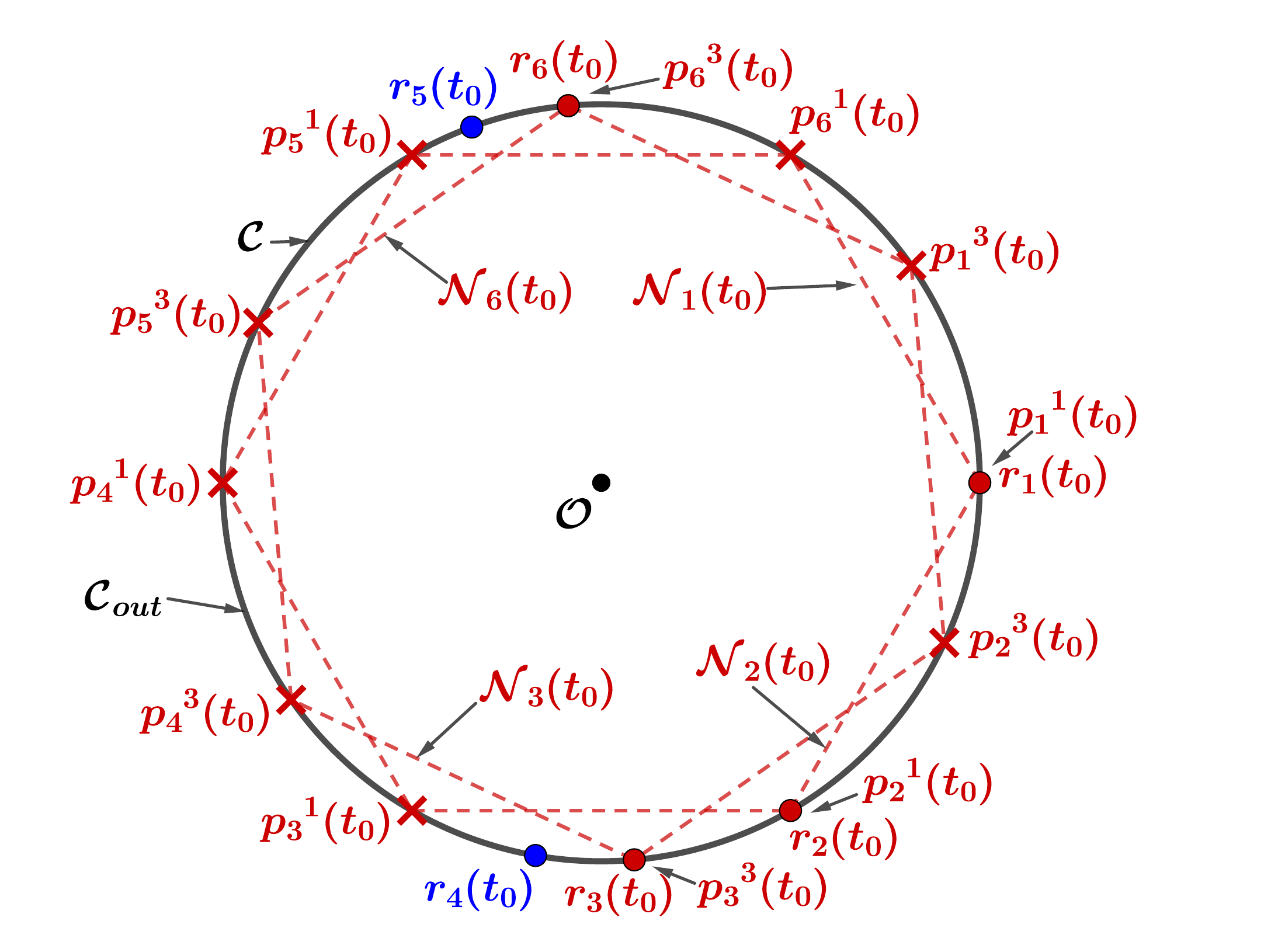}
        \caption*{(B)}
    \end{subfigure}
    \caption{\textit{An illustration of unique and multiple optimal assignments for six robots $r_1(t_0), \ldots, r_6(t_0)$ placed on a circle $\mathcal{C}$, with destination points determined according to \textbf{Result}~\ref{r-2}. \textbf{(A)} \emph{Unique optimal assignment:} the \textit{extremal} set $\mathcal{E}'(t_0)=\{r_1(t_0)\}$ is shown in red (dots), and the fixed destination points $\mathcal P^1(t_0)=\{p_1^1(t_0),p_2^1(t_0),\ldots,p_5^1(t_0)\}$ (red crosses) of the regular $n$-gon $\mathcal{N}_1(t_0)$ are uniquely determined by $r_1(t_0)$. The remaining robots (blue dots) are assigned consecutively to the remaining vertices. \textbf{(B)} \emph{Multiple optimal assignments:} the \textit{extremal} set is $\mathcal{E}'(t_0)=\{r_1(t_0), r_2(t_0), r_3(t_0), r_6(t_0)\}$ (red dots). In this case, multiple regular $n$-gons are possible: the destination points (red crosses) of $\mathcal{N}_1(t_0)$ or $\mathcal{N}_2(t_0)$ are determined by $\{r_1(t_0), r_2(t_0)\}$, while those of $\mathcal{N}_3(t_0)$ or $\mathcal{N}_6(t_0)$ are determined by $\{r_3(t_0), r_6(t_0)\}$. The remaining robots are shown as blue dots.}}
    \label{Assignment}
\end{figure*}

\item Let $d_{arc}(x,y)$ denote the arc distance on the circle $\mathcal{C}$ connecting any two points $x,y \in \mathcal{C}_{out}$. For a robot $r_i \in \mathcal{R}$ with initial position $r_i(t_0) \in \mathcal{R}(t_0)$ on $\mathcal{C}_{out}$ and final position $p_i \in \mathcal{P}$, $d_{arc}\Big(r_i(t_0), p_i\Big)$ denotes the arc path traversed by $r_i$, from its initial position $r_i(t_0)\in \mathcal{R}(t_0)$ to its destination $p_i \in \mathcal P$ during the execution of algorithm $\mathcal{A}$ on $\mathcal{C}$ (See \textbf{Figure} \ref{sum}(A)). Similarly, for any two robots $r_i, r_j \in \mathcal{R}$, $d_{arc}\Big(r_i(t_0), r_j(t_0)\Big)$ denotes the arc distance between their initial positions on $\mathcal{C}_{out}$ of the circle $\mathcal{C}$. Let $\mathcal{D}\big(\mathcal R(t_0), \mathcal P;\mathcal A\big)$ denote  the total arc path distance traveled by all robots in $\mathcal R(t_0)$ to reach their respective destination points in $\mathcal P$ during the  execution of algorithm $\mathcal A$, and is given by

\begin{equation}
\mathcal{D}\big(\mathcal R(t_0), \mathcal P;\mathcal A\big) =
\sum_{k=1}^{n}
d_{arc}\Big(r_k(t_0),\, p_k\Big),
\quad
\forall\, r_k \in \mathcal R(t_0),\;
p_k \in \mathcal P.
\label{A-1}
\end{equation}

\item We say that a {\it uniform circle formation} algorithm $\bar{\mathcal A}\in\widetilde{\mathscr A}$ solves the  \textit{min-sum uniform coverage} on a circle problem over  the circles in $\widetilde{C}$ if $\mathcal{D}\big(\mathcal R(t_0), \mathcal P;\widetilde A\big) =min\{\mathcal{D}\big(\mathcal R(t_0), \mathcal P;\mathcal A\big):\mathcal A\in\widetilde{\mathscr A}\}$, $\forall \mathcal C\in\widetilde{C},\mathcal R(t_0)\in\widetilde{\mathcal R}(t_0)$. If there is no ambiguity, for a \textit{min-sum uniform coverage} on a circle algorithm $\bar{\mathcal A}$, we denote $\mathcal{D}\big(\mathcal R(t_0), \mathcal P;\widetilde A\big)$ by $\mathcal D^*$. Let $\widetilde{\mathscr A^*}$ denote the set of all \textit{min-sum uniform coverage} on a circle algorithms $\bar{\mathcal A}$ over the circles $\widetilde{C}$. In \cite{bhattacharya2009optimal} and \cite{chen2015optimal}, two different centralized deterministic algorithms were proposed to solve the \textit{min-sum uniform coverage} on a circle problem over the circles in $\widetilde{\mathcal C}$. Both algorithms solve the problem under the \textit{FSYNC} model, where the robots are assumed to have unlimited visibility, agreements on both the coordinate axes, and {\it rigid} movements.

 \item Let $\mathcal{N}_{i}(t_0)$ denote the regular $n$-gon constructed on the circumference $\mathcal{C}_{out}$ of the circle $\mathcal C$ by fixing robot $r_i$ as one of its vertices at time $t_0$. The vertices of
$\mathcal{N}_{i}(t_0)$ (lying on $\mathcal{C}_{out}$) form the set
$\mathcal P^{i}(t_0)=\{p_1^{i}(t_0), p_2^{i}(t_0), \ldots, p_n^{i}(t_0)\}$, where $r_i(t_0)$ coincides with the vertex $p_i^{i}(t_0)$ (See \textbf{Figure} \ref{sum}(B)).

Suppose a robot $r_i(t) \in \mathcal{R}(t)$, located on $\mathcal{C}_{out}$, is kept fixed at its position, and the destination points of the remaining $n-1$ robots are assigned with respect to $r_i(t)$ (as described in \textbf{Result}~\ref{r-2} which is explained later in this paper (See page~\pageref{r-2})). If the total arc distance traveled by these $n-1$ robots is equal to $\mathcal{D}^*$, then $r_i(t)$ is called an \emph{extremal robot} at time $t$. The value of $\mathcal{D}^*$ is given by
\begin{equation}
\mathcal D^* =
\sum_{k=1}^{n} d_{\text{arc}}\Big(r_k(t_0),\, p_k^{i}(t_0)\Big),
\qquad
\forall\, p_k^{i}(t_0) \in \mathcal P^{i}(t_0) \land
\ \forall\, r_k \in \mathcal R(t_0).
\label{e-1}
\end{equation}
The set of all such \textit{extremal} robots in $\mathcal{R}(t)$ is referred to as the 
\emph{extremal set} and is denoted by $\mathcal{E}'(t)$.

\item Consider a circle $\mathcal C$, a robot configuration $\mathcal R(t)$ on $\mathcal{C}_{out}$ of  $\mathcal C$, a {\it regular $n$-gon} $\mathcal{N}_{i}(t)$ and an \textit{extremal robot} $r_i(t)\in \mathcal{R}(t)$.  An {\it assignment} of an \textit{extremal robot} $r_i(t)$ positions in  $\mathcal R(t)$ to the point set $\mathcal{P}^{i}(t)$ is a bijection $f_i(t):\mathcal R(t)\rightarrow \mathcal{P}^{i}(t)$. A point set $\mathcal{P}^{i}(t)$ on $\mathcal{C}_{out}$ is said to be the {\it destination point set} if there exists an  {\it assignment} $f$ from $\mathcal R(t)$ to $\mathcal{P}^{i}(t)$ such that the total arc distance traversed by all robots to reach their respective destination points is equal to $\mathcal{D}^*$. Such an {\it assignment} is called an  {\it optimal assignment} for $r_i(t)\in \mathcal{R}(t)$ for $\mathcal{P}^{i}(t)$ (See \textbf{Figure} \ref{Assignment}). An \textit{optimal assignment} always exists, as established in Lemma 2 of \cite{bhattacharya2009optimal}. Note that for a given robot configuration on the circle $\mathcal C$, there may exist more than one such \textit{extremal robot} position in $\mathcal{R}(t)$. To compute a \emph{destination point set} for $\mathcal{R}(t)$ on the circle $\mathcal C$, we first identify the set $\mathcal{E}'(t)$ using \textbf{Result}~\ref{r-2} (See page~\pageref{r-2}). If $|\mathcal{E}'(t)|=1$ and $r_i(t) \in \mathcal{E}'(t)$, then the {\it optimum assignment} is unique with respect to the robot $r_i(t)$, and consequently, $\mathcal{D}^*$ is unique. Since there may exist multiple assignments corresponding to different fixed robots $r_i$, we denote the total arc distance from the robot positions to their respective destination points at time $t$ by $\mathcal{D}_{r_i}(t)$. If $r_i$ is an \textit{extremal} robot, this value is denoted by $\mathcal{D}^{*}_{r_i}(t)$.

\end{itemize}

 \begin{figure}[!ht]
    \centering
    \includegraphics[width=0.5\textwidth]{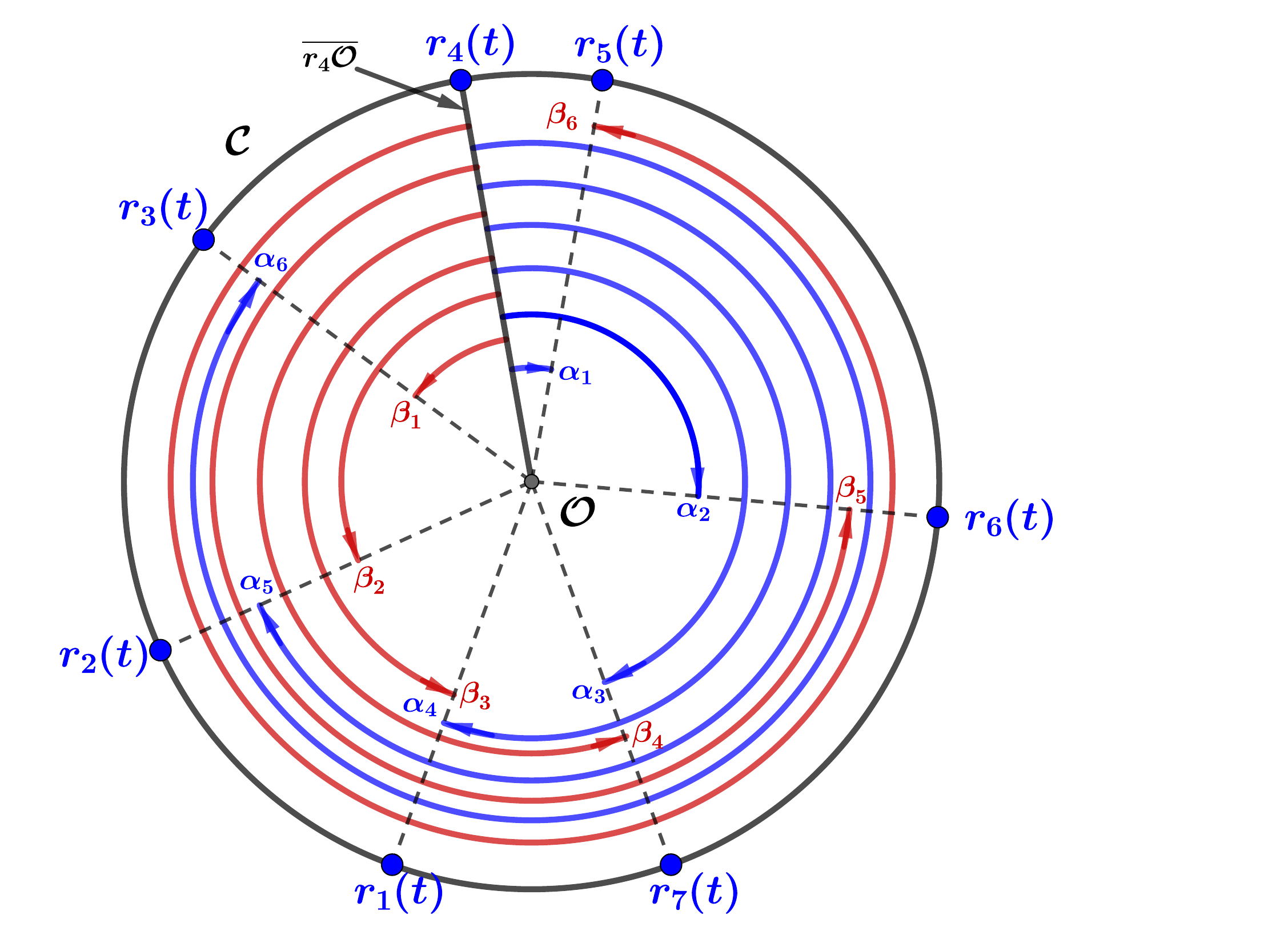}
    \caption{\textit{An illustration of configuration view of a robot. Let $\mathcal R=\{r_1,r_2,\ldots,r_7\}$. The clockwise view of $r_4$ is $\mathscr{V}^+(r_4) = (0,\alpha_1,\alpha_2,\ldots,\alpha_6)$, while the counterclockwise view is $\mathscr{V}^-(r_4) =(0,\beta_1,\beta_2,\ldots,\beta_6)$.  Using lexicographic ordering, the view of robot $r_4$ is  $\mathscr{V}(r_4)=\min\bigl(\mathscr{V}^+(r_4),\mathscr{V}^-(r_4)\bigr)= \mathscr{V}^+(r_4)$.}}
    \label{view}
\end{figure}

\noindent {\it \textbf{Configuration view}.} Let $\mathcal R(t)=\lbrace r_1(t),r_2(t),\ldots,r_n(t)\rbrace$ denote the set of robot positions in the order by which the robots would be encountered if the line segment $\overline{r_i(t)\mathcal{O}}$ is rotated by an angle of $2\pi$ in the clockwise direction. Let $\alpha_k$ denote the angle by which $\overline{r_i(t)\mathcal{O}}$ has been rotated when the $k^{th}$ robot position in $\mathcal R(t)$ is being encountered. Deﬁne the clockwise view of $r_i$ as $\mathscr{V}^+ (r_i ) = (\alpha_1, \alpha_2, \ldots, \alpha_n$) (See \textbf{Figure} \ref{view} ). Similarly, the counter-clockwise view of $r_i$ can be defined. The view of a robot $r_i$ is given by $\mathscr{V}(r_i)=\min (\mathscr{V}^+ (r_i ),\mathscr{V}^- (r_i ))$, where the minimum is computed by selecting the lexicographically smaller view between clockwise and counter-clockwise views. We have the following results.

\begin{result}\cite{das2022k}
\label{symm}
    A configuration $\widetilde{\mathcal R}(t)$ admits a line of symmetry if and only if there exist two robot positions $r_i,\; r_j\in \mathcal{R}(t)$, not necessarily distinct, such that $\mathscr{V}^+ (r_i )=\mathscr{V}^- (r_j )$.
\end{result} 
\begin{result}\cite{das2022k}
    A configuration $\widetilde{\mathcal R}(t)$ admits rotational symmetry if and only if there exist two distinct robot positions $r_i,\; r_j\in \mathcal{R}(t)$, such that $\mathscr{V}^+ (r_i )=\mathscr{V}^+ (r_j )$.
\end{result} 

\noindent We have the following observation.
\begin{observation}
\label{obs:asymmetric-ordering}
If the configuration is asymmetric, then the robots can be totally ordered according to the lexicographic ordering of their minimum views.
\end{observation}


 \noindent Note that if a configuration admits symmetry, then no algorithm can distinguish between a robot $r$ and its symmetric image. Consequently, a robot and its symmetric image(s) may decide to move simultaneously. A \textit{pending move} may exist if an algorithm permits at least two robots to move at the same time. Due to the asynchronous nature of the scheduler, it may happen that one of the robots allowed to move completes its entire Look-Compute-Move cycle while another robot does not perform the Move phase.
 \vspace{0.2cm}

 \noindent\textbf{Partitioning of Initial Configurations:}
    \label{class}
    Let $\widetilde{\mathcal R}(t_0)$ denote the set of all initial robot configurations. We classify $\widetilde{\mathcal R}(t_0)$ into six disjoint classes:
\begin{align*}
\mathcal{I}_1 &= 
\left\{\mathcal{R}(t_0)\in\widetilde{\mathcal R}(t_0)\; \middle|\;
\mathcal{R}(t_0)\text{ is asymmetric}\right\}, \\[0.4em]
\mathcal{I}_2 &= 
\left\{\mathcal{R}(t_0)\in\widetilde{\mathcal R}(t_0)\; \middle|\;
\begin{aligned}
&\mathcal{R}(t_0)\text{ is symmetric, admits exactly one LoS,} \\
&\text{and \textit{extremal} robot/robots on it}
\end{aligned}
\right\}, \\[0.4em]
\mathcal{I}_3 &= 
\left\{\mathcal{R}(t_0)\in\widetilde{\mathcal R}(t_0)\; \middle|\;
\begin{aligned}
&\mathcal{R}(t_0)\text{ is symmetric, admits a single LoS,} \\
&\text{and has no \textit{extremal} robot on it}
\end{aligned}
\right\}, \\[0.4em]
\mathcal{I}_4 &= 
\left\{\mathcal{R}(t_0)\in\widetilde{\mathcal R}(t_0)\; \middle|\;
\begin{aligned}
&\mathcal{R}(t_0)\text{ admits rotational symmetry} \\
&\text{of order } w\ge 2 \text{ with no line of symmetry}
\end{aligned}
\right\}, \\[0.4em]
\mathcal{I}_5 &= 
\left\{\mathcal{R}(t_0)\in\widetilde{\mathcal R}(t_0)\; \middle|\;
\begin{aligned}
&\mathcal{R}(t_0)\text{ is symmetric, admits multiple LoS,} \\
&\text{and has \textit{extremal} robots on at least one of them}
\end{aligned}
\right\}, \\[0.4em]
\mathcal{I}_6 &= 
\left\{\mathcal{R}(t_0)\in\widetilde{\mathcal R}(t_0)\; \middle|\;
\begin{aligned}
&\mathcal{R}(t_0)\text{ is symmetric, admits multiple LoS,} \\
&\text{and has no \textit{extremal} robot on any LoS}
\end{aligned}
\right\}.
\end{align*}

\begin{table}[!ht]
\centering
\setlength{\extrarowheight}{1.5pt}
\renewcommand{\arraystretch}{1.3}
\begin{tabular}{|c|p{10.5cm}|}
\hline
\textbf{Symbol} & \textbf{Description} \\
\hline
$\mathcal R=\{r_1,r_2,\ldots,r_n\}$ 
& Set of $n$ autonomous, anonymous, homogeneous, identical robots. \\
\hline
$r_i(t)$ 
& Position of robot $r_i$ at time $t$. \\
\hline
$\mathcal R(t)$ 
& Multiset of robot positions at time $t$ on $\mathcal{C}_{out}$. \\
\hline
$\widetilde{\mathcal R}(t)$ 
& Set of all robot configurations at time $t$. \\
\hline
$t_0$ 
& Initial time at which all robots are stationary. \\
\hline
$\mathcal C$ 
& Given circle with center $\mathcal O$. \\
\hline
$\mathcal{C}_{out}$ 
& Circumference (outer boundary) of circle $\mathcal C$. \\
\hline
$\widetilde{\mathcal C}$ 
& Set of all circles with finite positive radius. \\
\hline
$\widetilde{\mathscr A}$ 
& Set of all algorithms that solve the uniform circle formation problem. \\
\hline
$\mathcal A$ 
& A uniform circle formation algorithm, $\mathcal A \in \widetilde{\mathscr A}$. \\
\hline
$p_i$ 
& Destination point of robot $r_i$ on $\mathcal{C}_{out}$. \\
\hline
$\mathcal P=\{p_1,p_2,\ldots,p_n\}$ 
& Set of destination points forming a regular $n$-gon on $\mathcal{C}_{out}$. \\
\hline
$\mathcal N_i(t_0)$ 
& Regular $n$-gon constructed by fixing robot $r_i$ as a vertex at time $t_0$. \\
\hline
\end{tabular}
\caption{Table of notations used throughout the paper (Part I).}
\label{tab:notation-1}
\end{table}

\begin{table}[!ht]
\centering
\setlength{\extrarowheight}{1.5pt}
\renewcommand{\arraystretch}{1.3}
\begin{tabular}{|c|p{10.5cm}|}
\hline
\textbf{Symbol} & \textbf{Description} \\
\hline
$\mathcal P^{i}(t_0)$ 
& Vertex set of $\mathcal N_i(t_0)$ on $\mathcal{C}_{out}$. \\
\hline
$d_{\text{arc}}(x,y)$ 
& Arc distance between points $x$ and $y$ on $\mathcal{C}_{out}$. \\
\hline
$\mathcal D(\mathcal R(t_0),\mathcal P;\mathcal A)$ 
& Total sum of arc distance traveled by all robots under algorithm $\mathcal A$. \\
\hline
$\mathcal D^*$ 
& Minimum total sum of arc distance over all uniform circle formation algorithms. \\
\hline
$\widetilde{\mathscr A}^*$ 
& Set of all \textit{min-sum uniform coverage} on a circle algorithms. \\
\hline
$r_i$ (\textit{extremal}) 
& A robot whose fixed position yields total arc distance $\mathcal D^*$. \\
\hline
$\mathcal E'(t)$ 
& Set of all \textit{extremal} robots at time $t$. \\
\hline
$f_i(t)$ 
& Assignment (bijection) from robot positions to $\mathcal P^{i}(t)$. \\
\hline
$\mathcal D_{r_i}(t)$ 
& Total arc distance when robot $r_i$ is used as the fixed reference. \\
\hline
$\mathcal D^{*}_{r_i}(t)$ 
& Optimal total arc distance when $r_i$ is an \textit{extremal} robot. \\
\hline
$\mathscr V^+(r_i), \mathscr V^-(r_i)$  
& Clockwise and counterclockwise configuration views of robot $r_i$. \\
\hline
$\mathscr V(r_i)$ 
& View of a robot $r_i$, defined as $\min(\mathscr V^+(r_i),\mathscr V^-(r_i))$. \\
\hline
$\mathcal L$ 
& Single line of symmetry. \\
\hline
$w$ 
& Order of rotational symmetry of a configuration. \\
\hline
$\mathcal R_0,\mathcal R_1,\ldots,\mathcal R_{w-1}$ 
& Rotational equivalence classes induced by symmetry of order $w$. \\
\hline
\end{tabular}
\caption{Table of notations used throughout the paper (Part II).}
\label{tab:notation-2}
\end{table}

\noindent
We now introduce the notation and basic definitions used throughout the paper. \textbf{Table}~\ref{tab:notation-1} (See page~\pageref{tab:notation-1}) and \textbf{Table}~\ref{tab:notation-2} (See page~\pageref{tab:notation-2}) considers the symbols associated with the robot model, geometric setting, distance measures, assignments, symmetry, and configuration views, which are used consistently throughout the paper.
Unless stated otherwise, all robots operate under the \textit{ASYNC} model and move along the circumference of the given circle. 
This notation is used consistently in all subsequent sections.

\noindent In the initial configuration, we assume that the robots are located at distinct locations on the circumference of the given circle. The next lemma provides a characterization of the configurations from which the \textit{min-sum uniform coverage} on a circle problem cannot be solved, even in a stronger model than that assumed in this work.
\begin{lemma}\label{lemma1}
    Let $\mathcal A$ be an algorithm that solves the \textit{min-sum uniform coverage} problem on a circle for the robots under the \textit{FSYNC} scheduler. Suppose the robots are endowed with a weak-multiplicity detection capability. If there exists a configuration $\widetilde{\mathcal R}(t)$, for $t > 0$, containing a multiplicity, then $\mathcal A$ cannot solve the problem.
\end{lemma} 
\begin{proof}
 According to the hypothesis, we assume that there exists a configuration $\widetilde{\mathcal R}(t)$, $t > 0$, with a corresponding multiplicity. Under the \textit{FSYNC} model, all robots execute their \emph{Look} phase simultaneously. Consequently, co-located robots forming a multiplicity obtain identical views of the configuration and compute identical movements.
 If the adversary schedules all robots in the multiplicity to move by the same distance, the multiplicity present in configuration $\widetilde{\mathcal R}(t)$ is preserved in every subsequent configuration $\widetilde{\mathcal R}(t')$ for all $t' \geq t$. Consequently, the robots forming a multiplicity in $\widetilde{\mathcal R}(t)$ cannot reach different destination points on the circumference of the given circle $\mathcal C$. As a result, during the execution of algorithm $\mathcal A$, the robots never attain the final configuration in which all robots are positioned on a uniform circle while satisfying the objective constraint.
\end{proof}
\noindent Thus, \textbf{Lemma}~\ref{lemma1} shows that a collision-avoidance mechanism is necessary to ensure correct execution. In the next section, we discuss the min-sum uniform coverage problem on a line segment.

\section{Optimal Placement of Robots on Line Segment}
\label{Line Segment}

\noindent This section studies the min-sum problem on a line segment to isolate the effect of the min-sum constraint in a simpler geometric setting. Since robots are restricted to move along a single dimension, this setting allows us to analyze the structure of optimal and collision-free solutions more clearly. The observations obtained here, such as the role of \textit{extremal} robots and the behavior of optimal assignments, are later used directly in the algorithms for uniform circle formation, where robots are similarly constrained to move along the circumference of a circle.

\noindent
We assume that the robots are initially located on the line segment 
$[a,b]$ with $a<b$ and are allowed to move only along this segment, remaining within $[a,b]$ at all times. Let $r_1(t) < r_2(t)<\cdots < r_n(t)$ denote the robot positions ordered by increasing distance from the endpoint $a$. The ordering of the robots is denoted by $\mathcal{F}$. To simplify notation, we denote both a robot and its position by $r_i$. We will use the following result for optimal uniform coverage on a unit interval $[0,1]$.

\begin{result}\cite{bhattacharya2009optimal}
The optimal uniform coverage is obtained by moving point $A_i$ 
to position $\tfrac{2i-1}{2n}$, for $i = 1, 2, \dots, n$.
\label{r-1}
\end{result}

\noindent To extend this result from the unit interval $[0,1]$ to an arbitrary finite interval $[a,b]$, we consider the affine transformation
$T(z)=a+(b-a)z$, for $z\in[0,1]$. The mapping $T$ scales every distance
by the factor $(b-a)$ and then translates all points by $a$, while
preserving the ordering and relative spacing of points. Under this
transformation, every pairwise distance between the initial configuration $X$ and any feasible target configuration $Y$ is multiplied by $(b-a)$, and hence the corresponding movement cost satisfies $C(T(X),T(Y))=(b-a)\,C(X,Y)$. Since this scaling multiplies the cost of every feasible configuration by the same positive constant, the ordering of all solutions remains unchanged, and therefore the arg-min is preserved. Consequently, the optimal configuration obtained on $[0,1]$ remains optimal under $T$ on the interval $[a,b]$.

\begin{observation}
For a finite line segment $[a,b]$ containing $n$ robots, the optimal
target configuration that minimizes the total movement cost is obtained
by placing the robots at the evenly spaced positions
\[
x_i^{*}
=
a+\tfrac{2i-1}{2n}\cdot (b-a),
\qquad i = 1,2,\dots,n .
\]
    
\end{observation}

\noindent In this way, the solution for the unit interval can be applied directly to any finite segment, resulting in the optimal configuration for robot placement along $[a,b]$.
    


\begin{figure}[!ht]
    \centering
    
    \includegraphics[width=0.6\textwidth]{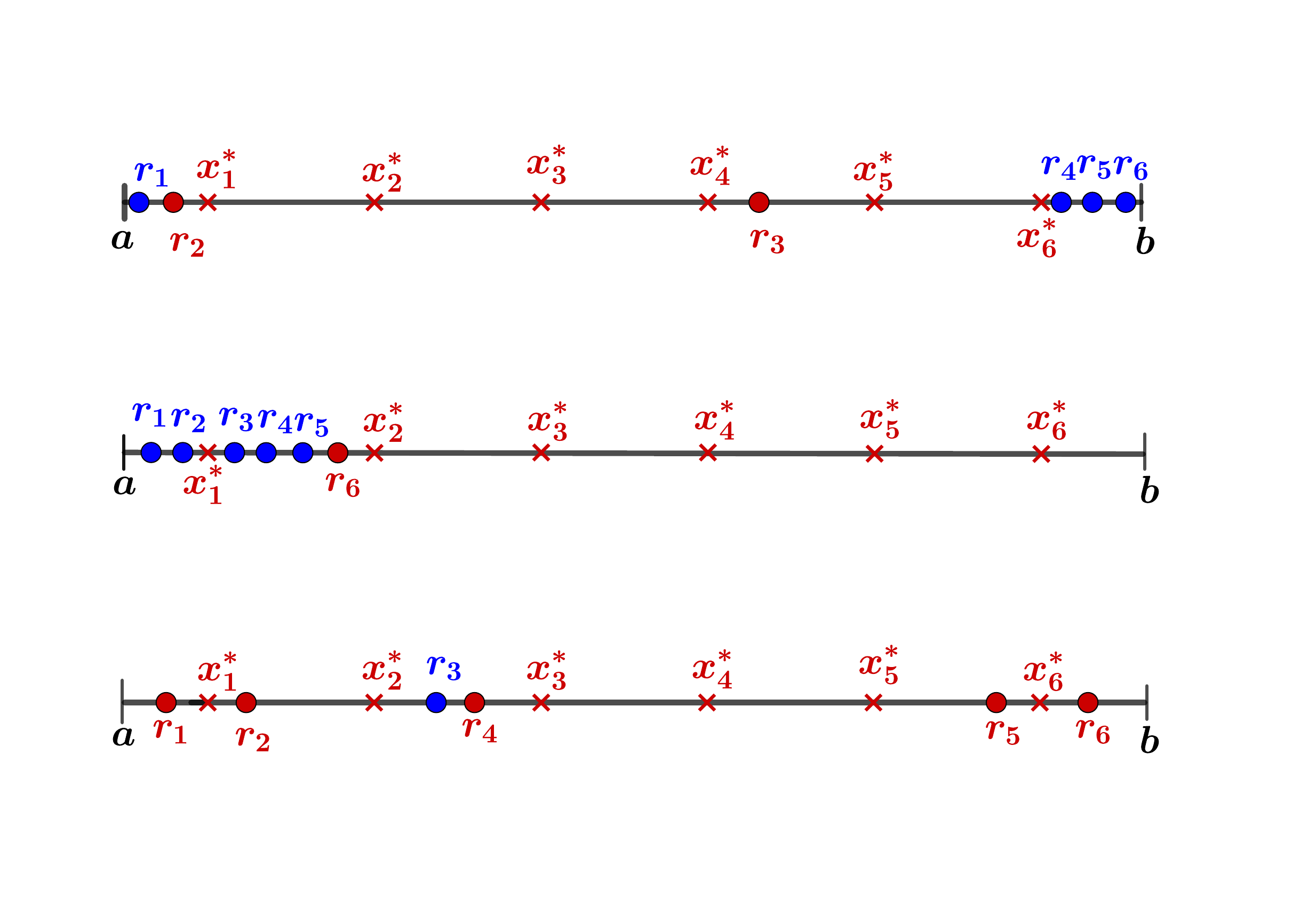}
    \caption{\textit{Illustration of candidate robots (red dots) on the line segment $[a,b]$ with unobstructed paths to their target positions $x_i^{*}$ (red crosses), and hence allowed to move. The remaining robots (blue dots), whose paths are obstructed, remain stationary until a free path becomes available.}}
    \label{Candidate}
\end{figure}

\subsection{Overview of the Algorithm \it{1dMinSumL()}}
\label{Overview-1}

\noindent This section presents a deterministic distributed algorithm designed to solve the optimal placement of robots on the line segment in finite time. The algorithm operates in two phases. First, the destination point set $x_i^*$ is computed, and each robot is assigned to a designated position within this set. A robot $r_i$ is identified as a \textit{candidate} robot if and only if it possesses an unobstructed path to its assigned destination $x_i^*$ (See \textbf{Figure} \ref{Candidate}). Movement is permitted exclusively for \textit{candidate} robots. It is necessary to establish an initial ordering of all robots, denoted by $\mathcal F$, to ensure systematic progress. At any time instant $t > t_0$, at least one \textit{candidate} robot will exist, thus ensuring that the algorithm will achieve optimal placement of robots on a line segment within a finite amount of time. If multiple candidate robots are present at any instant of time, the tie is broken by choosing the robot that attains the highest order under the ordering defined by 
$\mathcal F$. Thus, at any instant of time, exactly one candidate robot is permitted to move toward its destination, thereby ensuring a sequential execution of the algorithm. The pseudo-code of this procedure is given in \textbf{Algorithm} 1.

\begin{algorithm}
\caption{\textit{: 1dMinSumL()}}
\label{alg:1dMinSumL}
\begin{algorithmic}[1]
\Require Endpoints $a$ and $b$ of the line segment; initial robot positions
$r_1(t_0), r_2(t_0), \ldots, r_n(t_0)$
\Ensure Optimal assignment $r_i \rightarrow x_i^{*}$

\State Compute $x_i^{*} = a + \left(i - \tfrac{1}{2}\right)\dfrac{(b-a)}{n}$ for each $i$
\State Find the \textit{candidate} robot $r_i$

\For{$i = 1$ to $n$}
    \If{$r_i(t_0) = x_i^{*}$}
        \State Mark $r_i$ as terminated
    \Else
        \If{$r_i$ is the \textit{candidate} robot and has the maximum order with respect to $\mathcal F$}
            \State Move $r_i \rightarrow x_i^{*}$
        \Else
            \State $r_i$ waits
        \EndIf
    \EndIf
\EndFor

\end{algorithmic}
\end{algorithm}

\subsection{Correctness of the Algorithm \it{1dMinSumL()}}
\label{correct-1}

\begin{observation}
\label{obs-endpoint-invariant}
In the line-segment setting, the optimal min-sum assignment is independent of the choice of endpoint. Ordering the robots by increasing distance from endpoint $a$ or from endpoint $b$ results in the same robot--destination mapping.
\end{observation}
\noindent
This holds because reversing the robot ordering also reverses the ordering of destination points, without changing the final assignment.

\begin{lemma}
\label{l-1}
The algorithm \textit{1dMinSumL()} achieves collision-free optimal placement of robots for the min-sum problem on a line segment.
\end{lemma}
\begin{proof}
Let $\mathcal R(t) = \{r_1(t), r_2(t), \dots, r_n(t)\}$ denote the positions of the robots at time $t$, and let the target configuration be $X^{*}
= \lbrace x_1^{*}, x_2^{*}, \dots, x_n^{*}\}$. Now, the destination points are calculated by the given formula: 
\[
x_i^{*}
= a + \left(i - \tfrac{1}{2}\right)\frac{b-a}{n},
\quad i = 1,2,\dots,n.
\]

\noindent By definition, a robot $r_i$ is a \textit{candidate} robot if and only if the path from $r_i(t)$ to $x_i^{*}$ is unobstructed by any other robot, i.e., the line segment joining $r_i(t)$ and $x_i^{*}$ must not contain any other robot position except $r_i(t)$. More formally, 
\[
\big[r_i(t),x_i^{*}\big] \cap \mathcal R(t) = \{r_i(t)\}.
\]
Note that the assignment of the robots to their respective destinations preserves the ordering, i.e., if $r_i(t_0)<r_j(t_0)\implies x_i^{*}<x_j^{*}$, for all $i<j$. This ensures that no robot can cross one another during the execution of the algorithm. As a result, throughout the execution of the algorithm, the collision-free movement of the robots is ensured. For any time instant $t>t_0$, at least one candidate robot exists. The algorithm ensures that at any instant of time, exactly one candidate robot is allowed to move. Each candidate robot moves towards its destination position $x_i^{*}$, thereby reducing the distance to it. The distance is denoted by $d_i(t)=\text{dist}(r_i(t),x_i^{*})$, where $\text{dist}(r_i(t),x_i^{*})$ denote the distance between $r_i(t)$ and $x_i^{*}$. As a result, the system evolves monotonically toward the configuration $X^{*}$, and the algorithm guarantees collision-free convergence to the optimal placement of robots on the line segment.
\end{proof}

\begin{lemma}
\label{l-2}
Optimal placement of robots for the min-sum problem on a line segment is achieved in finite time.
\end{lemma}
\begin{proof}\normalfont
Let the initial robot configuration $\mathcal R(t_0)$ lie within the line segment $[a,b]$. Assume that the ordered set of robots defined using $\mathcal F$ is denoted as:
\[
\mathcal R_1(t_0) = \{r_1(t_0), r_2(t_0), \dots, r_n(t_0)\}.
\]
From \textbf{Result} \ref{r-1}, each robot $r_i \in \mathcal R_1(t_0)$ has an assigned destination $x_i^{*}$ that minimizes the total travel distance by all the robots. According to \textbf{Lemma} \ref{l-1}, the movement strategy for the \textit{candidate} robot ensures that no collisions occur.

\noindent We define the total \textit{min-sum} distance at time $t_0$ as
\[
S^{*}(t_0) = \sum_{j=1}^{n} d_j(t_0),
\]
where
\[
d_j(t_0) = \operatorname{dist}\big(r_j(t_0), x_j^{*}\big),
\qquad j = 1,2,\ldots,n,
\]
and $\operatorname{dist}\big(r_j(t_0), x_j^{*}\big)$ denotes the Euclidean distance between the two points $r_j(t_0)$ and $x_j^{*}$ on the line segment.

\noindent Suppose that at time $t_1 > t_0$, a robot $r_i(t_0) \in \mathcal R_1(t_0)$ 
moves toward its destination $x_i^{*}$ and reaches position $r_i(t_1)$. 
Since it moves closer to the destination, we have
\[
\text{dist}(r_i(t_0), x_i^{*}) > \text{dist}(r_i(t_1), x_i^{*}).
\]
Let $d = \text{dist}(r_i(t_0), r_i(t_1)) > 0$. Then the updated distance of robot $r_i$ at time $t_1$ is
\begin{equation}
    d_i(t_1) = d_i(t_0) - d.
    \label{e-2}
\end{equation}

\noindent For all robots except $r_i$ do not move, the distance remains unchanged for $n-1$ robots, that is $d_j(t_1) = d_j(t_0)$ for $j \neq i$ and $j\in\{1,2,\ldots,n\}$. Therefore, the total \textit{min-sum} distance at time $t_1$, defined by $S^{*}(t_1)$ is

\begin{align*}
    S^*(t_1) &= \sum_{j=1}^n d_j(t_1)\\
    &= \Big(\sum_{j=1}^n d_j(t_0)\Big) - d \\
    &= S^*(t_0) - d\ (\text{ By using equation (\ref{e-2})})
\end{align*}
which shows that
\[
S^*(t_1) < S^*(t_0).
\]

\noindent Hence, for all $t > t_0$, the sequence $\{S^*(t)\}$ is strictly decreasing and bounded below by zero. 
Consequently, it must converge to zero in finite time, corresponding to all robots reaching their respective destinations $x_i^{*}$. If more than one robot moves simultaneously, the decrease in $S^*(t)$ occurs even faster, further ensuring finite-time convergence. Therefore, the optimal placement of robots for the min-sum problem on a line segment is achieved within finite time.
\end{proof}
\noindent In the next section, we discuss the min-sum uniform coverage problem on the input circle, where the movement of the robots are restricted only along the arcs of the circle.
\section{Optimal Robot Placement on the Input Circle}
\label{Oncircle}

\noindent We now consider the case in which the robots are initially deployed on the circumference of the input circle. As in the line-segment case, the objective in this scenario is to minimize the total distance traveled by all the robots while forming a uniform configuration. The key difference in this setting is that the robot's motion is restricted to the circumference of the circle. Consequently, all the robot trajectories and the distance measurements are taken along the circular boundary.

\begin{figure}[!ht]
    \vspace{-1.5cm}
    
    \hspace{-2.8cm}
    \includegraphics[width=1.3\textwidth]{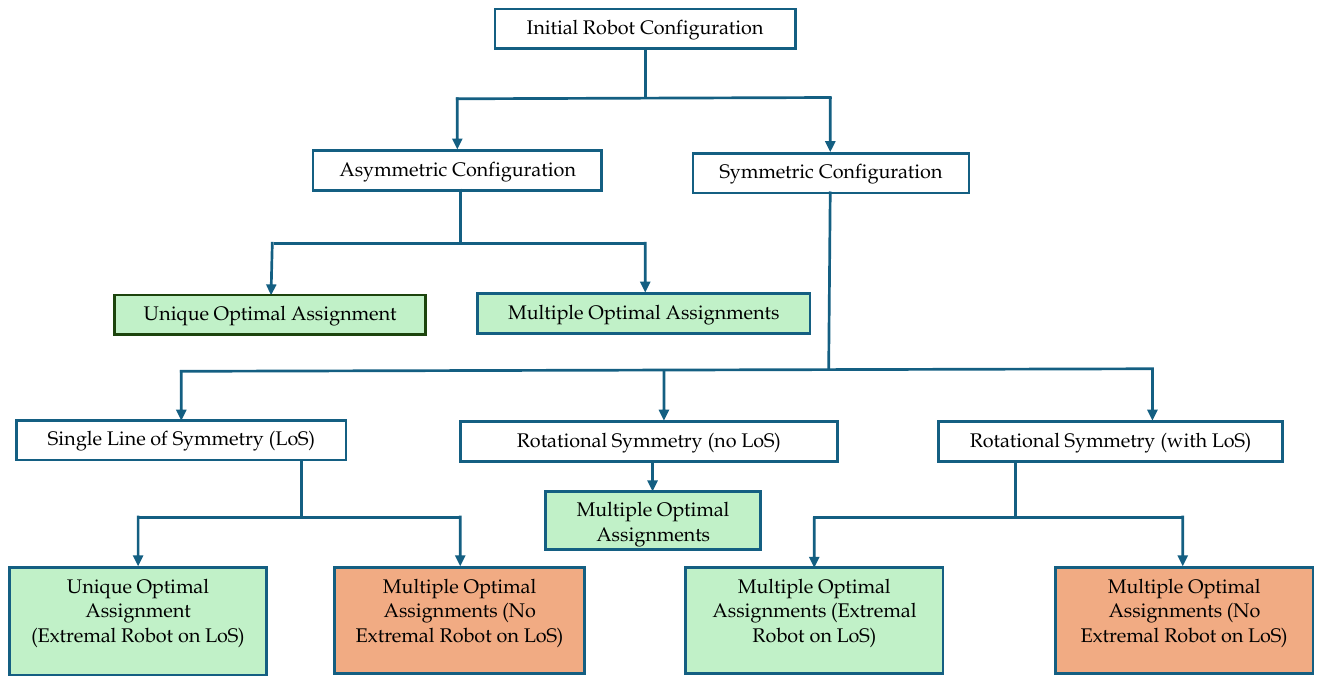}
    \vspace{-5.8cm}
    
    \caption{\textit{Illustrates the assignment tree used in this paper. Starting from the initial robot configuration, configurations are first classified according to their symmetry properties. Each class is further refined according to the cardinality of optimal assignments and the presence of an \textit{extremal} robot on a line of symmetry. Each leaf node corresponds to a distinct algorithmic strategy (green) or an impossibility result (red).}}
\label{Tree}
\end{figure}


\noindent Let $A_1, A_2, \ldots, A_n$ denote the initial positions of $n$ robots on $\mathcal{C}_{out}$ of the circle $\mathcal{C}$, and let $A_1', A_2', \ldots, A_n'$ denote their destination positions on the same circumference. The goal is to determine an assignment of robots to destination points that minimizes the total movement cost, measured along the circle. We denote by $\mathcal{D}^*$ the optimal value of the min-sum problem for $\mathcal{C}$, defined as
\[
\mathcal{D}^* = \min \sum_{i=1}^{n} d_{\text{arc}}(A_i, A_i')
\]

\noindent The structure of optimal assignments for this problem has been studied previously. In particular, the following results show that any optimal assignment must leave at least one robot fixed at its initial position.

\begin{result}~\cite{tan2010new} 
Suppose that all $n$ robots $A_1, A_2, \cdots, A_n$ are initially placed on the circumference of the unit-radius circle $\mathcal{C}$. In any assignment between the initial and destination positions of $n$ robots that yields $\mathcal{D}^*$, there exists at least one point $A_i$ $(1 \le i \le n)$ such that $A_i = A_i'$.
\label{r-02}
\end{result}

\noindent On this observation, Bhattacharya \textit{et al.}~\cite{bhattacharya2009optimal} characterized all possible optimal assignments. Their result shows that optimal assignments can be obtained by fixing one robot and placing the remaining robots at equally spaced positions around the circle.

\begin{result}~\cite{bhattacharya2009optimal} 
Suppose we are given $n$ robots $A_1, A_2, \cdots, A_n$ in clockwise order along the circumference of the disk. Then any minimal-cost assignment of destination points must be one of the following $n$ assignments:
\[
\big(A_1,A_2,\cdots,A_n\big) \to \big( A_i+\tfrac{(i-1)2\pi}{n},\; A_i+\tfrac{(i-2)2\pi}{n},\; \cdots,\; A_i+\tfrac{(i-n)2\pi}{n}\big),
\]
for $i = 1,2,\cdots,n$. Each such assignment fixes the robot $A_i$.
\label{r-2}
\end{result}

\noindent \textbf{Result}~\ref{r-2} implies that the optimal cost assignment is not always unique. Depending on the initial configuration, more than one fixed point may yield the same minimum total cost. Based on this observation, we analyze the problem in different cases. We first study configurations for which the optimal assignment is unique. We then consider configurations in which multiple distinct assignments achieve the optimal cost, and analyze how symmetry in the initial configuration gives rise to this ambiguity. To systematically organize this analysis, we introduce an assignment-based classification of configurations. \textbf{Figure}~\ref{Tree} presents the assignment tree used to classify robot configurations based on symmetry and optimal assignments. Each leaf of the tree represents a distinct configuration class, which is analyzed in the subsequent sections. \textbf{Table}~\ref{tab:config-summary} summarizes the classification of robot configurations induced by the assignment tree, along with their symmetry
properties and solvability status for the \textit{min-sum uniform coverage} on a circle problem. Each configuration class listed in the table corresponds to a distinct leaf of the assignment tree and is addressed in the subsequent sections. \noindent
We now present \textbf{Algorithm}~\ref{alg:global-minsum}, \textit{MinSumUniformCircleFormation()}, which brings together all the subroutines and applies to any initial robot configuration on the circumference $\mathcal{C}_{out}$.
 The algorithm first classifies the initial configuration based on its symmetry properties and then applies the appropriate procedure according to the identified configuration class. As established in the preceding sections, some configuration classes admit a deterministic and collision-free solution, while others are unsolvable. The pseudocode below summarizes this case-based strategy and specifies the corresponding action taken for each configuration class.

 \begin{table*}
\renewcommand{\arraystretch}{1.8}
\scriptsize
\begin{adjustwidth}{0.3in}{0in}
\begin{tabular}{|c|c|>{\centering\arraybackslash}p{2.2cm}|c|c|}
\hline
\textbf{Notation} 
& \textbf{Configuration Type} 
& \textbf{Line(s) of Symmetry (LoS)} 
& \textbf{Rotational Symmetry} 
& \textbf{Solvability} \\
\hline
$\mathcal{I}_1$ 
& Asymmetric 
& None 
& No 
& Solvable \\
\hline
$\mathcal{I}_2$ 
& Symmetry with \textit{extremal} robots on LoS 
& Exactly one 
& No 
& Solvable \\
\hline
$\mathcal{I}_3$ 
& Symmetric without \textit{extremal} robots on LoS 
& Exactly one 
& No 
& Unsolvable \\
\hline
$\mathcal{I}_4$ 
& Rotational symmetry (no LoS) 
& None 
& Yes 
& Solvable \\
\hline
$\mathcal{I}_5$ 
& Symmetric with \textit{extremal} robots on LoS 
& Multiple 
& Yes 
& Solvable \\
\hline
$\mathcal{I}_6$ 
& Symmetric without \textit{extremal} robots on LoS 
& Multiple 
& Yes 
& Unsolvable \\
\hline
\end{tabular}
\vspace{0.3cm}

\caption{Classification of robot configurations with respect to symmetry properties (line of symmetry and rotational symmetry) and their solvability status for the \textit{min-sum uniform coverage} on a circle problem.}
\label{tab:config-summary}
\end{adjustwidth}
\end{table*}

\begin{algorithm}[!ht]
\caption{\textit{: MinSumUniformCircleFormation()}}
\label{alg:global-minsum}
\begin{algorithmic}[1]
\Require Initial robot configuration $\mathcal{R}(t_0)$ on $\mathcal{C}_{out}$
\Ensure \textit{min-sum uniform coverage} on a circle (if solvable)

\State Determine the configuration class of $\mathcal{R}(t_0)$

\If{$\mathcal{R}(t_0) \in \mathcal I_1$ \Comment{Asymmetric configuration}}
    \State Execute \textit{AsymU1dMinSumC()}
    \State \Return
\EndIf

\If{$\mathcal{R}(t_0) \in \mathcal I_2$ \Comment{Single line of symmetry with an \textit{extremal} robot on the line}}
    \State Execute \textit{SymU1dMinSumC()}
    \State \Return
\EndIf

\If{$\mathcal{R}(t_0) \in \mathcal I_3$ \Comment{Single line of symmetry with no \textit{extremal} robot on the line}}
    \State \textbf{Stop: configuration is unsolvable}
    \State \Return
\EndIf

\If{$\mathcal{R}(t_0) \in \mathcal I_4$ \Comment{Rotational symmetry of order $w\ge 2$ with no line of symmetry}}
    \State Execute \textit{RotSymM1dMinSumC()}
    \State \Return
\EndIf

\If{$\mathcal{R}(t_0) \in \mathcal I_5$ \Comment{Multiple lines of symmetry, each containing an \textit{extremal} robot}}
    \State \textbf{Stop: already in uniform circle formation}
    \State \Return
\EndIf

\If{$\mathcal{R}(t_0) \in \mathcal I_6$ \Comment{Multiple lines of symmetry with no \textit{extremal} robot on any line}}
    \State \textbf{Stop: configuration is unsolvable}
    \State \Return
\EndIf

\end{algorithmic}
\end{algorithm}

\section*{4.1~~Configurations with a Unique Optimal Assignment}
\label{Configurations with a Unique Optimal Assignment}

\noindent In this section, we study the robot configurations that admit a unique optimal assignment, i.e., $|\mathcal{E}'(t)| = 1$. For such configurations, the minimum total movement cost uniquely determines the target position of each robot. As a result, the ambiguity arising from multiple feasible assignments is absent, which allows for a deterministic execution of the formation algorithm.

\noindent We consider two classes of configurations admitting a unique optimal assignment. The first class consists of asymmetric configurations, in which the lack of symmetry ensures the uniqueness of the assignment. The second class consists of symmetric configurations admitting a single line of symmetry, where the presence of an \textit{extremal} robot on the line of symmetry guarantees the solvability of the problem. We address these two classes separately in the following subsections.

\subsection*{4.1.1~~Asymmetric Configurations $\mathcal I'_1$: Unique Optimal Assignment}

\noindent In this subsection, we study asymmetric robot configurations $\mathcal I'_1 \in \mathcal I_1$ that admit a unique optimal assignment. Since the configuration is asymmetric, no nontrivial symmetry relates distinct robots, and the optimal assignment is uniquely determined by the relative order of robots along the circle. As a result, the \textit{extremal} robots are uniquely identified, and no ambiguity arises in assigning robots to destination points.
We show that, in this setting, the target configuration can be reached in a collision-free manner by directly assigning each robot to its unique destination.

\subsubsection*{4.1.1.1~~Overview of the Algorithm \it{AsymU1dMinSumC()}}

\noindent The algorithm \textit{AsymU1dMinSumC()} addresses the \textit{min-sum uniform coverage} on a circle problem for robots initially placed on the circumference of a circle~$\mathcal C$ with a unique optimal assignment, i.e., $|\mathcal{E}'(t)|=1$. It begins by computing the unique optimal assignment using \textbf{Result}~\ref{r-2}, which identifies the \textit{extremal} robot, say $r_e$, and determines the corresponding destination set $\mathcal P^e = \{p_1^{e}, p_2^{e}, \ldots, p_n^{e}\}$. In this paper, we assume that the robots move only along the circumference of $\mathcal C$, i.e., their motion is restricted to the circular path. For each robot $r_i$, the arc distance to its assigned destination is calculated as $ d_{arc}(r_i, p_i^{e})$. A robot is considered \textit{candidate} if its arc path toward its destination is unobstructed, and only such robots are allowed to move at a time, while the others remain stationary. In the presence of multiple \textit{candidate} robots, the robot with the minimum distance to its destination and the minimum view from \textbf{Observation}~\ref{obs:asymmetric-ordering} is chosen to move toward its destination. The process is repeated iteratively, and once a \textit{candidate} robot reaches its destination (i.e., $d_{arc} = 0$), the iteration terminates. Since exactly one \textit{candidate} robot advances toward its destination in each iteration, the algorithm ensures collision-free convergence in finite time, with all robots eventually occupying their designated points in $\mathcal P $. The pseudocode of the Algorithm {\it AsymU1dMinSumC()} is presented in \textbf{Algorithm}~\ref{alg:AsymU1dMinSumC}.

\begin{algorithm}
\caption{\textit{: AsymU1dMinSumC()}}
\label{alg:AsymU1dMinSumC}
\begin{algorithmic}[1]

\Require Circle $\mathcal{C}$; initial robot positions
$r_1(t_0), r_2(t_0), \ldots, r_n(t_0)$ on the circumference of $\mathcal{C}$
\Ensure Mapping $r_i \rightarrow p_i^{e}$ for all robots

\State Compute unique \textit{optimal assignment} using \textbf{Result}~\ref{r-2}
\State Identify \textit{extremal} robot $r_e$
\State Determine destination point set $\mathcal P^e=\{p_1^{e}, p_2^{e}, \ldots, p_n^{e}\}$
\State Calculate $d_i = d_{arc}\Big(r_i, p_i^{e}\Big)$ for each $i$

\While{not all robots have reached their destinations}
    \For{$i = 1$ to $n$}
        \If{$r_i = p_i^{e}$}
            \State Mark $r_i$ as terminated
        \Else
            \If{ $d_i$ is free, minimal, and $r_i$ has the minimum view (\textbf{Observation}~\ref{obs:asymmetric-ordering})}
                \State Move $r_i \rightarrow p_i^{e}$
            \Else
                \State $r_i$ waits
            \EndIf
        \EndIf
    \EndFor
\EndWhile

\end{algorithmic}
\end{algorithm}

\subsubsection*{4.1.1.2~~Correctness of the Algorithm \it{AsymU1dMinSumC()} }

\noindent Once an \textit{extremal} robot $r_e$ has been fixed, the destination point of each robot with respect to the optimal assignment is computed based on the position of $r_e$. A natural question arises: does this assignment remain invariant while a \textit{candidate} robot $r_j$ moves toward its assigned destination? We must show that the \textit{extremal} robot remains invariant while a \textit{candidate} robot $r_j$ moves toward its assigned destination. We next consider the following definition.

\begin{definition}[Intermediate Position]
\label{Intermediate}
Let $p_j(t_0)$ denote the initial position of robot $r_j$, and let $p_j^{e}$ be the destination assigned to $r_j$ at time $t_0$. For any time $t_1$ such that $r_j$ has not yet reached $p_j^{e}$, the position $p_j(t_1)$ is called an \emph{intermediate position} if
\[
p_j(t_1) \neq p_j^{e}.
\]
\end{definition}

\begin{lemma}[Invariance of Unique Optimal Assignment in $\mathcal I_1'$]
\label{l-3} 
Let $r_e$ be identified as an \textit{extremal} robot at time $t_0 \ge 0$. Fixing $r_e$, let $p_j^{e}$ on $\mathcal C_{out}$ denote the destination assigned at time $t_0$ to a candidate robot $r_j$. If $r_j$ reaches an intermediate position (See Definition~\ref{Intermediate}) on $\mathcal C_{out}$ at time $t_1$, then $r_e$ remains invariant over the interval $[t_0,t_1]$.
\end{lemma}


\begin{proof}
To justify that the optimal assignment remains invariant during the motion of the \textit{candidate} robot, it is necessary to examine how the total cost behaves for all robots in the system. In particular, we must verify that the movement of the \textit{candidate} robot $r_j$ does not affect the optimality of the configuration either with respect to its own assignment or with respect to the assignments of the other stationary robots. For this situation, the proof is divided into two parts:
We prove the following:

\begin{enumerate}
    \item \textbf{Invariance of the \textit{extremal} robot with respect to the \textit{Candidate} Robot ($r_j$)}
    \item \textbf{Invariance of the \textit{extremal} robot with other Robots ($r_k,\, k \ne e, j$)}
\end{enumerate}

\begin{figure*}[!ht]
    \centering
    \begin{subfigure}{0.48\textwidth}
        \includegraphics[width=\linewidth]{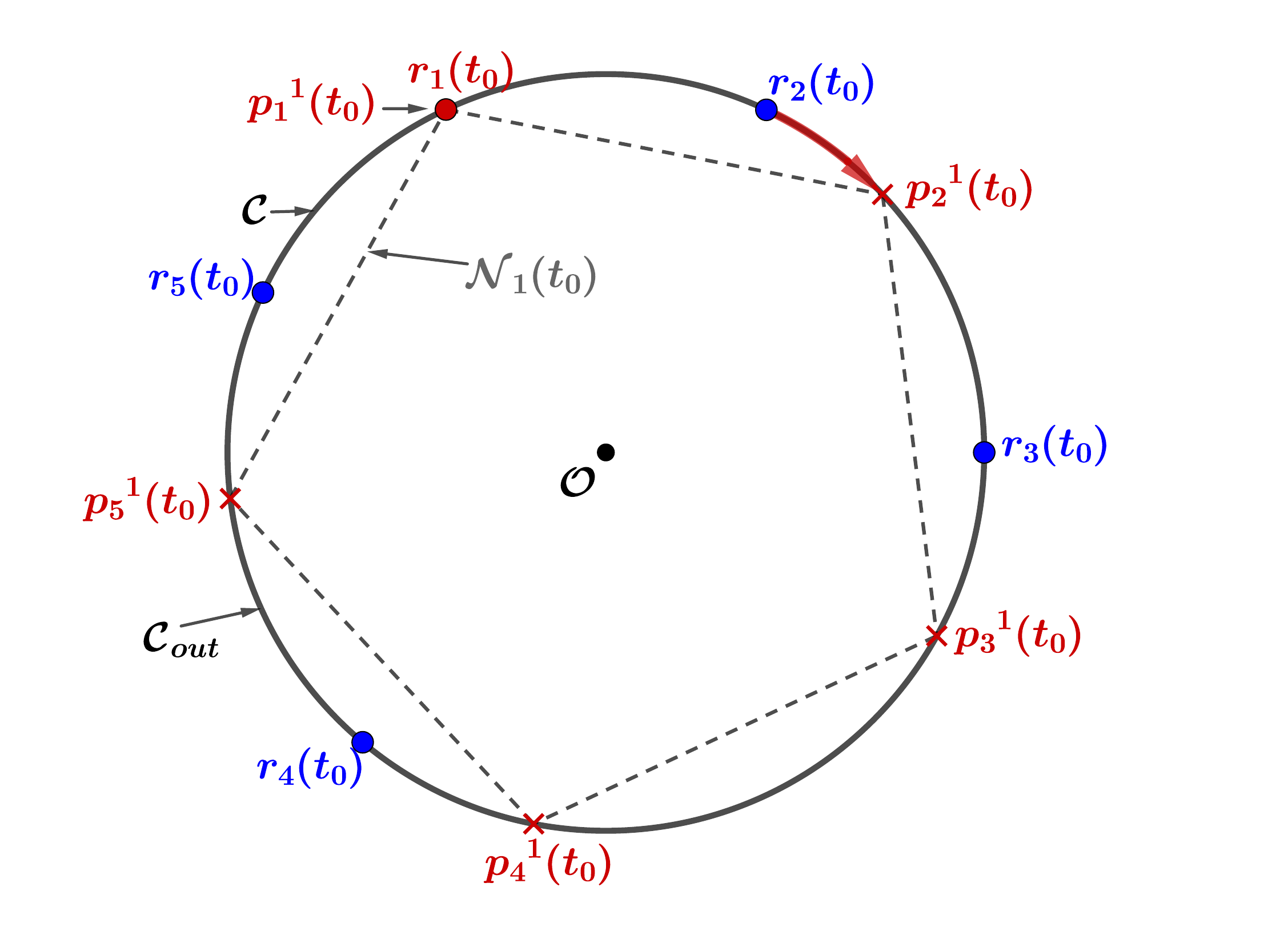}
        \caption*{(A)}
    \end{subfigure}
    \begin{subfigure}{0.48\textwidth}
        \includegraphics[width=\linewidth]{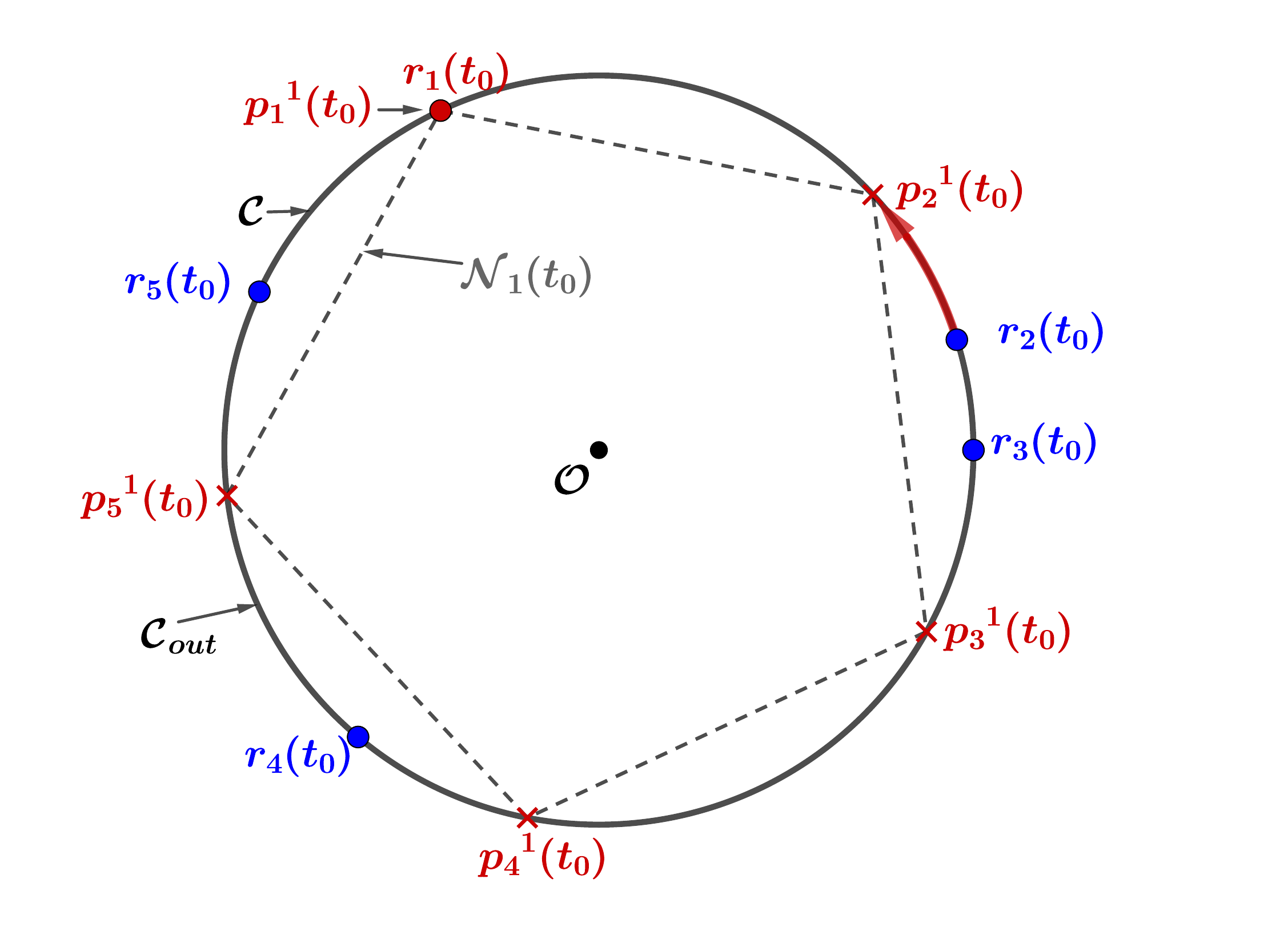}
        \caption*{(B)}
    \end{subfigure}

\caption{\textit{An illustration of the initial configuration of five robots
$\mathcal{R}(t_0)={r_1(t_0), r_2(t_0), \ldots, r_5(t_0)}$, where $r_1(t_0)$ is the \textit{extremal} robot (red dot).
Fixing $r_1(t_0)$ determines the destination set
${p_1^1(t_0), p_2^1(t_0), \ldots, p_5^1(t_0)}$ (red crosses).
\textbf{(A)} The destination $p_2^1(t_0)$ for robot $r_2(t_0)$ is assigned in the clockwise direction.
\textbf{(B)} The destination $p_2^1(t_0)$ for robot $r_2(t_0)$ is assigned in the counterclockwise direction.
}}
\label{Cw-Ccw}
\end{figure*}

\subsubsection*{1. Invariance with Respect to the Movement of $r_j$}
\noindent Let $\mathcal D^*_{r_e}(t_0)$ denote the unique optimal (minimum) sum distance traversed by all robots at time $t=t_0$ for the \textit{extremal} robot $r_e$. Also, let $\mathcal P^{e}(t_0)=\{p_1^{e}(t_0),p_2^{e}(t_0),\ldots,p_n^{e}(t_0)\}$
denote the set of $n$ destination points computed with respect to $r_e$ at time $t_0$. Note that the points in
$\mathcal P^{e}(t_0)$ lie on the circle $\mathcal C$ and form the vertices of a regular $n$-gon $\mathcal N_e(t_0)$ . Since we have assumed that there is a unique optimal assignment, thus we have:
\begin{equation}
\mathcal D^*_{r_e}(t_0) < \mathcal D_{r_j}(t_0)
\label{e-4}
\end{equation}
The candidate robot $r_j$ can move from its initial position toward its destination point $p_j^{e}(t_0)$ along the circle in either of the two possible directions, clockwise or counterclockwise (See \textbf{Figure}~\ref{Cw-Ccw}). Let $t_1 \in (t_0,T)$ be an arbitrary point of time at which  the \textit{candidate} robot $r_j$, and executes its LCM cycle. Suppose $x=d_{arc}\Big(r_j(t_0),\, r_j(t_1)\Big)$ denotes the distance traveled by $r_j$ during the time interval $[t_0,t_1]$. 

\noindent For all $t\geq 0$, let $R_{\mathrm{cw}}(t)$ and $R_{\mathrm{ccw}}(t)$ denote the sets of robots excluding $r_j$ whose destination points in $\mathcal N_j(t)$ determined by fixing $r_j$ with respect to the cost $\mathcal D_{r_j}(t)$ lies in the clockwise and counterclockwise directions, respectively. For $t\geq 0$, let 
$n_{\mathrm{cw}}(t) = |R_{\mathrm{cw}}(t)|$ and
$n_{\mathrm{ccw}} (t)= |R_{\mathrm{ccw}}(t)|$ denote their corresponding cardinalities.

\noindent As we want to prove that the assignment remains invariant, we need to prove that the condition holds at time instant $t_1$
\begin{equation}
\mathcal D^*_{r_e}(t_1)<\mathcal D_{r_{j}}(t_1) 
\end{equation}

\noindent Without loss of generality, suppose $r_j$ is assigned a \textit{clockwise} path to its destination point $p_j^{e}(t_0)$, as determined by the \textit{extremal} robot $r_e$, and reaches at position $r_j(t_1)$ at time $t_1$. The following cases are to be considered:

   \noindent \textbf{Case 1.}: $n_{\mathrm{cw}}(t_0)<n_{\mathrm{ccw}}(t_0)$ 
    
    \noindent If, at time $t_0$, the number of clockwise assignments of the other $(n-1)$ robots (when $r_j$ is fixed and it moves in clockwise direction) is smaller than the number of counterclockwise assignments, then the regular $n$-gon $\mathcal{N}_j(t_1)$ determined by fixing $r_j(t)$ at time $t_1$ may appear in two possible configurations (See \textbf{Figure}~\ref{Cross/No Cross}). In the first configuration, none of the vertices of $\mathcal{N}_j(t_1)$ intersect the initial positions (at time $t_0$) of the remaining $(n-1)$ robots. In the second configuration, one or more vertices of $\mathcal{N}_j(t_1)$ cross or coincide with the initial positions of other robots. The following subcases are to be considered.

\begin{figure*}[!ht]
    \centering
    \begin{subfigure}{0.48\textwidth}
        \includegraphics[width=\linewidth]{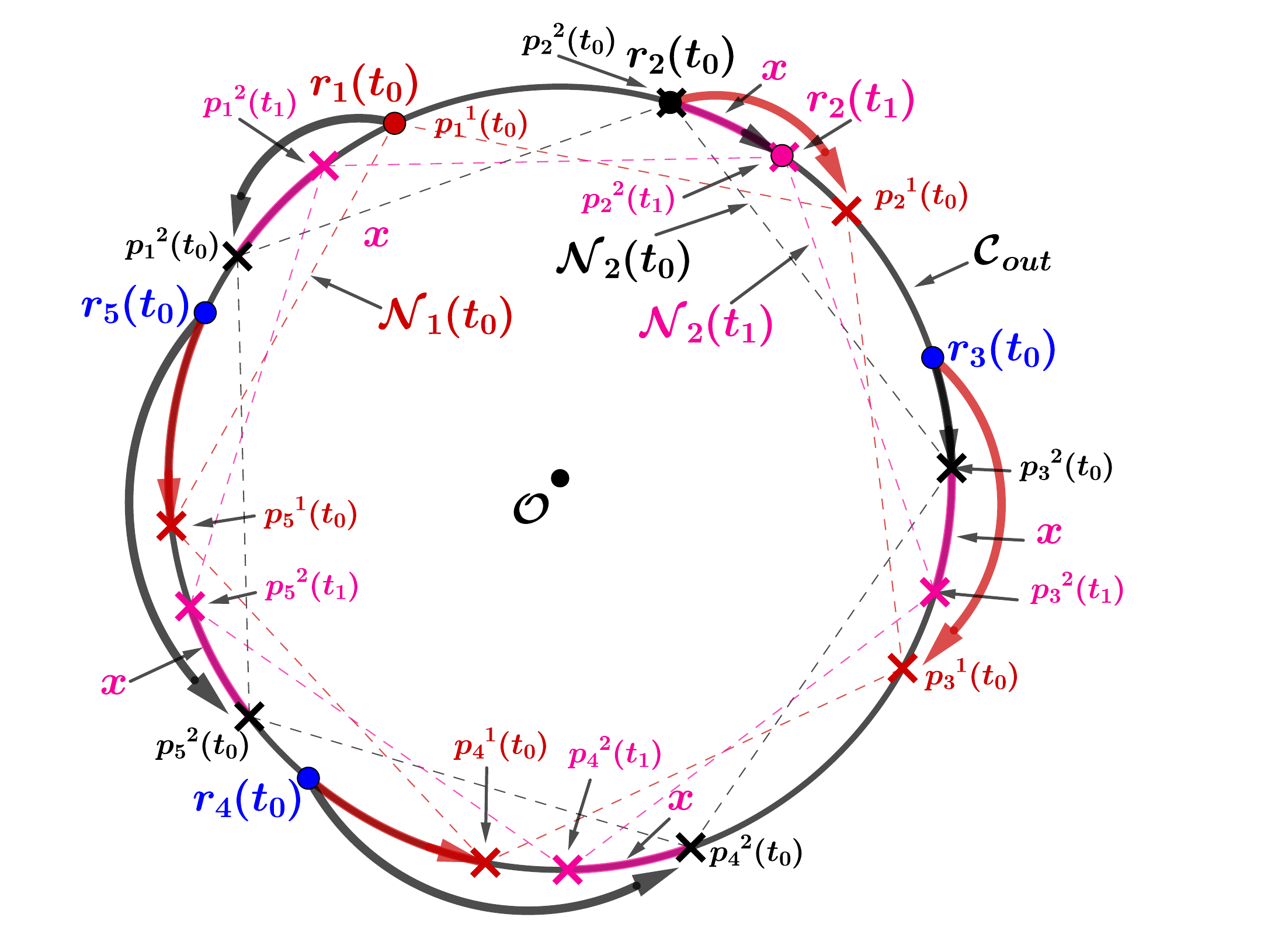}
        \caption*{(A)}
    \end{subfigure}
    \begin{subfigure}{0.48\textwidth}
        \includegraphics[width=\linewidth]{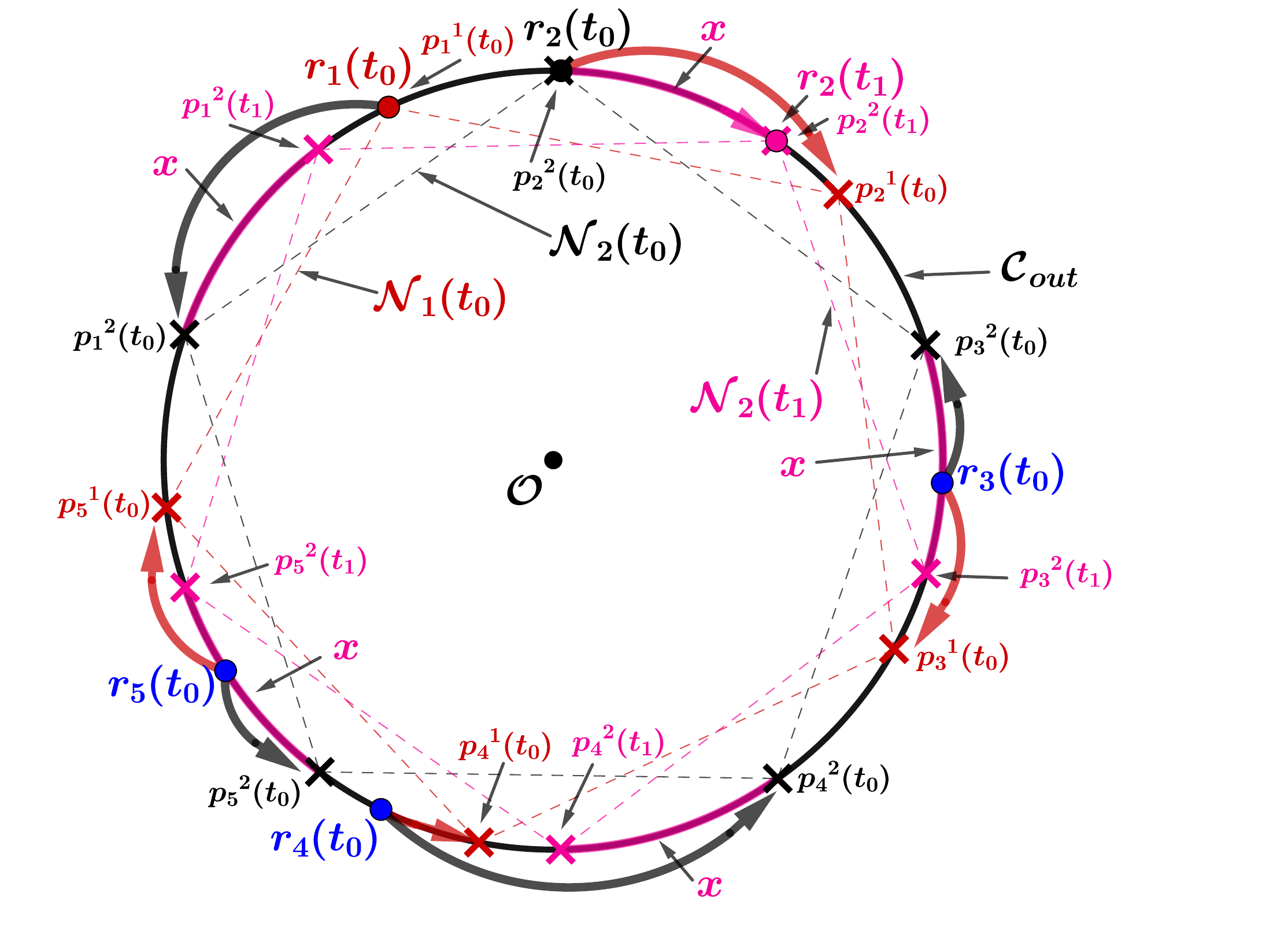}
        \caption*{(B)}
    \end{subfigure}

    \caption{ \textit{An illustration of the initial configuration of five robots $\mathcal{R}(t_0)=\{r_1(t_0), r_2(t_0), \ldots, r_5(t_0)\}$, where the \textit{extremal} robot $r_1(t_0)$ is shown as a red dot and the \textit{candidate} robot $r_2(t_0)$ as a black dot. The destination point set of the regular $n$-gon $\mathcal{N}_1(t_0)$, determined by fixing the \textit{extremal} robot $r_1(t_0)$, is $\mathcal{P}^1(t_0)=\{p_1^1(t_0), p_2^1(t_0), \ldots, p_5^1(t_0)\}$ (red crosses). The destination point sets of the regular $n$-gons $\mathcal{N}_2(t_0)$ and $\mathcal{N}_2(t_1)$, determined by fixing the candidate robot $r_2(t_0)$ at times $t_0$ and $t_1$, are given by $\mathcal{P}^2(t_0)=\{p_1^2(t_0), p_2^2(t_0), \ldots, p_5^2(t_0)\}$ (black crosses) and $\mathcal{P}^2(t_1)=\{p_1^2(t_1), p_2^2(t_1), \ldots, p_5^2(t_1)\}$ (pink crosses), respectively. Considering the clockwise assignment of $r_2(t_0)$ to $p_2^1(t_0)$: \textbf{(A)} none of the vertices of the regular polygon traced by $r_2(t_1)$ at time $t_1$ passes through or coincides with the position of any other robot; \textbf{(B)} some vertices of the polygon traced by $r_2(t_1)$ at time $t_1$ pass through the initial positions of robots $r_3(t_0)$ and $r_5(t_0)$, respectively. }}

    \label{Cross/No Cross}
\end{figure*}
\begin{itemize}
    \item  \textbf{Case 1(a):} \textit{No vertex of the regular $n$-gon $\mathcal{N}_j(t_1)$ crosses with any other stationary robot position during the movement $[t_0,t_1]$}

As $n_{\mathrm{cw}}(t_0)<n_{\mathrm{ccw}}(t_0)$, let $k=n_{\mathrm{ccw}}(t_0)-n_{\mathrm{cw}}(t_0)>0$. We establish equation (\ref{e-4}) through a contradiction argument. Let us assume that 

\begin{equation}
\mathcal D_{r_{j}}(t_1)<\mathcal D^*_{r_e}(t_1) 
\label{e-5}
\end{equation}

Since the \textit{candidate} robot $r_j$ moves toward its destination point $p_j^{e}(t_0)$ in the clockwise direction with $n_{\mathrm{cw}}(t_0) < n_{\mathrm{ccw}}(t_0)$, and no vertex of the regular $n$-gon $\mathcal{N}_j(t_1)$ crosses with any other stationary robot position (See \textbf{Figure}~\ref{Cross/No Cross}(A)), it follows that the number of arc lengths decreased by $x$ is greater than the number of arc lengths increased by $x$. The difference between these two counts is exactly $k$.
 Hence, the total sum reduces, and it is given by 
\begin{equation}
\mathcal D_{r_{j}}(t_1)=\mathcal D_{r_{j}}(t_0)-kx
\label{e-6}
\end{equation}


For the \textit{extremal} robot $r_e$, the total arc length decreases by exactly $x$ from the initial optimal sum distance at time $t_0$, that is, 
$\mathcal D^*_{r_e}(t_1) = \mathcal D^*_{r_e}(t_0) - x$. 
Since all other robots remain stationary, no additional changes occur in the overall configuration. Thus, we have,

\begin{equation}
     \mathcal D^*_{r_e}(t_0)=\mathcal D^*_{r_e}(t_1)+x
     \label{e-7}
\end{equation}

Using equation~(\ref{e-6}) in equation~(\ref{e-5}), we obtain
\begin{align*}
& \mathcal D_{r_{j}}(t_0) - kx < \mathcal D^*_{r_e}(t_1) \\
\implies\; 
& \mathcal D_{r_{j}}(t_0) - kx + x < \mathcal D^*_{r_e}(t_1) + x,
\quad \text{adding $x$ to both sides} \\
\implies\; 
& \mathcal D_{r_{j}}(t_0) - (k - 1)x < \mathcal D^*_{r_e}(t_0),
\quad \text{from equation~(\ref{e-7})}.
\end{align*}

We have an assignment of robot positions (say $\mathscr R$) whose cost is strictly smaller than $\mathcal D^*_{r_e}(t_0)$. In case, none of the $n$ robots is assigned to its own position at time $t_0$ with respect to $\mathscr R$, then it is a contradiction to \textbf{Result}~\ref{r-2}. Otherwise, it is a contradiction to our initial assumption that $\mathcal D^*_{r_e}(t_0)$ is the unique optimal assignment cost.

\begin{figure*}[!ht]
    \centering
    
    \includegraphics[width=0.5\textwidth]{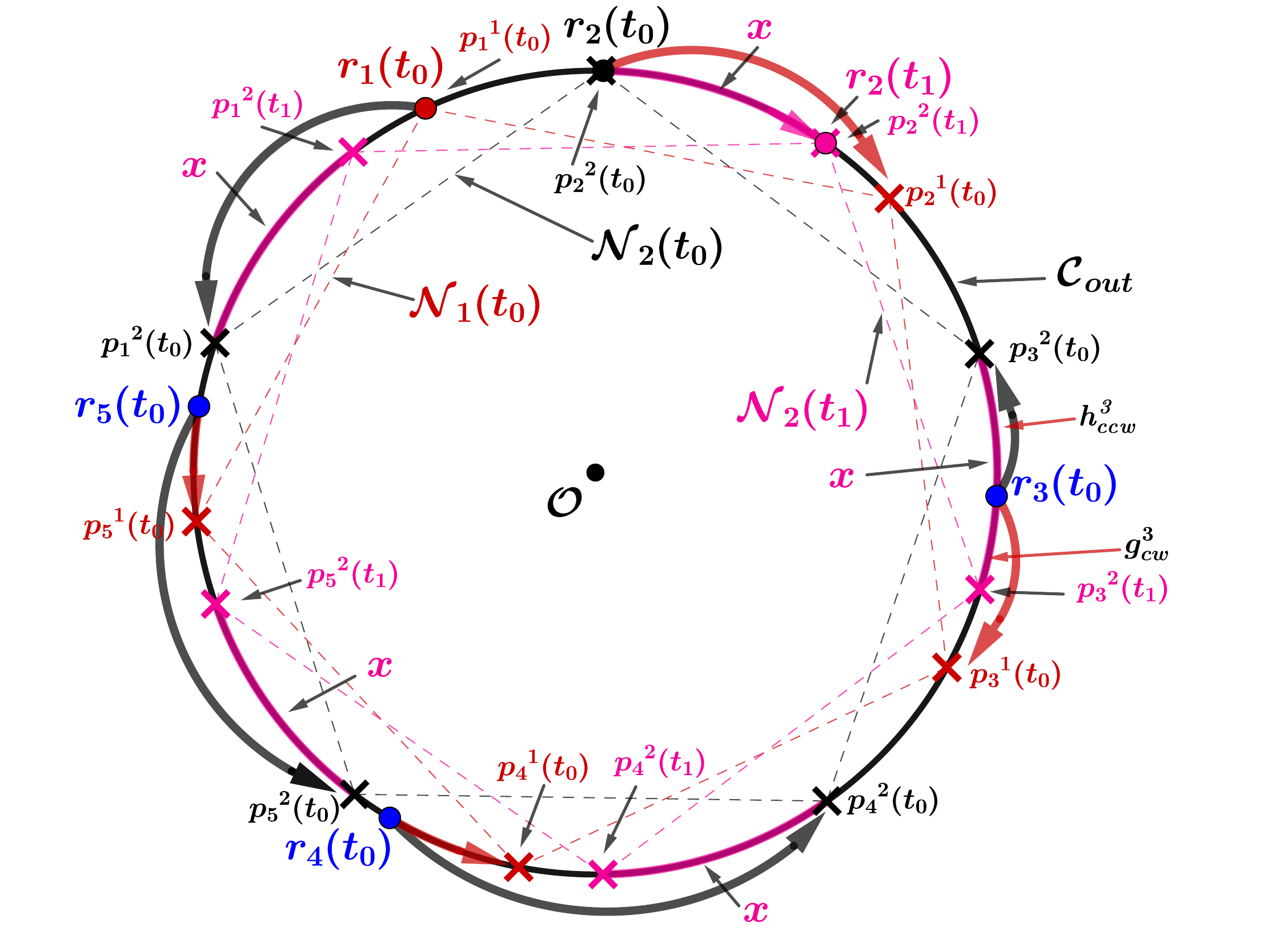}
    \caption{ \textit{An illustration of the initial configuration of five robots $\mathcal{R}(t_0)=\{r_1(t_0), r_2(t_0), \ldots, r_5(t_0)\}$, where the \textit{extremal} robot $r_1(t_0)$ is shown as a red dot and the \textit{candidate} robot $r_2(t_0)$ as a black dot. The destination point set of the regular $n$-gon $\mathcal{N}_1(t_0)$, determined by fixing $r_1(t_0)$, is $\mathcal{P}^1(t_0)=\{p_1^1(t_0), p_2^1(t_0), \ldots, p_5^1(t_0)\}$ (red crosses). The destination point sets of the regular $n$-gons $\mathcal{N}_2(t_0)$ and $\mathcal{N}_2(t_1)$, determined by fixing $r_2(t_0)$ at times $t_0$ and $t_1$, are $\mathcal{P}^2(t_0)$ (black crosses) and $\mathcal{P}^2(t_1)$ (pink crosses), respectively. The clockwise assignment of $r_2(t_0)$ to $p_2^1(t_0)$ illustrates the arc distances $h^3_{ccw}$ and $g^3_{cw}$. }}
    
     \label{Cross}

\end{figure*}

\item \textbf{ Case 1(b):} \textit{Some vertices of the regular $n$-gon $\mathcal{N}_j(t_1)$ cross the positions of some stationary robots during the movement $[t_0,t_1]$}

In this scenario, the \textit{candidate} robot $r_j$ moves toward its destination point $p_j^{e}(t_0)$ in the clockwise direction with $n_{\mathrm{cw}}(t_0) < n_{\mathrm{ccw}}(t_0)$, and some vertices of the regular $n$-gon $\mathcal{N}_j(t_1)$ cross the positions of certain stationary robots (See \textbf{Figure}~\ref{Cross/No Cross}(B)). Let $p$ denote the number of vertices of $\mathcal{N}_j(t_1)$ that cross stationary robot positions at time $t_1$. Hence, $(n-1)-p$ vertices of $\mathcal{N}_j(t_1)$ do not cross any stationary robot position at time $t_1$ (as the unique \textit{extremal} robot $r_e$ is kept fixed). Note that if a vertex of $\mathcal{N}_j(t_1)$ lies exactly on a stationary robot at time $t_1$, it contributes a decrease of $x$ to the total sum $D_{r_j}(t_0)$.

When the counterclockwise arc distance from a stationary robot to its destination is strictly smaller than $x$ before the movement of the candidate robot at $t_0$, then a vertex of the regular $n$-gon $\mathcal{N}_j(t_1)$ must cross the position of that robot at time $t_1$. Let $H_{ccw}(t_0) = \{h_{ccw}^i(t_0) \mid i = 1, 2, \dots, p\}$ denote the set of arc distances between the $p$ stationary robots and their corresponding destination points. Note that the vertices of the regular $n$-gon $\mathcal{N}_j(t_0)$) computed with respect to the \textit{candidate} robot $r_j$ at time $t_0$ (See \textbf{Figure} \ref{Cross}). After $r_j$ reaches the position $r_j(t_1)$ at time $t_1$, let $G_{cw}(t_1) = \{g_{cw}^i(t_1) \mid i = 1, 2, \dots, p\}$ denote the set of arc distances between the same $p$ stationary robots and their corresponding destination points (the vertices of the regular $n$-gon $\mathcal{N}_{j}(t_1)$) computed with respect to $r_j$ at time $t_1$ (See \textbf{Figure} \ref{Cross}). We have that $g_{cw}^i(t_1) \leq x$ and $h_{ccw}^i(t_0)<x$ for all $g_{cw}^i(t_1) \in G_{cw}(t_1)$ and $h_{ccw}^i(t_0) \in H_{ccw}(t_0)$. Since $p$ vertices of $\mathcal{N}_{j}(t_1)$ cross the positions of some stationary robots, then the initial counts of clockwise and counterclockwise assignments, $n_{\mathrm{cw}}(t_0)$ and $n_{\mathrm{ccw}}(t_0)$, for the stationary robots at time $t_0$ must change at time $t_1$. Let $n_{\mathrm{cw}}(t_1)$ and $n_{\mathrm{ccw}}(t_1)$ denote the updated numbers of clockwise and counterclockwise assignments, respectively, for all stationary robots except $r_j$ at time $t_1$. The possible relationships between $n_{\mathrm{cw}}(t_1)$ and 
$n_{\mathrm{ccw}}(t_1)$ at time $t_1$ are either
$n_{\mathrm{cw}}(t_1) < n_{\mathrm{ccw}}(t_1)$, or
$n_{\mathrm{cw}}(t_1) \ge n_{\mathrm{ccw}}(t_1)$. Our objective is to prove equation~(\ref{e-4}) by contradiction, assuming that equation~(\ref{e-5}) holds. Thus, we have the following cases.

 \begin{itemize}
 \item \boldmath\textbf{$n_{\mathrm{cw}}(t_1) < n_{\mathrm{ccw}}(t_1)\ \text{at}\ t_1:\ $}\unboldmath We assume $k'=n_{\mathrm{ccw}}(t_1)-n_{\mathrm{cw}}(t_1)$. Since $p$ vertices of $\mathcal{N}_{j}(t_1)$ cross the positions of some stationary robots, while the remaining $(n-1)-p$ vertices do not cross the positions of the other stationary robots, we apply the same argument discussed as in \textit{Subsubcase 1.1.1(a)} to obtain the relation between $\mathcal D_{r_{j}}(t_1)$ and $\mathcal D_{r_{j}}(t_0)$. The relation is given by

\begin{equation}
         \mathcal D_{r_{j}}(t_1)=\mathcal D_{r_{j}}(t_0)-k'x-\sum_{i=1}^ph_{ccw}^i+\sum_{i=1}^pg_{cw}^i
         \label{e-8}
\end{equation}
We can rewrite it 
    \begin{align*}
       & \mathcal D_{r_{j}}(t_0)-k'x-\sum_{i=1}^ph_{ccw}^i+\sum_{i=1}^pg_{cw}^i=\mathcal D_{r_{j}}(t_1)\\
     \implies & \mathcal D_{r_{j}}(t_0)-k'x-\sum_{i=1}^ph_{ccw}^i < \mathcal D_{r_{j}}(t_0)-k'x-\sum_{i=1}^ph_{ccw}^i+\sum_{i=1}^pg_{cw}^i=\mathcal D_{r_{j}}(t_1)\\
    \end{align*}
\vspace{-1.5cm}
  
    \begin{equation}
   \implies  \mathcal D_{r_{j}}(t_0)-k'x-\sum_{i=1}^ph_{ccw}^i <\mathcal D_{r_{j}}(t_1)
         \label{e-9}
\end{equation}

Since $h_{ccw}^i(t_0)<x,\ \forall\  h_{ccw}^i(t_0) \in H_{ccw}(t_0)$, we obtain 
\vspace{-0.5cm}

    \begin{align*}
       & \sum_{i=1}^ph_{ccw}^i< px\\
       \implies & -px<-\sum_{i=1}^ph_{ccw}\\
     \implies & \mathcal D_{r_{j}}(t_0)-k'x-px<\mathcal D_{r_{j}}(t_0)-k'x-\sum_{i=1}^ph_{ccw}^i\ \ [\text{as} \ k'x>0]\\
      \implies & \mathcal D_{r_{j}}(t_0)-k'x-px<\mathcal D_{r_{j}}(t_0)-k'x-\sum_{i=1}^ph_{ccw}^i<\mathcal D_{r_{j}}(t_1),\ \ \text{from equation (\ref{e-9})}
    \end{align*}
    \vspace{-0.5cm}
  
    \begin{equation}
   \implies  \mathcal D_{r_{j}}(t_0)-k'x-px<\mathcal D_{r_{j}}(t_1)
         \label{e-10}
\end{equation}

Using equation (\ref{e-10}) in equation (\ref{e-5}), we get
\vspace{-0.5cm}

    \begin{align*}
       & \mathcal D_{r_{j}}(t_0)-k'x-px<\mathcal D^*_{r_e}(t_1)\\
       \implies & \mathcal D_{r_{j}}(t_0) - k'x-px + x < \mathcal D^*_{r_e}(t_1) + x,\ \text{adding $x$ to both sides} 
    \end{align*}
From equation (\ref{e-7}), we get
\vspace{-0.5cm}

\begin{equation}
    \mathcal D_{r_{j}}(t_0) - (k'+p - 1)x < \mathcal D^*_{r_e}(t_0)
    \label{e-11}
\end{equation}
This contradicts equation~(\ref{e-4}), since it implies the existence of an assignment with total cost strictly smaller than the optimal cost at time $t_0$. Hence, the movement cost obtained by fixing $r_e$ is optimal.

     \item \boldmath\textbf{$n_{\mathrm{cw}}(t_1) = n_{\mathrm{ccw}}(t_1)\ \text{at}\ t_1:\ $}\unboldmath Substituting $k' = 0$ into equation~(\ref{e-11}) (since the number of positive and negative $x$ terms is equal, and their net contribution to the total sum is zero) leads to a contradiction with equation~(\ref{e-4}).

     \item \boldmath\textbf{$n_{\mathrm{cw}}(t_1) > n_{\mathrm{ccw}}(t_1)\ \text{at}\ t_1:\ $}\unboldmath We assume $n_{\mathrm{cw}}(t_1)-n_{\mathrm{ccw}}(t_1)=k''$. Equation (\ref{e-8}) changes to
     
     \begin{equation}
         \mathcal D_{r_{j}}(t_1)=\mathcal D_{r_{j}}(t_0)+k''x-\sum_{i=1}^ph_{ccw}^i+\sum_{i=1}^pg_{cw}^i
         \label{e-12}
     \end{equation}
     
     Using equation~(\ref{e-12}) with the same idea discussed in the part of \boldmath\textbf{$n_{\mathrm{cw}}(t_1)< n_{\mathrm{ccw}}(t_1)$}\unboldmath $\ \text{at}\ t_1$ , we obtain  
\begin{equation}
    \mathcal D_{r_{j}}(t_0) - (p-k'' - 1)x < \mathcal D^*_{r_e}(t_0),
    \label{e-13}
\end{equation}
which contradicts equation~(\ref{e-4}), as it implies the existence of an assignment with total cost strictly smaller than the optimal cost at time $t_0$. Hence, the movement cost obtained by fixing $r_e$ is optimal.

 \end{itemize}
\end{itemize}

\noindent The cases where $n_{\mathrm{cw}}(t_0) = n_{\mathrm{ccw}}(t_0)$ (Case 2) at time $t_0$ and $n_{\mathrm{cw}}(t_0) > n_{\mathrm{ccw}}(t_0)$ (Case 3) can be proved using arguments analogous to those used in the proof of \textbf{Case 1}, which considers the situation $n_{\mathrm{cw}} < n_{\mathrm{ccw}}$. Suppose the robot $r_j$ is assigned to its destination point $p_j^{e}(t_0)$ through a \textit{counterclockwise assignment} determined by the \textit{extremal} robot $r_e$, and subsequently reaches the position $r_j(t_1)$ at time $t_1$. Using the same idea discussed in \textit{Case 1.1}, we obtain the inequality $\mathcal D^*_{r_e}(t_1) < \mathcal D_{r_j}(t_1)$ at time $t_1$.\\

\noindent This shows that fixing the \textit{extremal} robot $r_e$ results in a smaller total distance than fixing the \textit{candidate} robot $r_j$ at time $t_1$. Therefore, the configuration obtained with respect to $r_e$ computes the optimal movement cost.

\subsubsection*{2. Invariance with respect to any other stationary robot $r_k$, where $k \ne e, j$}
\noindent Assume that another stationary robot $r_k$ does not yield the optimal sum at time $t_0$, where the corresponding sum obtained by fixing it is $\mathcal D_{r_k}(t_0)$. Let $p_j^k(t_0)$ denote the destination point assigned to robot $r_j$ when the assignment is determined by fixing robot $r_k$. Clearly,
\begin{equation}
\mathcal D^*_{r_e}(t_0) < \mathcal D_{r_k}(t_0)
\label{e-21}
\end{equation} Depending on the positions of $r_j(t_0)$, $p_j^{e}(t_0)$ and $p_j^k(t_0)$ by fixing the \textit{extremal} robots $r_e$ and another stationary robot $r_k$ in \textbf{Figure} \ref{5(1)}, \textbf{Figure} \ref{5(2)} and \textbf{Figure} \ref{5(3)}, respectively.


\begin{figure*}[!ht]
    \centering
    \begin{subfigure}{0.5\textwidth}
        \includegraphics[width=\linewidth]{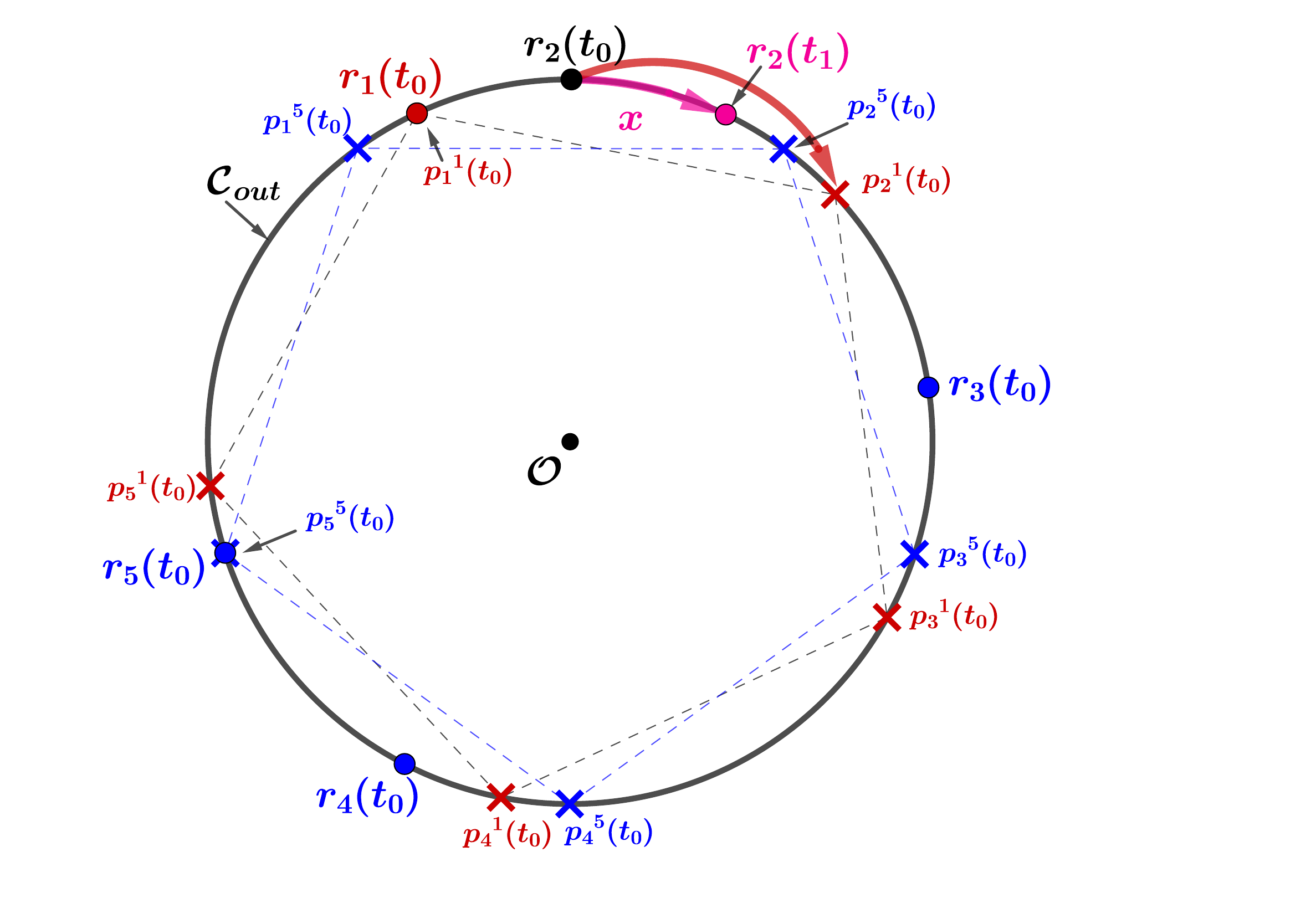}
        \caption*{(A)}
    \end{subfigure}
    \begin{subfigure}{0.48\textwidth}
        \includegraphics[width=\linewidth]{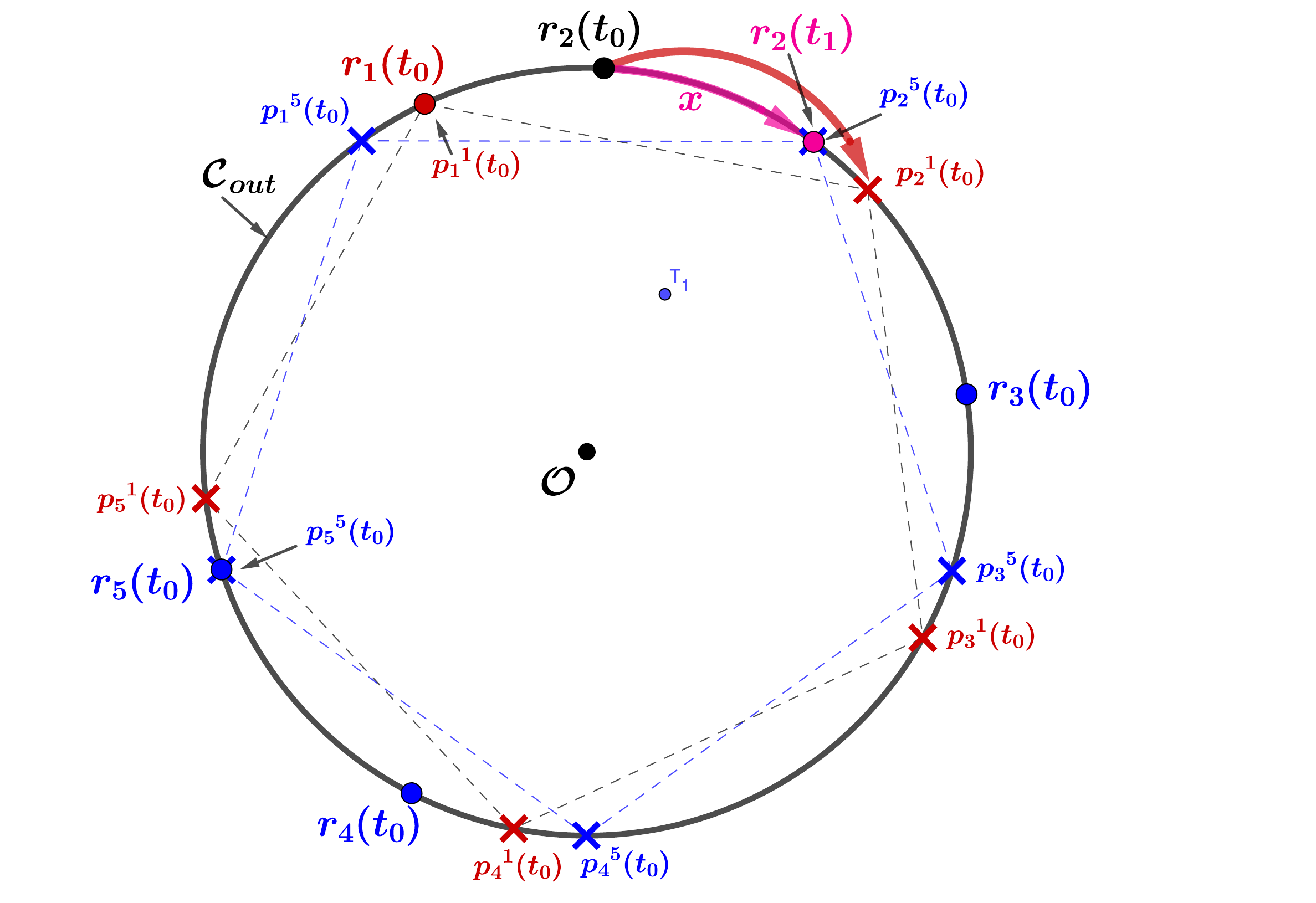}
        \caption*{(B)}
    \end{subfigure}
    \caption{ \textit{An illustration of the initial configuration of five robots $\mathcal{R}(t_0)=\{r_1(t_0),r_2(t_0),\ldots,r_5(t_0)\}$ with the \textit{extremal} robot $r_1(t_0)$ (red dot) and the \textit{candidate} robot $r_2(t_0)$ (black dot) and stationary robot is $r_5(t_0)$. The destination sets determined by fixing the \textit{extremal} robot $r_1(t_0)$ and the stationary robot $r_5(t_0)$ are $\mathcal{P}^1(t_0)=\{p_1^1(t_0),p_2^1(t_0),\ldots,p_5^1(t_0)\}$ (red crosses) and $\mathcal{P}^5(t_0)=\{p_1^5(t_0),p_2^5(t_0),\ldots,p_5^5(t_0)\}$ (blue crosses), respectively. \textbf{(A)} The destination point $p_2^5(t_0)$ lies ahead of robot $r_2(t_1)$ \textbf{(B)} The robot $r_2(t_1)$ coincides with $p_2^5(t_0)$.}}
    \label{5(1)}
\end{figure*}

\noindent Without loss of generality, assume that robot $r_j$ is assigned to its destination point $p_j^{e}(t_0)$ in the  \textit{clockwise} direction, as determined by the \textit{extremal} robot $r_e$,  and that it reaches the position $r_j(t_1)$ at time $t_1$. The possible cases are the following.

\begin{itemize}
    \item {\textbf{Case A:}} \textit{The destination point} $p_j^k(t_0)$ \textit{lies on the arc from} $r_j(t_1)$ \textit{toward its assigned destination (See \textbf{Figure} \ref{5(1)}(A))}.
    \vspace{0.2cm}
    
    Let, $d_{arc}\Big(r_j(t_0),r_j(t_1)\Big)=x$. Also suppose that $\mathcal D_{r_k}(t_1)$ and $\mathcal D^*_{r_e}(t_1)$ are the sums by fixing the robots $r_k$ and $r_e$ at time $t_1$ respectively.

Now, we have $\mathcal D_{r_k}(t_1)=\mathcal D_{r_k}(t_0)-x$ and $\mathcal D^*_{r_e}(t_1)=\mathcal D^*_{r_e}(t_0)-x$. From equation (\ref{e-21}), we get

\begin{align*}
&\mathcal D^*_{r_e}(t_0)-x < \mathcal D_{r_k}(t_0)-x\\
\implies\;& \mathcal D^*_{r_e}(t_1)<\mathcal D_{r_k}(t_1)
\end{align*}

\item {\textbf{Case B:}} \textit{The robot} $r_j(t_1)$ \textit{coincides with} $p_j^k(t_0)$\textit{(See \textbf{Figure} \ref{5(1)}(B))}.
 \vspace{0.2cm}

Let, $d_{arc}\Big(r_j(t_0),r_j(t_1)\Big)=x$. Also suppose that $\mathcal D_{r_k}(t_1)$ and $\mathcal D^*_{r_e}(t_1)$ are the sums by fixing the robots $r_k$ and $r_e$ at time $t_1$ respectively.

Now, we have $\mathcal D_{r_k}(t_1)=\mathcal D_{r_k}(t_0)-x$ and $\mathcal D^*_{r_e}(t_1)=\mathcal D^*_{r_e}(t_0)-x$. From equation (\ref{e-21}), we get

\begin{align*}
&\mathcal D^*_{r_e}(t_0)-x < \mathcal D_{r_k}(t_0)-x\\
\implies\;& \mathcal D^*_{r_e}(t_1)<\mathcal D_{r_k}(t_1)
\end{align*}

\begin{figure*}[!ht]
    \centering
    \begin{subfigure}{0.48\textwidth}
        \includegraphics[width=\linewidth]{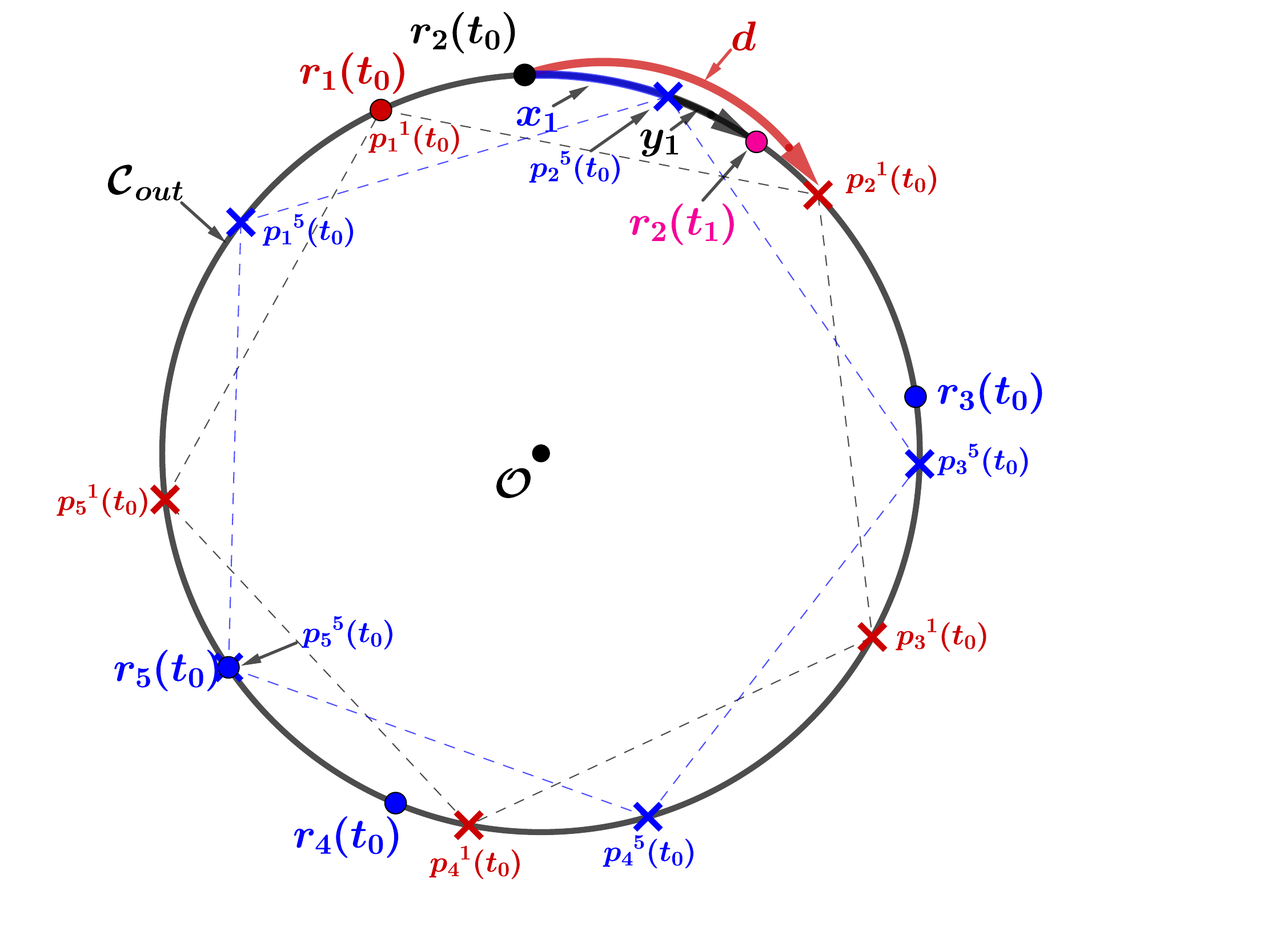}
        \caption*{(A)}
    \end{subfigure}
    \begin{subfigure}{0.48\textwidth}
        \includegraphics[width=\linewidth]{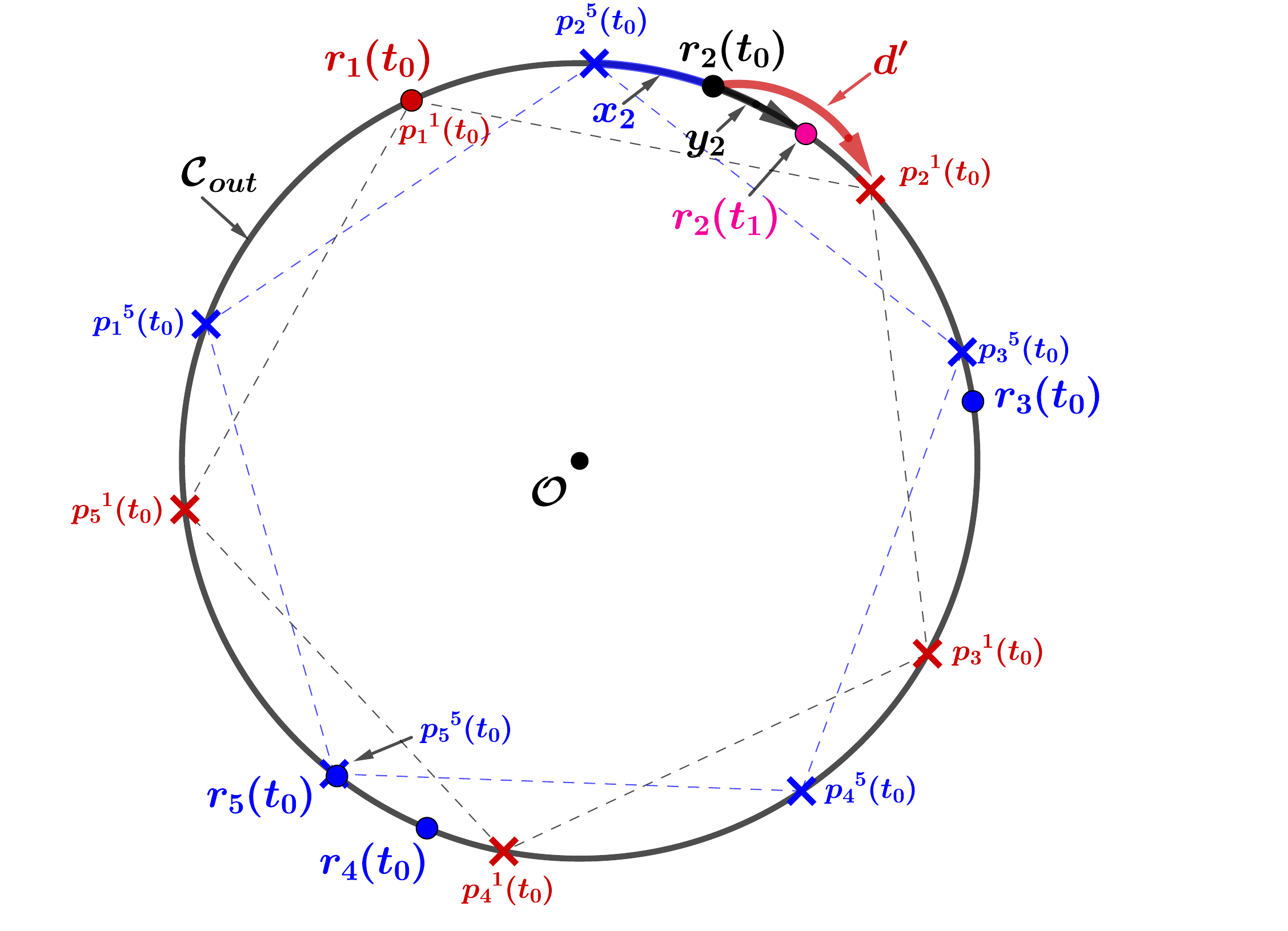}
        \caption*{(B)}
    \end{subfigure}
    \caption{ \textit{An illustration of the initial configuration of five robots $\mathcal{R}(t_0)=\{r_1(t_0),r_2(t_0),\ldots,r_5(t_0)\}$ with the \textit{extremal} robot $r_1(t_0)$ (red dot) and the \textit{candidate} robot $r_2(t_0)$ (black dot) and stationary robot is $r_5(t_0)$. The destination sets determined by fixing the \textit{extremal} robot $r_1(t_0)$ and the stationary robot $r_5(t_0)$ are $\mathcal{P}^1(t_0)=\{p_1^1(t_0),p_2^1(t_0),\ldots,p_5^1(t_0)\}$ (red crosses) and $\mathcal{P}^5(t_0)=\{p_1^5(t_0),p_2^5(t_0),\ldots,p_5^5(t_0)\}$ (blue crosses), respectively. \textbf{(A)} The robot $r_j(t_1)$ lies ahead of $p_j^k(t_0)$ \textbf{(B)} The robot $r_j(t_0)$ lies ahead of $p_j^{k}(t_0)$.}}
    \label{5(2)}
\end{figure*}

\item {\textbf{Case C:}} \textit{The robot} $r_j(t_1)$ \textit{lies further along the arc than the destination point} $p_j^k(t_0)$ \textit{(See \textbf{Figure}~\ref{5(2)}(A))}.

\vspace{0.2cm}
 
Let \(d_{arc}\Big(r_j(t_0), p_j^{k}(t_0)\Big) = x_1\), 
\(d_{arc}\Big(p_j^{k}(t_0), r_j(t_1)\Big) = y_1\), and 
\(d_{arc}\Big(r_j(t_0), p_j^{e}(t_0)\Big) = d\). 
The total movement cost is defined as 
\(\mathcal D^*_{r_e}(t_0) = D + d\) and \(\mathcal D_{r_k}(t_0) = D' + x_1\). Here, \(D\) and \(D'\) denote the total sum of distances contributed by all stationary 
robots, excluding the candidate robot \(r_j\), when evaluated with respect to the \textit{extremal} robot \(r_e\) and the robot \(r_k\), respectively, at time \(t_0\).

Also suppose that $\mathcal D_{r_k}(t_1)$ and $\mathcal D^*_{r_e}(t_1)$ are the sums by fixing the robots $r_k$ and $r_e$ at time $t_1$ respectively.

Clearly, $\mathcal D_{r_k}(t_1)=D'+x_1+y_1$ and $\mathcal D^*_{r_e}(t_1)=D+x_1+y_1$.

Using, \(\mathcal D^*_{r_e}(t_0) = D + d\) and \(\mathcal D_{r_k}(t_0) = D' + x_1\), we get, $\mathcal D_{r_k}(t_1)=\mathcal D_{r_k}(t_0)+y_1$ and $\mathcal D^*_{r_e}(t_1)=\mathcal D^*_{r_e}(t_0)-d+x_1+y_1=\mathcal D^*_{r_e}(t_0)-\big(d-(x_1+y_1)\big)$, respectively.

From equation (\ref{e-21}), we get
\vspace{-0.5cm}

\begin{align*}
&\mathcal D^*_{r_e}(t_0) < \mathcal D_{r_k}(t_0)\\
\implies\;& \mathcal D^*_{r_e}(t_0)+y_1 < \mathcal D_{r_k}(t_0)+y_1 = \mathcal D_{r_k}(t_1),
\quad\text{adding $y_1$ both sides and from the above equality}\\
\implies\;& \mathcal D^*_{r_e}(t_0)+y_1 < \mathcal D_{r_k}(t_1)\\
\implies\;& \mathcal D^*_{r_e}(t_0) < \mathcal D^*_{r_e}(t_0)+y_1 < \mathcal D_{r_k}(t_1)\\
\implies\;& \mathcal D^*_{r_e}(t_0)-\big(d-(x_1+y_1)\big) < \mathcal D^*_{r_e}(t_0) < \mathcal D^*_{r_e}(t_0)+y_1 < \mathcal D_{r_k}(t_1),
\quad\text{$d-(x_1+y_1)$ is a positive}\\
\implies\;& \mathcal D^*_{r_e}(t_1)=\mathcal D^*_{r_e}(t_0)-\big(d-(x_1+y_1)\big)
< \mathcal D^*_{r_e}(t_0) < \mathcal D^*_{r_e}(t_0)+y_1 < \mathcal D_{r_k}(t_1)\\
\implies\;& \mathcal D^*_{r_e}(t_1) < \mathcal D_{r_k}(t_1)
\end{align*}


\item {\textbf{Case D:}} \textit{The robot} $r_j(t_0)$ \textit{appears before the destination point} $p_j^{k}(t_0)$ \textit{along the arc} \textit{(See \textbf{Figure}~\ref{5(2)}(B))}.
 \vspace{0.2cm}

Let $d_{arc}\Big(r_j(t_0), p_j^{k}(t_0)\Big) = x_2$, where $x_2$ is measured in counterclockwise direction, while the \textit{candidate} robot $r_j$ moves in the clockwise direction. Further, let $d_{arc}\Big(r_j(t_0), r_j(t_1)\Big) = y_2$ and $d_{arc}\Big(r_j(t_0), p_j^{e}(t_0)\Big) = d'$.

\begin{figure}[!ht]
    \centering
    \vspace{-0.3cm}
    
    \includegraphics[width=0.5\textwidth]{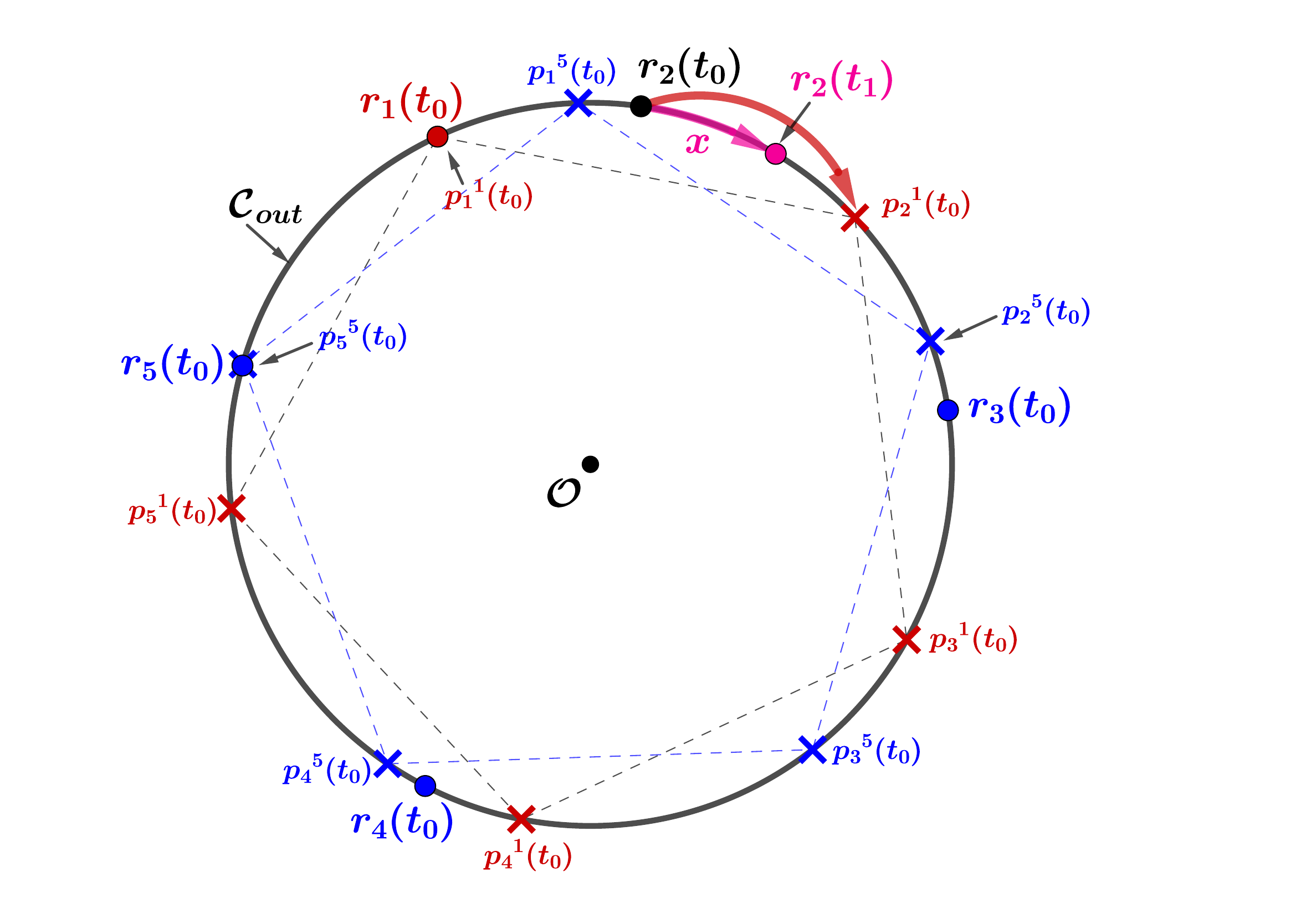}
    \caption{ \textit{An illustration of the initial configuration of five robots $\mathcal{R}(t_0)=\{r_1(t_0),r_2(t_0),\ldots,r_5(t_0)\}$ with the \textit{extremal} robot $r_1(t_0)$ (red dot) and the \textit{candidate} robot $r_2(t_0)$ (blue dot) and stationary robot is $r_5(t_0)$. The destination sets determined by fixing the \textit{extremal} robot $r_1(t_0)$ and the stationary robot $r_5(t_0)$ are $\mathcal{P}^1(t_0)=\{p_1^1(t_0),p_2^1(t_0),\ldots,p_5^1(t_0)\}$ (red crosses) and $\mathcal{P}^5(t_0)=\{p_1^5(t_0),p_2^5(t_0),\ldots,p_5^5(t_0)\}$ (blue crosses), respectively. The position $p_j^{k}(t_0)$ lies ahead of $p_j^{e}(t_0)$.}}
    
     \label{5(3)}

\end{figure}

Define the total sum at time \( \mathcal D^*_{r_e}(t_0) = D + d' \) and \( \mathcal D_{r_k}(t_0) = D' + x_2 \), where \( D \) and \( D' \) denote the total sum of all stationary robots excluding the candidate robot \( r_j \), evaluated by fixing the \textit{extremal} robot \( r_e \) and the robot \( r_k \), respectively, at time \( t_0 \). Also suppose that $\mathcal D_{r_k}(t_1)$ and $\mathcal D^*_{r_e}(t_1)$ are the sums by fixing the robots $r_k$ and $r_e$ at time $t_1$ respectively.

Clearly, $\mathcal D_{r_k}(t_1)=D'+y_2+x_2$ and $\mathcal D^*_{r_e}(t_1)=D+y_2$.

Using, \(\mathcal D^*_{r_e}(t_0) = D + d'\) and \(\mathcal D_{r_k}(t_0) = D' + x_2\), we get, $\mathcal D_{r_k}(t_1)=\mathcal D_{r_k}(t_0)+y_2$ and $\mathcal D^*_{r_e}(t_1)=\mathcal D^*_{r_e}(t_0)+y_2-d'$, respectively.

From equation (\ref{e-21}), we get

\begin{align*}
&\mathcal D^*_{r_e}(t_0) < \mathcal D_{r_k}(t_0)\\
\implies\;& \mathcal D^*_{r_e}(t_0)+y_2 < \mathcal D_{r_k}(t_0)+y_2 = \mathcal D_{r_k}(t_1),
\quad\text{adding $y_2$ and from the above equality}\\
\implies\;& \mathcal D^*_{r_e}(t_0)+y_2 < \mathcal D_{r_k}(t_1)\\
\implies\;& \mathcal D^*_{r_e}(t_0)+y_2-d' < \mathcal D^*_{r_e}(t_0)+y_2 < \mathcal D_{r_k}(t_1)\\
\implies\;& \mathcal D^*_{r_e}(t_1)=\mathcal D^*_{r_e}(t_0)+y_2-d'< \mathcal D^*_{r_e}(t_0)+y_2 < \mathcal D_{r_k}(t_1),
\quad\text{from the above equality}\\
\implies\;& \mathcal D^*_{r_e}(t_1) < \mathcal D_{r_k}(t_1)
\end{align*}


\item {\textbf{Case E:}} \textit{The position} $p_j^{k}(t_0)$ \textit{appears before} $p_j^{e}(t_0)$ \textit{along the arc} \textit{(See \textbf{Figure}~\ref{5(3)})}.

\vspace{0.2cm}

Let, $d_{arc}\Big(r_j(t_0),r_j(t_1)\Big)=x$. Also suppose that $\mathcal D_{r_k}(t_1)$ and $\mathcal D^*_{r_e}(t_1)$ are the two total movement costs by fixing the robots $r_k$ and $r_e$ at time $t_1$ respectively.

Now, we have $\mathcal D_{r_k}(t_1)=\mathcal D_{r_k}(t_0)-x$ and $\mathcal D^*_{r_e}(t_1)=\mathcal D^*_{r_e}(t_0)-x$. From equation (\ref{e-21}), we get
\vspace{-0.5cm}

\begin{align*}
&\mathcal D^*_{r_e}(t_0)-x < \mathcal D_{r_k}(t_0)-x\\
\implies\;& \mathcal D^*_{r_e}(t_1)<\mathcal D_{r_k}(t_1)
\end{align*}

\end{itemize}

\noindent This shows that fixing the \textit{extremal} robot $r_e$ leads to a smaller total distance than fixing the other stationary robot $r_k$ at time $t_1$.

\noindent Suppose robot $r_j$ moves toward its assigned destination point $p_j^{e}(t_0)$ in the \textit{counterclockwise} direction, as determined by the \textit{extremal} robot $r_e$, and reaches position $r_j(t_1)$ at time $t_1$. Using the same type of argument as before, we show that in all possible cases, fixing robot $r_e$ still results in a smaller total movement cost than fixing robot $r_k$ at time $t_1$. Therefore, the \textit{extremal} robot remains unchanged, which completes the proof of \textbf{Lemma~\ref{l-3}}.
\end{proof}
\noindent Note that at any instant of time, exactly one robot is allowed to move towards its destination. The robot moves only when it finds an unobstructed path towards its destination. Thus, we have the following observation. 
\begin{observation}
\label{obs-01}
The algorithm \textit{AsymU1dMinSumC()} achieves collision-free optimal placement of robots for the min-sum problem on the circle $\mathcal C $.
\end{observation}

\begin{lemma}
\label{l-4}
Suppose the \textit{candidate} robot \( r_j \) arrives at its destination 
\( p_j^{i} \), determined with respect to the \textit{extremal} robot \( r_e \), at time \( T \). Then the optimal assignment of robots to destination points remains identical whether it is computed relative to the position of \( r_e \) at time \( T \) 
or to the position of \( r_j \). Consequently, the total optimal cost is the same in both cases.
\end{lemma}

\begin{figure}[!ht]
    \centering
    \begin{subfigure}{0.48\linewidth}
        \includegraphics[width=\linewidth]{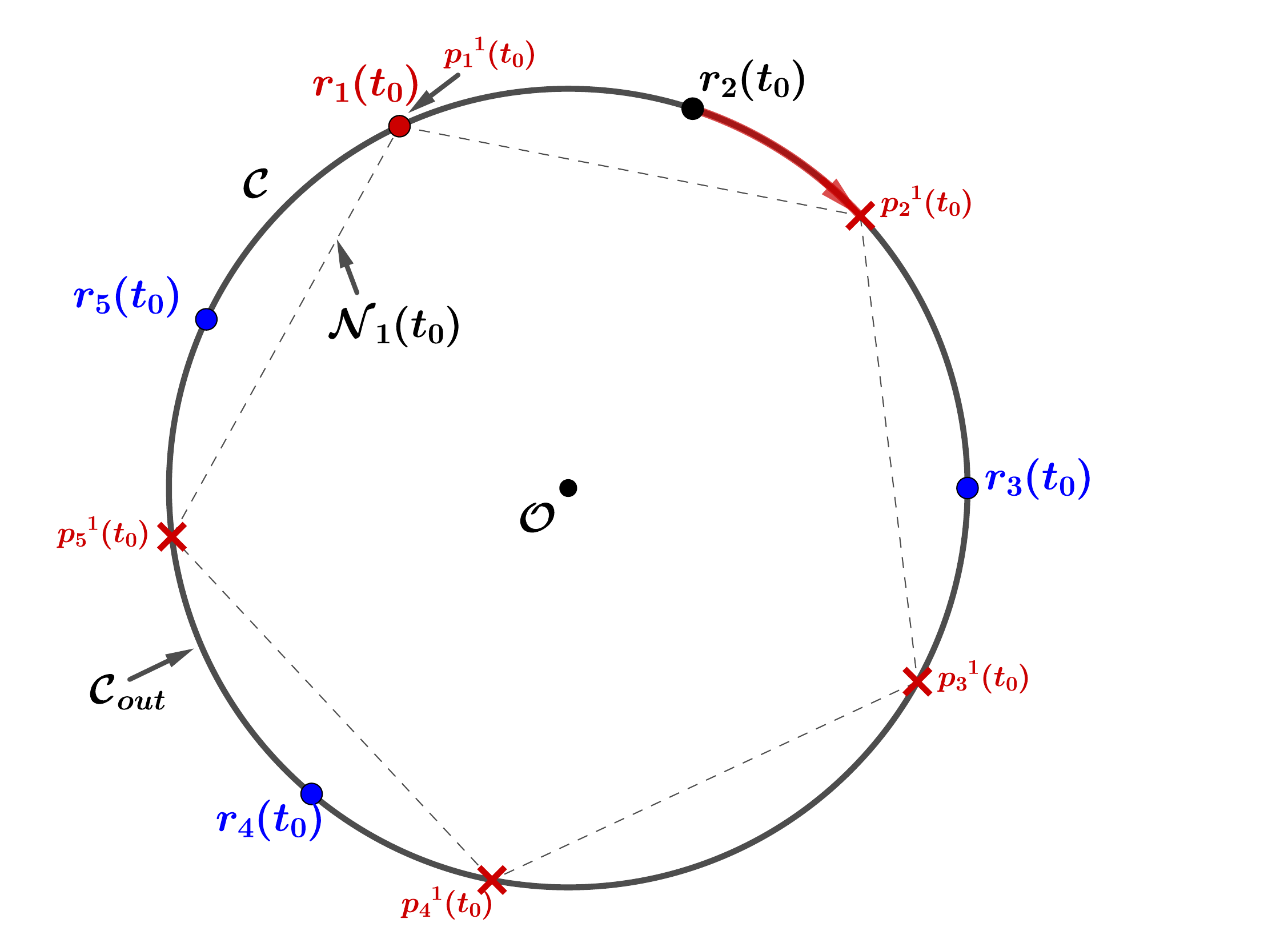}
        \caption*{(A)}
    \end{subfigure}
    \begin{subfigure}{0.48\linewidth}
        \includegraphics[width=\linewidth]{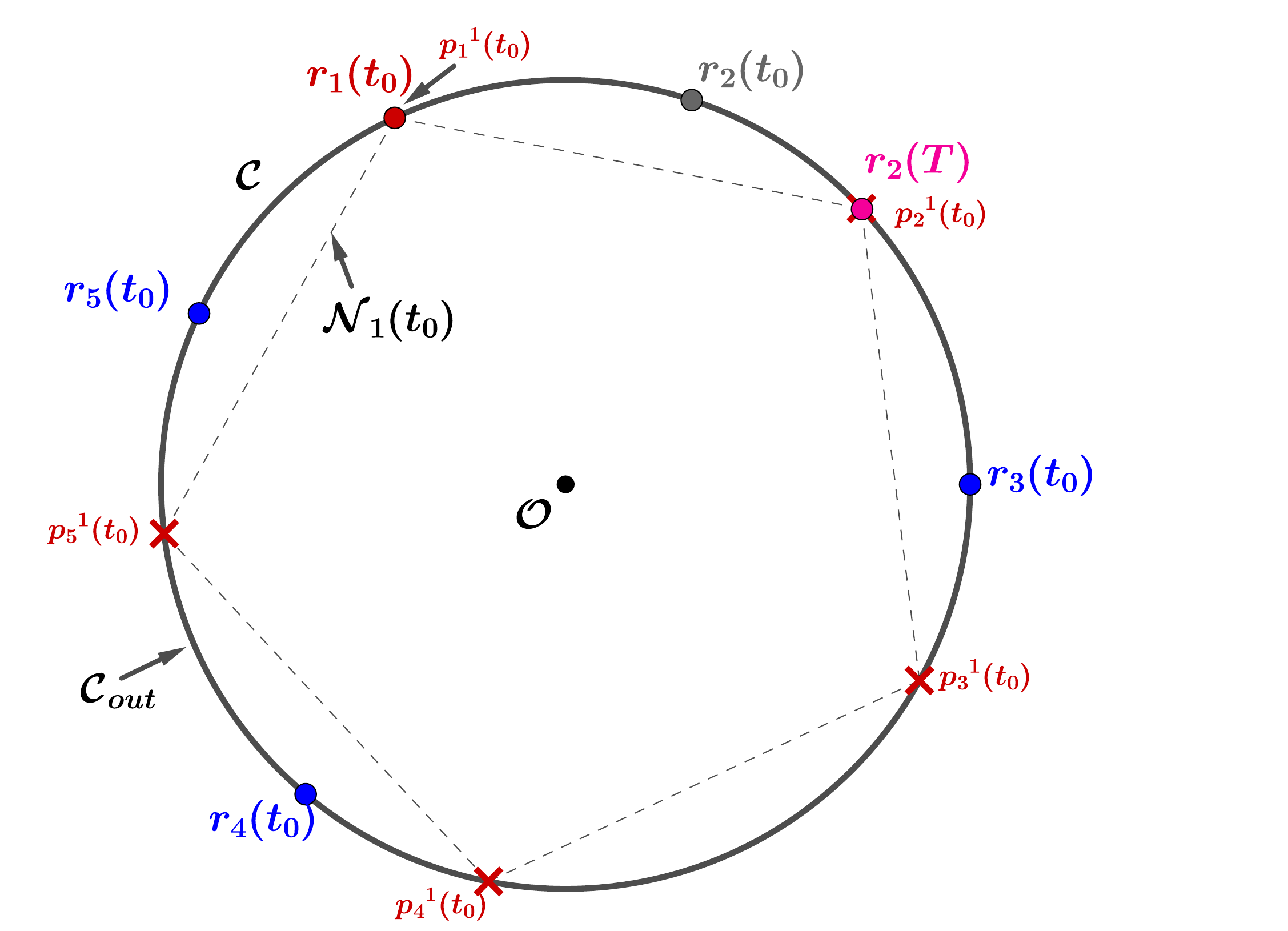}
        \caption*{(B)}
    \end{subfigure}
    \caption{An illustration of the initial configuration of five robots $\mathcal{R}(t_0)=\{r_1(t_0),r_2(t_0),\ldots,r_5(t_0)\}$ with the \textit{extremal} robot $r_1(t_0)$ (red dot) and the \textit{candidate} robot $r_2(t_0)$ (black dot). The destination set determined by fixing the \textit{extremal} robot $r_1(t_0)$ is $\mathcal{P}^1(t_0)=\{p_1^1(t_0),p_2^1(t_0),\ldots,p_5^1(t_0)\}$ (red crosses).The position $p_j^{k}(t_0)$ lies ahead of $p_j^{e}(t_0)$. \textbf{(A)} Position of the candidate robot $r_2$ before reaching $p_2^1(t_0)$ \textbf{(B)} Position of the candidate robot $r_2$ after reaching $p_2^1(t_0)$.}
    \label{Lem-4}
\end{figure}

\begin{proof}
  
\normalfont
Let $\mathcal{R}(t_0) = \{ r_1(t_0), r_2(t_0), \ldots, r_n(t_0) \}$ be a set of $n$ robots.
Also, let $\mathcal{P}^k(t_0) = \{ p_1^{k}(t_0), p_2^{k}(t_0),$ $ \ldots, p_n^{k}(t_0) \}$ denote the destinations forming a regular $n$-gon $\mathcal N_k(t_0)$ at time $t_0$, computed by robot $r_k$ using \textbf{Result} \ref{r-2}. Fix a robot $r_k$. At time $t_0$, let 
\[
f_k(t_0) : \mathcal{R}(t_0) \rightarrow \mathcal{P}^k(t_0)
\]
denote an assignment function that maps each robot in $\mathcal{R}(t_0)$ to a distinct target point on the circumference of the circle defined by $\mathcal{P}^k(t_0)$. Further let $\mathcal D_{r_e}(t)$ represent the optimum sum for \textit{extremal} robot $r_e$ under $f_k(t_0)$. At initial time \( t_0 \), the \textit{optimal} assignment calculated by fixing \textit{extremal} robot $r_e$ is given by

\[
f_e(t_0)
=
\left\{
\mathcal{R}(t_0)\!\rightarrow\!\mathcal{P}^{e}(t_0)\;\middle|\;
d_{arc}\Big(r_e(t_0),p_e^{e}(t_0)\Big) = 0,\;
d_{arc}\Big(r_i(t_0),p_i^{e}(t_0)\Big) \ne 0,\;
\right\}
\]
\[
\forall\, r_i(t_0)\in \mathcal R(t_0)
\ \land \
\forall\, p_i^{e}(t_0)\in \mathcal{P}^{e}(t_0).
\]


\noindent Suppose the \textit{candidate} robot \( r_j \) reaches its assigned destination \( p_j^e(t_0) \) calculated by fixing $r_e$ at time \( T \), with all other robots remain stationary (See \textbf{Figure} \ref{Lem-4}). Then, the positions at time \( T \) satisfy
\[
p_e^{e}(t_0)=r_e(t_0), \quad p_j^e(t_0)= r_j(T).
\]

\noindent At time \( T \), the optimal assignment calculated by fixing \textit{extremal} robot $r_e$ is given by

\[
f_e(T)
=
\left\{
\mathcal{R}(T)\!\rightarrow\!\mathcal{P}^{e}(t_0)\;\middle|\;
d_{\text{arc}}\Big(r_i(T), p_i^{e}(t_0)\Big)\neq 0,\;
d_{\text{arc}}\Big(r_e(t_0), p_e^{e}(t_0)\Big)=0,\;
d_{\text{arc}}\Big(r_j(T), p_j^{e}(t_0)\Big)=0
\right\}
\]
\[
\forall\, r_i(T)\in \mathcal R(T)
\ \land \
\forall\, p_i^{e}(t_0)\in \mathcal{P}^{e}(t_0).
\]




\noindent By the symmetry of a regular \( n \)-gon, the entire destination set calculated from \textbf{Result} \ref{r-2} remains invariant. Thus,

\[
\mathcal{P}^{e}(t_0) = \mathcal{P}^{j}(T)
\]
\[
 \text{i.e.,}\ \Big\{p_1^{e}(t_0), p_2^{e}(t_0), \ldots, p_n^{e}(t_0)\Big\}=\Big\{\ p_1^{j}(T), p_2^{j}(T), \ldots, p_n^{j}(T)\Big\}
\]

\noindent Consequently, the assignment computed by robot \( r_j \) at time \( T \) using \textbf{Result} \ref{r-2} yields the following:


\[
f_j(T)
=
\left\{
\mathcal{R}(T)\!\rightarrow\!\mathcal{P}^{j}(T)\;\middle|\;
d_{arc}\Big(r_i(T), p_i^{j}(t_0)\Big) \neq 0,\;
d_{arc}\Big(r_e(t_0), p_e^{j}(t_0)\Big) = 0,\;
d_{arc}\Big(r_j(T), p_j^{j}(t_0)\Big) = 0,\;
\right\}
\]



\noindent Therefore, the optimal assignment remains unchanged at time $T$. 
Hence, the assignment is determined with respect to $r_e$ and that determined 
with respect to $r_j$ are identical, i.e.,
\[
f_e(T)=f_j(T).
\]
Furthermore, since the assignment does not change, the corresponding optimal
total movement cost also remains the same, and thus
\[
\mathcal{D}_{r_e}(T)=\mathcal{D}_{r_j}(T).
\]
Hence, it is proved.

\end{proof}

\subsection*{4.1.2~~Symmetric Configurations $\mathcal I_2$: Unique Optimal Assignment}
\noindent This subsection studies symmetric robot configurations that admit exactly one line of symmetry and for which a unique optimal assignment exists. Although symmetry creates mirrored pairs of robots, the presence of a robot on the line of symmetry guarantees that the optimal assignment is uniquely defined. Consequently, the robot on the line of symmetry remains fixed, while the other robots move in symmetric pairs relative to it, thereby preserving the minimum total movement cost. This structure enables such configurations to be handled deterministically and collision-free. To formalize the concept of symmetry, we introduce a mapping that characterizes the reflective correspondence between robots with respect to the line of symmetry.

\begin{definition}[Mirror Mapping with Respect to a Line of Symmetry]
\label{def-mirror-map}
Let $\mathcal{R}(t)$ be a robot configuration on the circumference $\mathcal{C}_{out}$ that admits a line of symmetry $\mathcal{L}$.  
We define the \emph{mirror mapping} with respect to $\mathcal{L}$ as a function
\[
\mu_{\mathcal{L}} : \mathcal{R}(t) \rightarrow \mathcal{R}(t),
\]
such that for any robot $r_i(t) \in \mathcal{R}(t)$, $\mu_{\mathcal{L}}(r_i(t)) = r_j(t)$ if and only if the position of $r_j(t)$ is the reflection of the position of $r_i(t)$ across the line $\mathcal{L}$.
\end{definition}

\noindent We first establish a set of structural lemmas characterizing these configurations before presenting the algorithm \textit{SymU1dMinSumC()}.

\begin{figure}[!ht]
    \centering
     \vspace{-0.3cm}
    
    \includegraphics[width=0.5\textwidth]{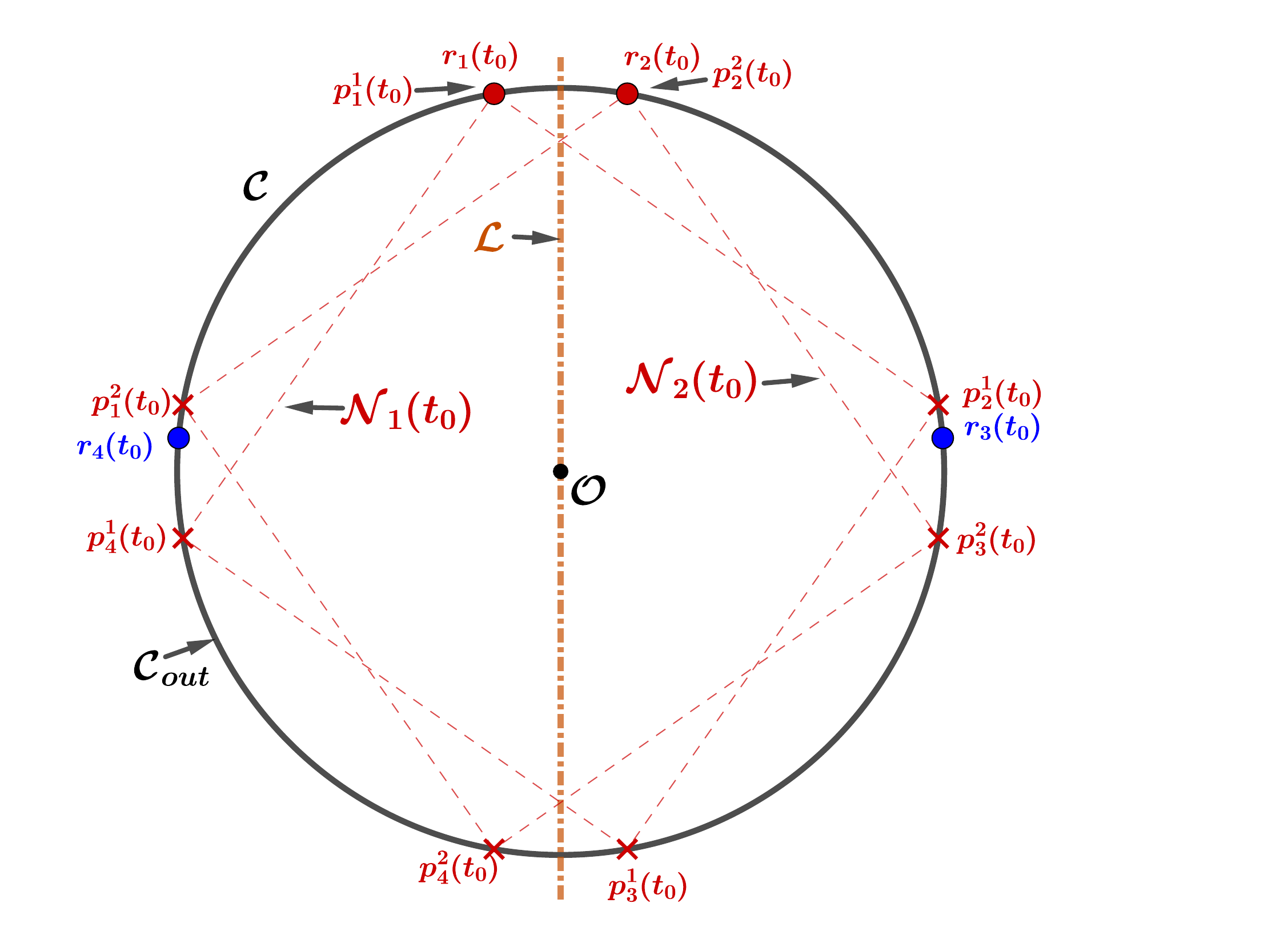}

    \caption{ \textit{An illustration of the initial configuration of four robots 
$\mathcal{R}(t_0)=\{r_1(t_0), r_2(t_0), r_3(t_0), r_4(t_0)\}$ admitting a line of symmetry $\mathcal L$. The \textit{extremal} robots are $r_1(t_0)$ and $r_2(t_0)$ (red dots), while the remaining robots are shown as blue dots.  The destination point sets determined by fixing the \textit{extremal} robots $r_1(t_0)$ and $r_2(t_0)$ are $\mathcal{P}^1(t_0)=\{p_1^1(t_0), p_2^1(t_0), \ldots, p_4^1(t_0)\}$ and 
$\mathcal{P}^2(t_0)=\{p_1^2(t_0), p_2^2(t_0), \ldots, p_4^2(t_0)\}$ (red crosses), respectively. 
With respect to the mirror mapping $\mu_{\mathcal L}$, we have $\mu_{\mathcal L}(r_1(t_0)) = r_2(t_0)$ and $\mu_{\mathcal L}(r_3(t_0)) = r_4(t_0)$.}}

\label{Lem-5}

\end{figure}

\begin{lemma}
\label{l-5}
If a robot configuration admits a single line of symmetry $\mathcal L$, and the
destination point sets are computed using \textbf{Result}~\ref{r-2} by fixing a robot
$r_i$ or its mirror robot $\mu_{\mathcal L}(r_i)$, then the optimal assignment
cost remains the same, i.e.,
\[
\mathcal{D}_{r_i}(t_0)
=
\mathcal{D}_{\mu_{\mathcal L}(r_i)}(t_0).
\]
\end{lemma}

\begin{proof}\normalfont
Let $\mathcal{R}(t_0)=\{r_1(t_0),r_2(t_0),\ldots,r_n(t_0)\}$ be a robot
configuration that admits a single line of symmetry $\mathcal L$. By
\textbf{Definition}~\ref{def-mirror-map}, for every robot $r_i(t_0)\in\mathcal{R}(t_0)$,
there exists a unique robot $\mu_{\mathcal L}(r_i(t_0))\in\mathcal{R}(t_0)$ whose
position is the reflection of $r_i(t_0)$ across $\mathcal L$ (See \textbf{Figure}~\ref{Lem-5}). If $n$ is odd, the
robot lying on $\mathcal L$ satisfies
$\mu_{\mathcal L}(r_i(t_0))=r_i(t_0)$. Fix a robot $r_k(t_0)$ and compute, using \textbf{Result}~\ref{r-2}, the
corresponding destination point set
\[
\mathcal{P}^{k}(t_0)
=
\{p_1^{k}(t_0),p_2^{k}(t_0),\ldots,p_n^{k}(t_0)\},
\]
which forms the vertex set of the regular $n$-gon $\mathcal{N}_{k}(t_0)$. Let
$f_k(t_0):\mathcal{R}(t_0)\rightarrow \mathcal{P}^{k}(t_0)$ denote the associated
assignment.

\noindent Now fix the mirror robot $\mu_{\mathcal L}(r_k(t_0))$ and recompute,
using \textbf{Result}~\ref{r-2}, the corresponding destination point set
\[
\mathcal{P}^{\mu_{\mathcal L}(k)}(t_0)
=
\{p_1^{\mu_{\mathcal L}(k)}(t_0),p_2^{\mu_{\mathcal L}(k)}(t_0),\ldots,
p_n^{\mu_{\mathcal L}(k)}(t_0)\},
\]
where each destination point
$p_i^{\mu_{\mathcal L}(k)}(t_0)$ is the reflection of
$p_i^{k}(t_0)$ across $\mathcal L$. Let
$f_{\mu_{\mathcal L}(k)}(t_0):\mathcal{R}(t_0)\rightarrow
\mathcal{P}^{\mu_{\mathcal L}(k)}(t_0)$ denote the associated assignment. The corresponding total arc-distance costs are
\[
\mathcal{D}_{r_k}(t_0)=
\sum_{i=1}^{n} d_{arc}\big(r_i(t_0),\,p_i^{k}(t_0)\big),
\]
and
\[
\mathcal{D}_{\mu_{\mathcal L}(r_k)}(t_0)=
\sum_{i=1}^{n}
d_{arc}\big(\mu_{\mathcal L}(r_i(t_0)),\,p_i^{\mu_{\mathcal L}(k)}(t_0)\big),
\]
with
\[
d_{arc}\big(r_k(t_0),p_k^{k}(t_0)\big)=0
\quad\text{and}\quad
d_{arc}\big(\mu_{\mathcal L}(r_k(t_0)),
p_k^{\mu_{\mathcal L}(k)}(t_0)\big)=0.
\]

\noindent Since reflection across $\mathcal L$ preserves arc distance along the
circle, for every $i$ we have
\[
d_{arc}\big(r_i(t_0),p_i^{k}(t_0)\big)
=
d_{arc}\big(\mu_{\mathcal L}(r_i(t_0)),
p_i^{\mu_{\mathcal L}(k)}(t_0)\big).
\]

\noindent Therefore,
\[
\mathcal{D}_{r_k}(t_0)
=
\mathcal{D}_{\mu_{\mathcal L}(r_k)}(t_0),
\]
which shows that fixing a robot or its mirror robot yields the same total assignment cost.
\end{proof}

\begin{figure}[!ht]
    \centering
    \vspace{-0.3cm}
    
    \includegraphics[width=0.5\textwidth]{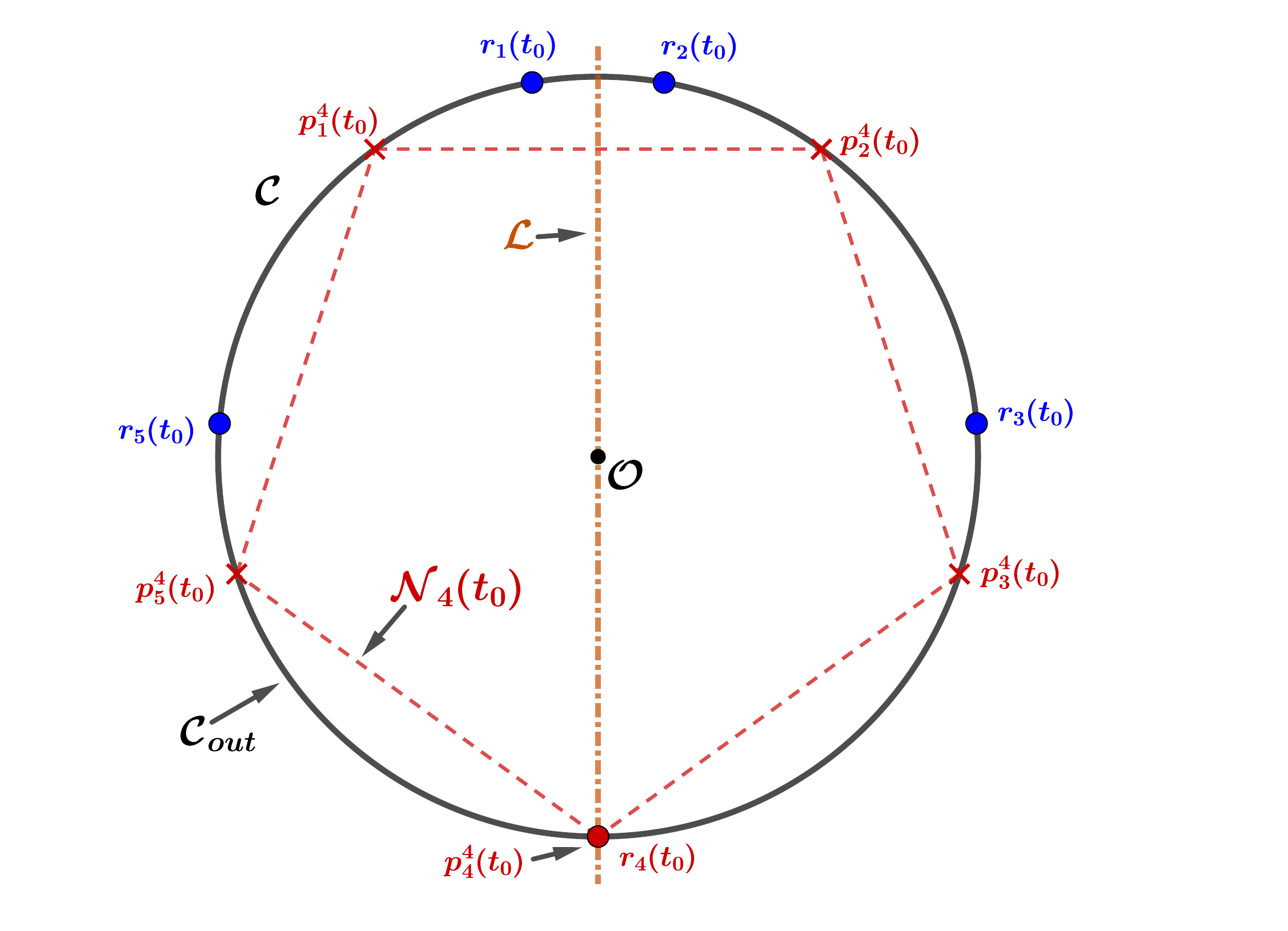}
    \caption{ \textit{An illustration of the initial configuration of five robots 
$\mathcal{R}(t_0)=\{r_1(t_0), r_2(t_0), \ldots, r_5(t_0)\}$ with a line of symmetry $\mathcal L$. 
The \textit{extremal} robot $r_4(t_0)$ (red dot) lies on $\mathcal L$. 
The destination set determined by fixing the \textit{extremal} robot $r_4(t_0)$ is 
$\mathcal{P}^4(t_0)=\{p_1^4(t_0), p_2^4(t_0), \ldots, p_5^4(t_0)\}$ (red crosses). 
With respect to the mirror mapping $\mu_{\mathcal L}$, we have 
$\mu_{\mathcal L}(r_1(t_0)) = r_2(t_0)$ and 
$\mu_{\mathcal L}(r_3(t_0)) = r_5(t_0)$.}}
  
     \label{Lem-6}

\end{figure}

\begin{lemma}
\label{l-6}
Consider a robot configuration that is symmetric with respect to a single line of symmetry $\mathcal{L}$. Let $r_e \in \mathcal{R}$ be the robot located at the intersection of $\mathcal{L}$ with the circumference of the circle. If the optimal assignment is unique and is obtained by fixing an \emph{\textit{extremal}} robot, then $r_e$ must be the \emph{\textit{extremal}} robot.
\end{lemma}

\begin{proof}
Let $\mathcal{R}(t_0)=\{r_1(t_0),r_2(t_0),\ldots,r_n(t_0)\}$ be a robot
configuration that is symmetric with respect to a single line of symmetry
$\mathcal{L}$. For every robot $r_i(t_0)\in\mathcal{R}(t_0)$, there exists a
unique robot $\mu_{\mathcal L}(r_i(t_0))\in\mathcal{R}(t_0)$ such that
$\mu_{\mathcal L}(r_i(t_0))$ is the reflection of $r_i(t_0)$ across
$\mathcal{L}$ (See \textbf{Figure}~\ref{Lem-6}). If $n$ is odd, exactly one robot lies on $\mathcal{L}$ and is
mapped to itself.

\noindent Let $r_e$ be the robot located at the intersection of $\mathcal{L}$ with the circumference of the circle. Suppose, for the sake of contradiction, that the unique optimal assignment is obtained by fixing some robot
$r_k(t_0)\neq r_e(t_0)$.

\noindent Since the configuration is symmetric, there exists a unique mirror robot
$\mu_{\mathcal L}(r_k(t_0))$ which is the reflection of $r_k(t_0)$ across
$\mathcal{L}$. Applying \textbf{Result}~\ref{r-2} to $r_k(t_0)$ and
$\mu_{\mathcal L}(r_k(t_0))$ produces the vertex sets of two regular $n$-gons,
$\mathcal N_k(t_0)$ and $\mathcal N_{\mu_{\mathcal L}(k)}(t_0)$ respectively
\[
\mathcal{P}^{k}(t_0)
\quad\text{and}\quad
\mathcal{P}^{\mu_{\mathcal L}(k)}(t_0),
\]
where each $p_i^{\mu_{\mathcal L}(k)}(t_0) \in
\mathcal{P}^{\mu_{\mathcal L}(k)}(t_0)$ is the reflection of
$p_i^{k}(t_0)\in \mathcal{P}^{k}(t_0)$ across $\mathcal{L}$.

\noindent Let $f_k(t_0)$ and $f_{\mu_{\mathcal L}(k)}(t_0)$ denote the corresponding assignments, and let $\mathcal{D}_{r_k}(t_0)$ and $\mathcal{D}_{\mu_{\mathcal L}(r_k)}(t_0)$ be their respective total
arc-distance costs. From \textbf{Lemma}~\ref{l-5}, reflection preserves the total movement cost, and hence
\[
\mathcal{D}_{r_k}(t_0)=\mathcal{D}_{\mu_{\mathcal L}(r_k)}(t_0).
\]

\noindent Thus, two distinct robots ($r_k$ and $\mu_{\mathcal L}(r_k)$) gives the same total assignment cost. This contradicts the assumption that the optimal assignment is unique. Therefore, the only robot that can yield a unique optimal assignment is the robot that lies on the line of symmetry, namely $r_e$. Hence, $r_e$ must be the \textit{extremal} robot.
\end{proof}

\begin{algorithm}
\caption{\textit{: SymU1dMinSumC()}}
\label{alg:SymU1dMinSumC}
\begin{algorithmic}[1]
\Require Initial robot configuration $\mathcal{R}(t_0)$ on the circumference of circle $\mathcal C$
\Ensure Optimal assignment achieving uniform circle formation with minimum total arc distance

\State Detect whether $\mathcal{R}(t_0)$ admits a single line of symmetry $\mathcal L$
\State Identify the robot $r_e$ at the intersection of $\mathcal L$ with $\mathcal C$
\State Compute the destination set $\mathcal P^e=\{p_1^{e},p_2^{e},\ldots,p_n^{e}\}$ by fixing $r_e$ using \textbf{Result}~\ref{r-2}

\For{$i = 1$ to $n$}
    \State Compute $d_i = d_{arc}(r_i(t_0),p_i^{e})$
\EndFor

\For{$i = 1$ to $n$}
    \If{$r_i(t_0)=p_i^{e}$}
        \State Mark $r_i$ as terminated
    \Else
        \If{$r_i$ is a \textit{candidate} robot and has minimum $d_i$}
            \If{$r_i$ has the minimum view (\textbf{Observation}~\ref{obs:asymmetric-ordering})} 
                \State Move $r_i \rightarrow p_i^{e}$ along $\mathcal C$
            \Else
                \State $r_i$ waits
            \EndIf
        \Else
            \State $r_i$ waits
        \EndIf
    \EndIf
\EndFor

\end{algorithmic}
\end{algorithm}

\noindent Let $\mathcal R(t_0) \in \mathcal I_2$ be an initial robot configuration that has a single line of symmetry $\mathcal L$, where $\mathcal L$ intersects the circle $\mathcal{C}_{out}$ at two robots, say $r_i$ and $r_j$. Because of symmetry, both robots compute the same regular $n$-gon using \textbf{Result}~\ref{r-2}. Hence, both robots are associated with the vertices of the same regular $n$-gon. Therefore, by \textbf{Lemma}~\ref{l-6}, the following two observations follow.

\begin{observation}
\label{obs-two-on-los}
Consider a robot configuration in the set $\mathcal{I}_2$ that admits a single line of symmetry $\mathcal{L}$. If $\mathcal{L}$ intersects the circumference $\mathcal{C}_{out}$ at exactly two robot positions, then both robots are \textit{extremal} and lie on the same unique $n$-gon.
\end{observation}

\begin{observation}
\label{obs-multi-los}
If a robot configuration belonging to the set $\mathcal I_5$ admits multiple lines of symmetry and each line of symmetry intersects the circumference $\mathcal{C}_{out}$ at exactly one \textit{extremal} robot, then all robots are already positioned on the vertices of a regular $n$-gon on $\mathcal{C}_{out}$. Hence, the uniform circle formation is already achieved at time $t_0$, and no robot movement is required.
\end{observation}

\subsubsection*{4.1.2.1~~Overview of the Algorithm \it{SymU1dMinSumC()}}

\noindent The algorithm \textit{SymU1dMinSumC()} addresses the \textit{min-sum uniform coverage} on a circle problem for the robot configurations that admit a single line of symmetry~$\mathcal L$ and for which the optimal assignment is unique. It begins by detecting the line of symmetry $\mathcal L$ and identifying the robot $r_e$ located at the intersection of $\mathcal L$ with the circumference of the circle~$\mathcal C$. By \textbf{Lemma}~\ref{l-6}, this robot must be the \textit{extremal} robot. The algorithm then computes the unique optimal assignment using \textbf{Result}~\ref{r-2} by fixing $r_e$, which determines the corresponding destination set $\mathcal P^e = \lbrace{p_1^{e}, p_2^{e}, \ldots, p_n^{e}} \rbrace$. In this paper, we assume that the robots move only along the circumference of $\mathcal C$. For each robot $r_i$, the arc distance to its assigned destination is calculated as $d_{arc}(r_i, p_i^{e})$. A robot is considered \textit{candidate} if its minimum arc path toward its destination is unobstructed, and only such robots are allowed to move at a time, while the others remain stationary. If multiple candidate robots exist, those with the smallest distance to their respective destinations are chosen to move. Due to the symmetry of the configuration, the selected candidates form a mirror pair, consisting of a robot and its symmetric counterpart. If several such mirror pairs exist, the pair containing the robot with the minimum view is selected according to Observation~\ref{obs:asymmetric-ordering}. If symmetry is preserved during the movement, the mirror pair of \textit{candidate} robots converges to their respective destination points. Due to asynchrony and \textit{pending} movements, the configuration may be transformed into an asymmetric configuration, and the execution switches to \textbf{Algorithm}~\ref{alg:AsymU1dMinSumC}. This process is repeated iteratively, and an iteration terminates once a \textit{candidate} robot reaches its destination (i.e., $d_{\mathrm{arc}} = 0$). Since in each iteration at most one mirror pair advances toward its destination points, the algorithm guarantees collision-free convergence in finite time, with all robots eventually occupying their designated points in $\mathcal{P}^e$. The pseudocode of the Algorithm {\it SymU1dMinSumC()} is presented in \textbf{Algorithm}~\ref{alg:SymU1dMinSumC}.

\subsubsection*{4.1.2.2~~Correctness of the Algorithm \it{SymU1dMinSumC()}}
\noindent We prove the correctness of the algorithm \textit{SymU1dMinSumC()} for robot configurations that admit exactly one line of symmetry and have a unique optimal min-sum assignment. Let $\mathcal{R}(t_0)$ be such a configuration with symmetry line $\mathcal{L}$. Under the \textit{ASYNC} model, the algorithm fixes the target points once computed; hence, the total distance to the targets remains optimal throughout the execution. Two robots are selected as candidate robots. As long as their movements preserve symmetry, both identify the same \textit{extremal} robot located at the intersection of $\mathcal{L}$ and $\mathcal{C}_{\text{out}}$, and the target points do not change. If the configuration becomes asymmetric due to pending asynchronous moves, \textbf{Lemma}~\ref{lem:n-robots} ensures that the \textit{extremal} robot remains invariant, and thus the assigned destinations remain unchanged. Therefore, the total distance stays optimal even after symmetry is lost. Finally, \textbf{Lemma}~\ref{lem-correct-symu} shows that under the \textit{ASYNC} model, the algorithm terminates in finite time, avoids collisions, and places all robots at positions achieving the optimal total distance.

\begin{lemma}[Invariance of Unique Optimal Assignment in $\mathcal I_2$]
\label{lem:n-robots}
Let $r_e$ be identified as an \textit{extremal} robot in a configuration admitting a single line of symmetry at time $t_0 \ge 0$. Fixing $r_e$, let $p_i^{e}, p_j^{e} \in \mathcal{C}_{out}$ denote the destinations assigned at time $t_0$ to the candidate robots $r_i$ and $r_j$, respectively. If, at time $t_1$, both $r_i$ and $r_j$ reach intermediate positions (see \textbf{Definition}~\ref{Intermediate} on page~\pageref{Intermediate}) on $\mathcal{C}_{out}$ before reaching their respective destinations, then the \textit{extremal} robot $r_e$ remains invariant throughout the time interval $[t_0,t_1]$.
\end{lemma}

\begin{proof}

Let
\[
\mathcal{R}(t_0) = \{ r_1(t_0), r_2(t_0), \ldots, r_n(t_0) \}
\]
be a configuration of an odd number of $n$ robots (if $n$ is even, see \textbf{Observation}~\ref{obs-two-on-los-2}) that admits a unique line of symmetry $\mathcal{L}$, and let $r_e\in\mathcal R(t_0)$ denote the \textit{extremal} robot such that $
r_e \in \mathcal{L} \cap \mathcal{R}(t_0).
$

\noindent To justify that the optimal assignment remains invariant during the movement of the \textit{candidate} robots, it is necessary to examine how the total cost behaves for all robots in the system. In particular, we must verify that the movement of the \textit{candidate} robot does not affect the optimality of the configuration either with respect to its own assignment or with respect to the assignments of the other stationary robots. Since $\mathcal{R}(t_0)$ is symmetric, the two selected \textit{candidate} robots in $\mathcal{R}(t_0)$ form a mirror pair, consisting of a robot and its mirror image with respect to $\mathcal{L}$. Specifically, they are $r_i$ and $\mu_{\mathcal{L}}(r_i)$, as defined in \textbf{Definition}~\ref{def-mirror-map} (See Page~\pageref{def-mirror-map}). These robots are selected for the initial movement toward their respective destination points determined by fixing the \textit{extremal} robot $r_e$. If the candidate robots $r_i$ and $\mu_{\mathcal{L}}(r_i)$ preserve the symmetry of the configuration during execution the \textbf{Algorithm}~\textit{SymU1dMinSumC()}, then they uniquely determine both the line of symmetry and the \textit{extremal} robot $r_e$ at each instant of time. In this situation, the optimal assignment remains unchanged throughout the execution, and therefore, no separate proof of correctness to verify the invariance of the candidate robot is required. However, due to asynchrony and \textit{pending} movements, the configuration may become asymmetric. In that case the execution switches to \textbf{Algorithm}~\ref{alg:AsymU1dMinSumC}. For this scenario, the proof is divided into two parts. We will establish the following: \textbf{(i) Invariance of the \textit{extremal} robot $r_e$ with respect to the candidate robots $r_i$ and $\mu_{\mathcal{L}}(r_i)$,} and \textbf{(ii) invariance of the \textit{extremal} robot with respect to all other $(n-3)$ stationary robots except $r_e$, } which follows similar to \textbf{Lemma}~\ref{l-3}.

\begin{figure}[!ht]
    \centering
    \vspace{-0.1cm}
    
    \includegraphics[width=0.6\textwidth]{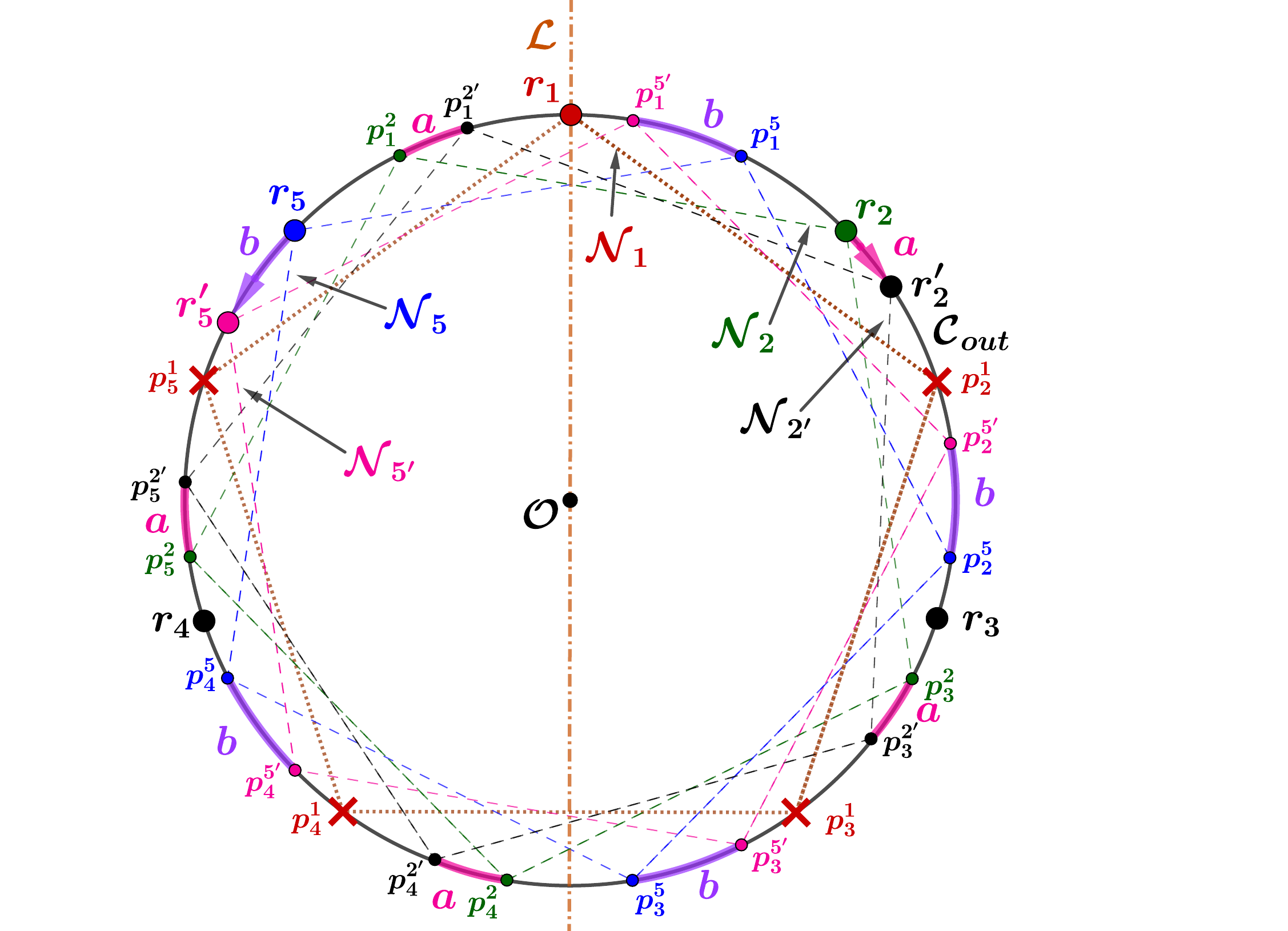}
\caption{\textit{Illustration of \textbf{Lemma}~\ref{lem:n-robots} for a symmetric initial configuration
$\mathcal{R}(t_0)=\{r_1(t_0), r_2(t_0), \ldots, r_5(t_0)\}$ admitting a single line of symmetry $\mathcal{L}$. The mirror images of robots $r_2$ and $r_3$ with respect to $\mathcal{L}$ are $r_5$ and $r_4$, respectively. The symmetric pair $(r_i,\mu_{\mathcal{L}}(r_i))$ corresponds to the candidate robots $r_2$ (green) and $r_5$ (blue). The extremal robot $r_e=r_1$ lies on $\mathcal{L}$ and is shown in red. The destination points of robots $r_2, r_3, r_4,$ and $r_5$ are denoted by $p_2^{1}, p_3^{1}, p_4^{1},$ and $p_5^{1}$, respectively. The candidate robots move along the circumference $\mathcal{C}_{out}$, reaching intermediate positions $r_2'$ and $r_5'$ with arc distances $d_{\mathrm{arc}}(r_2,r_2')=a$ and $d_{\mathrm{arc}}(r_5,r_5')=b$.}}

 \label{Last-proof}

\end{figure}

\begin{itemize}
    \item \textit{\textbf{Invariance of the \textit{extremal} robot $r_e$ with respect to a candidate robots  $r_i$ and $\mu_{\mathcal{L}}(r_i)$:}}
   Let the robots $r_i$ and $\mu_{\mathcal{L}}(r_i)$ move along the circumference of $\mathcal{C}$ toward their respective destinations and reach the positions $r_i'$ and $\mu_{\mathcal{L}}(r_i')$ at time $t_1$, respectively, with arc distances
\[
d_{\mathrm{arc}}\big(r_i,r_i'\big)=a
\quad \text{and} \quad
d_{\mathrm{arc}}\big(\mu_{\mathcal{L}}(r_i),\mu_{\mathcal{L}}(r_i')\big)=b, \quad \text{where $a<b$.}
\]
Without loss of generality, assume that the robot $r_i$ moves in the clockwise direction, and by symmetry, the robot $\mu_{\mathcal{L}}(r_i)$ moves in the counterclockwise direction during execution the \textbf{Algorithm}~\textit{SymU1dMinSumC()} (See \textbf{Figure}~\ref{Last-proof}).

By symmetry, the total arc-distance cost obtained by fixing the two candidate robots $r_i$ and $\mu_{\mathcal{L}}(r_i)$ at time $t_0$ is equal, that is,
\[
\mathcal{D}_{r_i}(t_0) = \mathcal{D}_{\mu_{\mathcal{L}}(r_i)}(t_0).
\]
Moreover, since $r_e$ is \textit{extremal} and lies on $\mathcal{L}$, its initial optimal total arc-distance cost for $r_i$ and $\mu_{\mathcal{L}}(r_i)$ satisfies
\begin{equation}
\mathcal{D}^*_{r_e}(t_0) < \mathcal{D}_{r_i}(t_0)
= \mathcal{D}_{\mu_{\mathcal{L}}(r_i)}(t_0).
\label{eee-16}
\end{equation}

At time $t_1$, only the two candidate robots $r_i$ and $\mu_{\mathcal{L}}(r_i)$ reach the positions $r_i'$ and $\mu_{\mathcal{L}}(r_i')$, respectively, while the remaining $n-2$ robots remain stationary. Consequently, the configuration becomes asymmetric at time $t_1$, and the remaining $n-3$ robots except $r_e,\ r_i$ and $\mu_{\mathcal{L}}(r_i)$ are partitioned into symmetric pairs with respect to $\mathcal{L}$ where $r_e$ lies on $\mathcal L$. After the candidate robots move, at time $t_1$, the cost of $r_e$ decreases by exactly $a+b$, since only its distances to $r_i$ and $r_j$ are affected, while all other distances remain unchanged. Thus,

\begin{equation}
    D_{r_e}(t_1) = D_{r_e}^*(t_0) - a - b.
    \label{eee-17}
\end{equation}

On the other hand, by symmetry, except for the robots $r_e$, $r_i$, and $\mu_{\mathcal{L}}(r_i)$, at time $t_1$, fixing $r_i$, the total arc-distance for all robots can be expressed as a sum over the corresponding partitions, given by

\[
\mathcal{D}_{r_i}(t_1)
=
\boxed{\text{arc-distance of } \mu_{\mathcal{L}}(r_i)}
+
\boxed{\text{arc-distance of } r_e}
+
\boxed{\text{arc-distance of symmetric $\frac{n-3}{2}$ pairs}}
\]
and, upon fixing $\mu_{\mathcal{L}}(r_i)$, the total arc-distance at time $t_1$ can be written as
\[
\mathcal{D}_{\mu_{\mathcal{L}}(r_i)}(t_1)
=
\boxed{\text{arc-distance of } r_i}
+
\boxed{\text{arc-distance of } r_e}
+
\boxed{\text{arc-distance of symmetric $\frac{n-3}{2}$ pairs}}
\]

Observing the configuration after the movement of the candidate robots $r_i$ and $\mu_{\mathcal{L}}(r_i)$, at time $t_1$ we obtain
\begin{equation}
\mathcal{D}_{r_i}(t_1)= \mathcal{D}_{r_i}(t_0)-2a-b,
\label{eee-18}
\end{equation}
and
\begin{equation}
\mathcal{D}_{\mu_{\mathcal{L}}(r_i)}(t_1)= \mathcal{D}_{\mu_{\mathcal{L}}(r_i)}(t_0)-2b-a.
\label{eee-19}
\end{equation}

Our goal is to establish that
\begin{equation}
\mathcal{D}_{r_e}(t_1) < \mathcal{D}_{r_i}(t_1),
\end{equation}
and
\begin{equation}
\mathcal{D}_{r_e}(t_1) < \mathcal{D}_{\mu_{\mathcal{L}}(r_i)}(t_1).
\end{equation}

Assume that
\begin{equation}
\mathcal{D}_{r_i}(t_1) < \mathcal{D}_{r_e}(t_1),
\label{eee-22}
\end{equation}
and
\begin{equation}
\mathcal{D}_{\mu_{\mathcal{L}}(r_i)}(t_1) < \mathcal{D}_{r_e}(t_1).
\label{eee-23}
\end{equation}

Substituting the expressions from~\eqref{eee-17} and~\eqref{eee-18} into the inequality~\eqref{eee-22}, we obtain for $r_i$
\begin{align*}
\mathcal{D}_{r_i}(t_1)
&< \mathcal{D}_{r_e}(t_1) \\
\implies \mathcal{D}_{r_i}(t_0)-2a-b
&< \mathcal{D}^*_{r_e}(t_0)-a-b \\
\implies \mathcal{D}_{r_i}(t_0)-a
&< \mathcal{D}^*_{r_e}(t_0).
\end{align*}
 and substituting the expressions from~\eqref{eee-17} and~\eqref{eee-19} into the inequality~\eqref{eee-23}, we obtain  for $\mu_{\mathcal{L}}(r_i)$

 \begin{align*}
\mathcal{D}_{\mu_{\mathcal{L}}(r_i)}(t_1) 
&< \mathcal{D}_{r_e}(t_1) \\
\implies \mathcal{D}_{\mu_{\mathcal{L}}(r_i)}(t_0)-2b-a
&< \mathcal{D}^*_{r_e}(t_0)-a-b \\
\implies \mathcal{D}_{\mu_{\mathcal{L}}(r_i)}-b
&< \mathcal{D}^*_{r_e}(t_0).
\end{align*}

This implies the existence of an assignment of robot positions, say $\mathscr{R}'$, whose total cost is strictly smaller than $\mathcal{D}^*_{r_e}(t_0)$. If none of the $n$ robots is assigned to its own position at time $t_0$ with respect to $\mathscr{R}'$, this contradicts \textbf{Result}~5. Otherwise, it contradicts equation~\eqref{eee-16}, that is our initial assumption that $\mathcal{D}^*_{r_e}(t_0)$ is the optimal assignment cost.

\noindent Hence, the \textit{extremal} robot remains optimal throughout the interval $[t_0,t_1]$ for the candidate robots $r_i$ and $\mu_{\mathcal{L}}(r_i)$.

\item \textbf{Invariance of the \textit{extremal} robot with respect to all other $(n-3)$ stationary robots except $r_e$:} Let $r_k(t_0) \in \mathcal{R}(t_0)\setminus\{r_e, r_i, \mu_{\mathcal{L}}(r_i)\}$ denote another robot at time $t_0$. Then we have 
\begin{equation}
\mathcal{D}^*_{r_e}(t_0) < \mathcal{D}_{r_k}(t_0).
\label{eee-24}
\end{equation}

Let at time $t_1$, only the two candidate robots $r_i$ and $\mu_{\mathcal{L}}(r_i)$ reach the positions $r_i'$ and $\mu_{\mathcal{L}}(r_i')$, respectively, while the remaining $n-2$ robots remain stationary. All possible updated costs for a robot $r_k$ at time $t_1$ are given by one of the following:

\[
\mathcal{D}_{r_k}(t_0)-a-b,\quad
\mathcal{D}_{r_k}(t_0)-a+b,\quad
\mathcal{D}_{r_k}(t_0)+a-b\quad
\text{and}\quad
\mathcal{D}_{r_k}(t_0)+a+b.
\] 

Let \[
\mathcal{D}'_{r_k}(t_1)
=
\min\Big\{
\mathcal{D}_{r_k}(t_0)-a-b,\;
\mathcal{D}_{r_k}(t_0)-a+b,\;
\mathcal{D}_{r_k}(t_0)+a-b,\;
\mathcal{D}_{r_k}(t_0)+a+b
\Big\}.
\]

This implies
\begin{equation}
\mathcal{D}'_{r_k}(t_1)=\mathcal{D}_{r_k}(t_0)-a-b.
\label{eee-25}
\end{equation}

Subtracting $a+b$ from both sides of equation~\eqref{eee-24}, we obtain

\begin{equation}
\mathcal{D}^*_{r_e}(t_0)-a-b < \mathcal{D}_{r_k}(t_0)-a-b.
\label{eee-26}
\end{equation}

Combining~\eqref{eee-26} with~\eqref{eee-17}, we conclude that
\begin{equation}
\mathcal{D}_{r_e}(t_1) < \mathcal{D}'_{r_k}(t_1).
\label{eee-27}
\end{equation}

Hence, the \textit{extremal} robot remains optimal throughout the interval $[t_0,t_1]$ with respect to the stationary robots $r_k$.
\end{itemize}
 This completes the lemma.

\end{proof}

\noindent The following observation shows that the argument of \textbf{Lemma}~\ref{lem:n-robots} extends naturally to symmetric configurations with an even number of robots.

\begin{observation}
\label{obs-two-on-los-2}
Consider a configuration of an even number of $n$ robots in the set $\mathcal{I}_2$ that admits a single line of symmetry $\mathcal{L}$. If the intersection of $\mathcal{L}$ with the circumference $\mathcal{C}_{out}$ contains two extremal robot positions, then the proof follows analogously to that of \textbf{Lemma}~\ref{lem:n-robots}.
\end{observation}

\begin{lemma}
\label{lem-correct-symu}
Let $\mathcal{R}(t_0)$ be a robot configuration admitting exactly one line of symmetry $\mathcal L$ and a unique optimal assignment. Under the \textit{ASYNC} model, the algorithm \textit{SymU1dMinSumC()} terminates in finite time and positions all the robots at the distinct vertices of a regular $n$-gon, while ensuring collision-free motion and minimizing the total sum of arc distances.
\end{lemma}

\begin{proof}\normalfont
Let $\mathcal{R}(t_0)$ be an initial robot configuration that admits a single line of symmetry $\mathcal L$ and a unique optimal assignment. By \textbf{Lemma}~\ref{l-6}, the robot $r_e$ located at the intersection of $\mathcal L$ with the circumference $\mathcal C_{out}$ must be the unique \textit{extremal} robot. By fixing $r_e$, the destination point set $\mathcal P^e={p_1^{e},p_2^{e},\ldots,p_n^{e}}$ becomes uniquely determined, according to \textbf{Result}~\ref{r-2}. Since the optimal assignment is unique, this destination set remains fixed throughout the execution of the algorithm.

\noindent The algorithm allows movement only for \textit{candidate} robots, namely the robots that have a free arc path to their assigned destination. Any candidate robot that moves strictly towards its destination, decreases its arc distance to the destination, while robots whose paths are obstructed remain stationary. Since robots move exclusively along free arc paths, the collisions are avoided. As the configuration admits a single line of symmetry, the candidate robots are selected in mirror pairs with respect to $\mathcal L$. If both robots of a mirror pair move concurrently, the symmetry of the configuration is preserved, and both robots reduce
the same arc distance toward their respective destinations. If, due to asynchrony and pending movements, only one robot of a mirror pair moves, the symmetry may be broken, and the configuration becomes asymmetric. In this case, the execution switches to \textbf{Algorithm}~\ref{alg:AsymU1dMinSumC}, which is known to guarantee progress toward the same unique optimal assignment by \textbf{Lemma} \ref{l-3}.

\noindent In every execution step, at least one robot strictly decreases its arc distance to its destination, and no robot ever increases its distance. Since all arc distances are finite and monotonically decreasing, every robot reaches its assigned destination in finite time under the \textit{ASYNC} model. In the final configuration, all robots occupy distinct vertices of the regular $n$-gon defined by $\mathcal P^e$. As this destination set corresponds to the unique optimal assignment, the total sum of arc distances traveled by all robots is minimized. Therefore, the \textbf{Algorithm} \textit{SymU1dMinSumC()} terminates in finite time, avoids collisions, and correctly solves the min-sum uniform circle
formation problem.
\end{proof}

\section*{4.2~~Configurations with Multiple Optimal Assignments}
\label{Configurations with Multiple Optimal Assignments}
\noindent This section addresses all those robot configurations that admit multiple optimal
assignments, that is, configurations for which $|\mathcal{E}'(t)| > 1$. This multiplicity arises from inherent ambiguities in the configuration, such as the presence of multiple \textit{extremal} robots or symmetry, which allow different assignments to yield the same minimum total movement cost. Unlike configurations with a unique optimal assignment, no single assignment can be selected deterministically without additional structural considerations.
Using the classes of partitions and the assignment tree introduced in \textbf{Section}~\ref{Oncircle}, (Page No.~\pageref{Oncircle}), we systematically analyze these configurations and propose algorithmic strategies for the solvable cases.

\subsection*{4.2.1~~Impossibility Result}
\label{Impossibility Result}

\noindent This section establishes impossibility results for robot configurations that admit multiple optimal assignments, for which no deterministic distributed algorithm can guarantee convergence to a unique target configuration. In particular, we focus on symmetric configurations with multiple optimal assignments, where the inherent symmetry prevents the robots from consistently
breaking ties among equivalent optimal solutions.

\subsubsection*{4.2.1.1~~Symmetric Configurations $\mathcal I_3$: Multiple Optimal Assignments}

\noindent This subsection addresses symmetric robot configurations that admit multiple optimal assignments. Unlike the configurations considered in Subsections~4.1.1 and~4.1.2, the presence of symmetry combined with multiple assignments prevents the robots from deterministically selecting a unique target configuration.
Due to symmetry, the robots see multiple choices as identical, so the uncertainty remains and cannot be eliminated with the assumed robot model. We show that, for such configurations, no deterministic distributed algorithm can guarantee convergence to a unique target configuration.


\begin{figure}[!ht]
    \centering
    \vspace{-0.3cm}
    
    \includegraphics[width=0.5\textwidth]{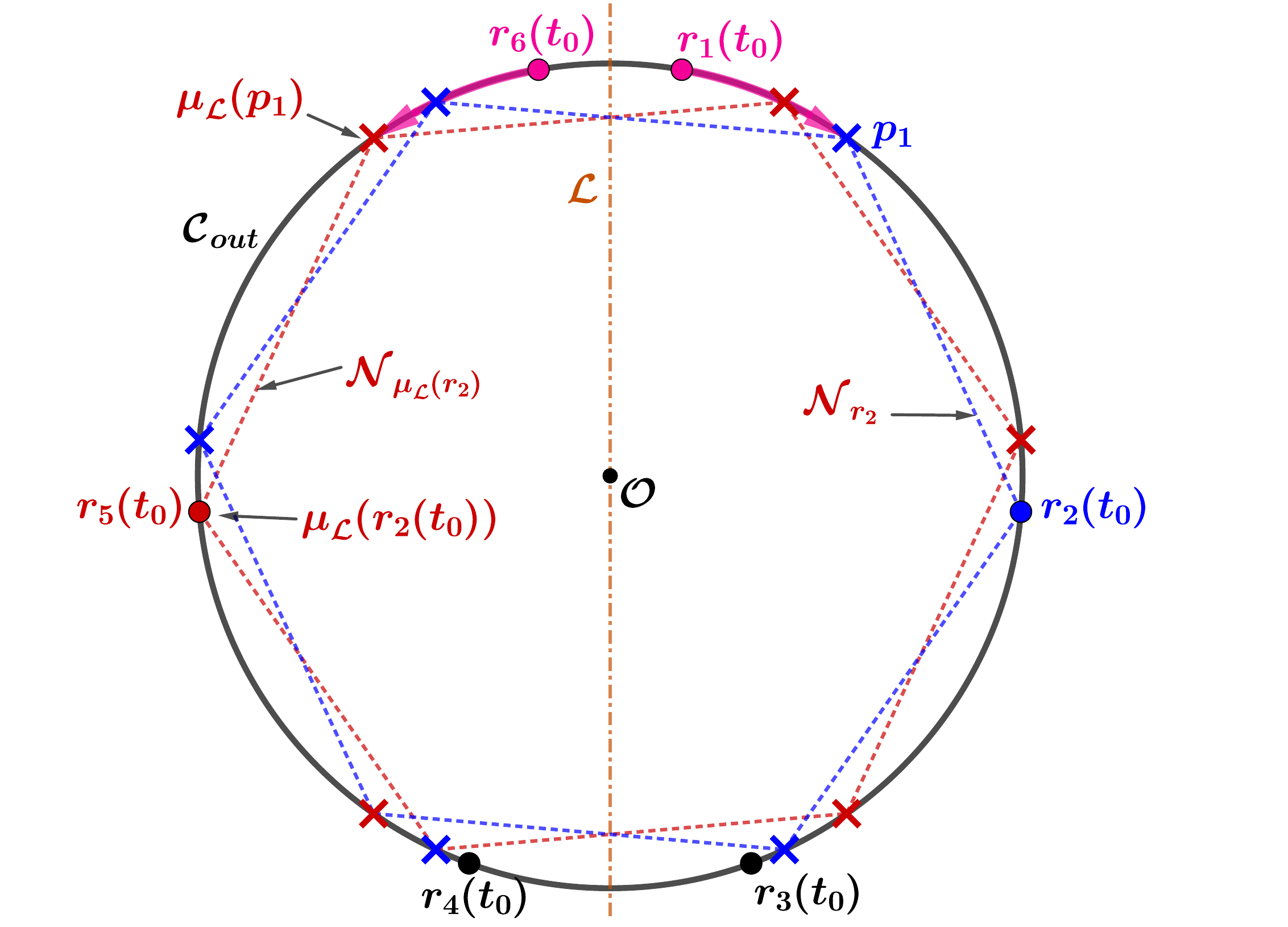}
\caption{\textit{Illustration of the impossibility result for a symmetric initial configuration
$\mathcal{R}(t_0)=\{r_1(t_0), r_2(t_0), \ldots, r_6(t_0)\}$ admitting a single line of symmetry $\mathcal L$. The symmetric pair $(r_j,\mu_{\mathcal L}(r_j))$ corresponds to $r_2(t_0)$ and $\mu_{\mathcal L}(r_2(t_0))$. The \textit{extremal} robots are $r_2(t_0)$ and $\mu_{\mathcal L}(r_2(t_0))$, shown as blue and red dots. The mirror images of $r_1(t_0)$, $r_2(t_0)$, and $r_3(t_0)$ with respect to $\mathcal L$ are $r_6(t_0)$, $r_5(t_0)$, and $r_4(t_0)$, respectively. Here, $r_i=r_1(t_0)$ and $\mu_{\mathcal L}(r_i)=\mu_{\mathcal L}(r_6(t_0))$. The target vertex $p_i$ of $\mathcal N_{r_i}$ for $r_i$ is represented by $p_1$ of $\mathcal N_{r_2}$ for $r_1(t_0)$.}}

 \label{Impossible}

\end{figure}


 \begin{lemma}
\label{lemma-8}
 If the initial configuration admits a single line of symmetry  $\mathcal{L}$ and $\mathcal{L}\cap\mathcal R(t_0)=\emptyset$, then the  \textit{min-sum uniform coverage} on a circle problem is deterministically unsolvable by {\it oblivious} and {\it silent} robots under asynchronous scheduler.
\end{lemma}

\begin{proof}\normalfont
Consider the initial configuration
$\mathcal R(t_0)=\{r_1(t_0), r_2(t_0), \ldots, r_6(t_0)\}$ shown in \textbf{Figure}~\ref{Impossible}, which admits a single line of symmetry
$\mathcal L$. Let $\mathscr{A}\in \widetilde{\mathscr A}^*$ be a deterministic algorithm that solves the  \textit{min-sum uniform coverage} on a circle problem. Suppose the robots operate under \textit{SSYNC} scheduler with rigid motion with the additional assumption that $r_i$ and $\mu_{\mathcal L}(r_i)$ are activated simultaneously. Since robots are homogeneous, they execute the same deterministic algorithm. By \textbf{Result}~\ref{symm}, it follows that the views of a robot $r_i$ and $\mu_{\mathcal L}(r_i)$ are exactly same. As the initial configuration was symmetric, the robots would not be able to deterministically break the symmetry in this setting. Now, if a robot $r_i$ identifies the optimal assignment with respect to a robot $r_j$, then $\mu_{\mathcal L}(r_i)$ will identify the assignment with respect to $\mu_{\mathcal L}(r_j)$ to be optimal due to symmetric view. Thus, if $r_i$ moves towards a vertex $p_i$ of $\mathcal{N}_{r_j}$, then $\mu_{\mathcal L}(r_i)$ will also move towards the vertex $\mu_{\mathcal L}(p_i)$ of $\mathcal{N}_{\mu_{\mathcal L}(r_j)}$. As robots exhibit rigid motion, $r_i$ and $\mu_{\mathcal L}(r_i)$ will reach $p_i$ and $\mu_{\mathcal L}(p_i)$, respectively both of which are vertices of distinct regular $n$-gons. To solve \textit{min-sum uniform coverage} on a circle, one of the robots between $r_i$ and $\mu_{\mathcal L}(r_i)$ must move to form a unique regular $n$-gon. Now, if one them decides to move, the adversary would force the other to move, traveling a greater distance than the distance it should travel for optimal assignment, which is a contradiction. Hence, the \textit{min-sum uniform coverage} on a circle problem by {\it oblivious} and {\it silent} robots under an asynchronous scheduler is deterministically unsolvable.   
\end{proof}

\noindent Note that the above argument shows that if an initial configuration admits a single line of symmetry and the \textit{extremal} robots do not lie on the line of symmetry, then uniform circle formation cannot be achieved. The extension of this argument shows that when a configuration belonging to the set $\mathcal I_6$ admits multiple lines of symmetry together with rotational symmetry, and no \textit{extremal} robot lies on any line of symmetry, the ambiguity among optimal assignments cannot be removed by \textit{oblivious} and \textit{silent} robots. Consequently, the \textit{min-sum uniform coverage} on a circle problem is deterministically
unsolvable under an asynchronous scheduler in such configurations.

\subsection*{4.2.2~~Asymmetric Configurations $\mathcal I''_1$: Multiple Optimal Assignments}

\noindent This subsection considers all asymmetric robot configurations $\mathcal I''_1 \in \mathcal I_1$ that admit multiple optimal assignments,  i.e., $|\mathcal{E}'(t)| > 1$. In the absence of symmetry, multiple optimal assignments arise from different choices of \textit{extremal} robots, which may induce the same regular $n$-gon or distinct regular $n$-gons while achieving the same minimum total movement cost. Although the configuration is asymmetric, the presence of multiple optimal assignments introduces ambiguity in the assignment process.
We show how this ambiguity can be systematically resolved using the classes of
partitions and the assignment tree introduced in Section~2.

\begin{algorithm}
\caption{\textit{: AsymM1dMinSumC()}}
\label{alg:AsymM1dMinSumC}
\begin{algorithmic}[1]

\Require Circle $\mathcal{C}$; initial robot positions
$r_1(t_0), r_2(t_0), \ldots, r_n(t_0)$ on the circumference of $\mathcal{C}$
\Ensure Final placement of robots on target configuration
$\mathcal{P} = \{p_1^e, p_2^e, \ldots, p_n^e\}$

\State Compute all \textit{optimal assignments} using \textbf{Result}~\ref{r-2}
\State Identify the set of \textit{\textit{extremal} robots} $\mathcal{E}'$
\State Select one \textit{extremal} robot $r_e \in \mathcal{E}'$ using ordering $\mathscr{O}$
\State Assign destinations $\mathcal P^e=\{p_1^{e}, p_2^{e}, \ldots, p_n^{e}\}$ according to the chosen assignment
\State Compute $d_i = d_{arc}(r_i, p_i^e)$ for each robot $r_i$

\While{not all robots have reached their destinations}
    \For{$i = 1$ to $n$}
        \If{$d_i = 0$}
            \State Mark $r_i$ as \textit{terminated}
        \Else
            \If{$r_i$ has a free path towards $p_i^e$ with minimum distance $d_i$ and minimum order w.r.t $\mathscr{O}$}
                \State Move $r_i \rightarrow p_i^e$ along the arc
                \State Update $d_i$
            \Else
                \State$r_i$ waits
            \EndIf
        \EndIf
    \EndFor
\EndWhile

\end{algorithmic}
\end{algorithm}

\subsubsection*{4.2.2.1~~Overview of the Algorithm \it{AsymM1dMinSumC()}}
\noindent The algorithm \textit{AsymM1dMinSumC()} addresses the \textit{min-sum uniform coverage} on a circle problem for robots initially placed on the circumference of a circle~$\mathcal C$ in cases where multiple optimal assignments exist. The procedure first determines all optimal assignments using \textbf{Result}~\ref{r-2}, which specifies the set of \textit{extremal} robots $\mathcal E$. Then, using \textbf{Observation}~\ref{obs:asymmetric-ordering}, one \textit{extremal} robot $r_e \in \mathcal E$ is chosen as the robot with the minimum view (From \textbf{Observation}~\ref{obs:asymmetric-ordering}) among all such \textit{extremal} robots. The corresponding destination set is fixed as $\mathcal P^e = \{p_1^{e}, p_2^{e}, \ldots, p_n^{e}\}$. In this work, the robots are restricted to move only along the circumference of the circle~$\mathcal C$. Each robot $r_i$ computes its arc distance to the assigned destination point $p_i^{e}$ as $d_{arc}(r_i, p_i^{e})$. A robot is considered \textit{candidate} if its arc path to the destination is unobstructed and corresponds to the minimum arc length; only such \textit{candidate} robots are allowed to move in each iteration, while the others remain stationary. If more than one such \textit{candidate} robot exists, the candidate robot with the minimum arc distance towards its destination and with the minimum order with respect to the \textbf{Observation}~\ref{obs:asymmetric-ordering} in case of ties is selected to move towards its destination. Once a \textit{candidate} robot reaches its destination (i.e., $d_{arc}= 0$), the current iteration terminates, and the procedure repeats accordingly. Since exactly one \textit{candidate} robot progresses toward its destination in every iteration, the algorithm ensures collision-less movement of the robots and finite-time convergence. Eventually, all the robots occupy their designated positions in the destination configuration~$\mathcal P^e$. The pseudo-code corresponding to the procedure is given in the \textbf{Algorithm}~\ref{alg:AsymM1dMinSumC}.

\subsubsection*{4.2.2.2~~Correctness of the Algorithm \it{AsymM1dMinSumC()} }
\noindent We now establish the correctness of the algorithm
\textit{AsymM1dMinSumC()} for asymmetric configurations and  $|\mathcal{E}'(t)| > 1$. In particular, we show that even when multiple optimal assignments exist, fixing a deterministic ordering among the robots is sufficient to ensure a unique and consistent assignment.

\begin{lemma}[Invariance of Multiple Optimal Assignments in $\mathcal I_1''$]
    If an asymmetric robot configuration $\mathcal{R}(t_0)$ admits multiple optimal assignments, then using \textbf{Observation}~\ref{obs:asymmetric-ordering} for $\mathcal R(t_0)$, the unique optimal assignment of robots to the destination set remains invariant with respect to the position of the robot $r_i\in\mathcal E(t_0)$,  selected according to the \textbf{Observation}~\ref{obs:asymmetric-ordering}.
    \label{l-8}
\end{lemma}

\begin{figure}[!ht]
    \centering
    \vspace{-0.3cm}
    
    \includegraphics[width=0.5\textwidth]{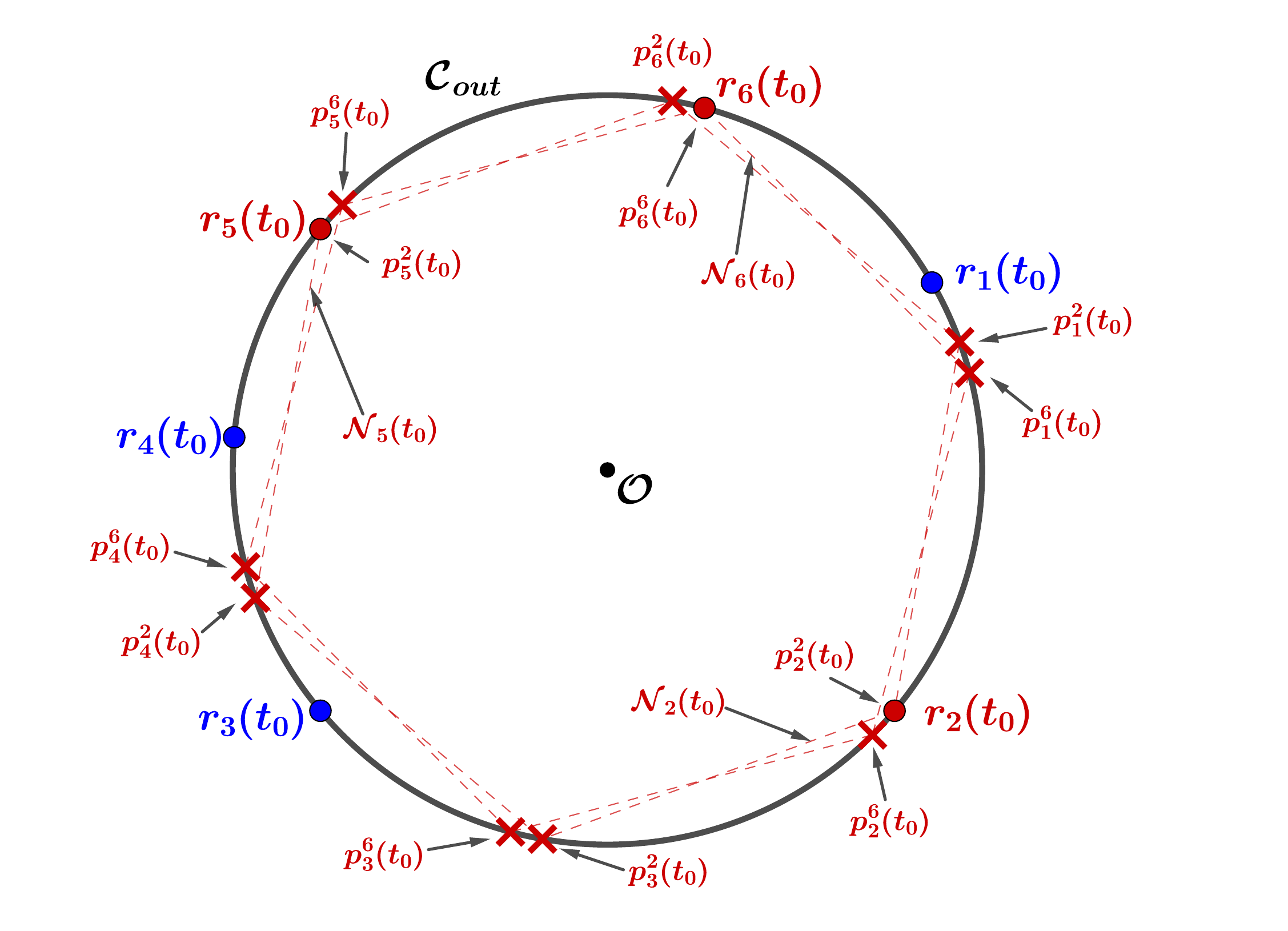}
    \caption{ \textit{An illustration of the initial configuration of six robots $\mathcal{R}(t_0)=\{r_1(t_0),r_2(t_0),\ldots,r_6(t_0)\}$ with the \textit{extremal} robot set $\mathcal E'(t_0)=\{r_2(t_0),r_5(t_0),r_6(t_0)\}$ (red dots). Robot $r_5(t_0)$ lies on a vertex of the regular $n$-gon $\mathcal N_2(t_0)$ determined by robot $r_2(t_0)$. The destination sets determined by fixing the \textit{extremal} robots $r_2(t_0)$ and $r_6(t_0)$ are $\mathcal P^2(t_0)=\{p_1^2(t_0),p_2^2(t_0),\ldots,p_6^2(t_0)\}$ and $\mathcal P^6(t_0)=\{p_1^6(t_0),p_2^6(t_0),\ldots,p_6^6(t_0)\}$ (red crosses), respectively.}}

     \label{Lem-8(1)}

\end{figure}

\begin{proof}
\normalfont Let  
\[
\mathcal{R}(t_0) = \{ r_1(t_0), r_2(t_0), \ldots, r_n(t_0) \}
\]  
denote the asymmetric configuration of robots at time $t_0$, and let the corresponding set of \textit{extremal} robots be  
\[
\mathcal{E}'(t_0) = \{ r_1(t_0), r_2(t_0), \ldots, r_k(t_0) \}, \quad k \leq n.
\]  
According to \textbf{Result}~\ref{r-2}, each \textit{extremal} robot generates a distinct regular co-circular polygon on the circle $\mathcal{C}$ if no others \textit{extremal} robots lie on its vertex of regular polygon. Specifically, for the \textit{extremal} robots  
$r_1(t_0), r_2(t_0), \ldots, r_k(t_0)$, we obtain $k$ such polygons, denoted by  
\[
\mathcal N_1(t_0),\ \mathcal N_2(t_0),\ \ldots,\ \mathcal N_k(t_0).
\]

\noindent Each polygon $\mathcal N_i(t_0)$, associated with \textit{extremal} robot $r_i(t_0) \in \mathcal{E}'(t_0)$, consists of $n$ vertices positioned on the circumference of $\mathcal{C}$. These vertices serve as destination points for the robots in $\mathcal{R}(t_0)$. Let the vertex set of $\mathcal N_i(t_0)$ be expressed as
\[
\{ p_1^{i}(t_0),\ p_2^{i}(t_0),\ \ldots,\ p_n^{i}(t_0) \}
\quad \text{(see \textbf{Figure}~\ref{Lem-8(1)})}.
\]
 
\noindent
Corresponding to the $k$ \textit{extremal} robots, we obtain $k$ distinct collections of arc-distance values. Each collection is induced by a regular co-circular $n$-gon constructed by fixing one \textit{extremal} robot as a vertex. For each $i = 1,2,\ldots,k$, we represent the arc distances between the robots and their respective destination points associated with the $i$-th regular co-circular polygon in matrix form.

\medskip

\noindent
Formally, for each $i = 1,2,\ldots,k$, define
\[
P_{r_i}(t_0)
=
\begin{bmatrix}
d_{i1}(t_0) & d_{i2}(t_0) & \cdots & d_{in}(t_0)
\end{bmatrix},
\]
where each entry is given by
\[
d_{ij}(t_0)
=
d_{\text{arc}}\!\big(r_j(t_0),\, p_j^{\,i}(t_0)\big),
\qquad
j = 1,2,\ldots,n.
\]

\noindent In particular, we have $d_{ii}(t_0)=0$, since the $i$-th robot already occupies
its corresponding target position, i.e., $r_i(t_0)=p_i^{\,i}(t_0)$, for $i=1, 2, \ldots, k$.

\medskip

\noindent
By collecting all such row vectors, we define the optimal distance matrix
\[
\mathcal{D}^*(t_0)
=
\begin{bmatrix}
P_{r_1}(t_0)\\
P_{r_2}(t_0)\\
\vdots\\
P_{r_k}(t_0)
\end{bmatrix}
=
\big[d_{ij}(t_0)\big]_{k \times n}
=
\left[
\begin{array}{cccc|ccc}
d_{11}(t_0) & d_{12}(t_0) & \cdots & d_{1k}(t_0)
& d_{1,k+1}(t_0) & \cdots & d_{1n}(t_0)\\
d_{21}(t_0) & d_{22}(t_0) & \cdots & d_{2k}(t_0)
& d_{2,k+1}(t_0) & \cdots & d_{2n}(t_0)\\
\vdots      & \vdots      & \ddots & \vdots
& \vdots      & \ddots & \vdots\\
d_{k1}(t_0) & d_{k2}(t_0) & \cdots & d_{kk}(t_0)
& d_{k,k+1}(t_0) & \cdots & d_{kn}(t_0)\\[8pt]
\multicolumn{4}{c}{\underbrace{\hspace{4.2cm}}_{k\times k}}
&
\multicolumn{3}{c}{\underbrace{\hspace{3.0cm}}_{k\times (n-k)}}
\end{array}
\right].
\]

\medskip

\noindent
Equivalently, the distance matrix can be written explicitly as $\mathcal{D}^*(t_0)=$
\vspace{0.5cm}

\resizebox{\dimexpr\paperwidth-3.4cm\relax}{!}{$
\left[
\begin{array}{cccc|ccc}
0 &
d_{arc}(r_2(t_0),p_2^{1}(t_0)) &
\cdots &
d_{arc}(r_k(t_0),p_k^{1}(t_0)) &
d_{arc}(r_{k+1}(t_0),p_{k+1}^{1}(t_0)) &
\cdots &
d_{arc}(r_n(t_0),p_n^{1}(t_0))\\
d_{arc}(r_1(t_0),p_1^{2}(t_0)) &
0 &
\cdots &
d_{arc}(r_k(t_0),p_k^{2}(t_0)) &
d_{arc}(r_{k+1}(t_0),p_{k+1}^{2}(t_0)) &
\cdots &
d_{arc}(r_n(t_0),p_n^{2}(t_0))\\
\vdots & \vdots & \ddots & \vdots & \vdots & \ddots & \vdots\\
d_{arc}(r_1(t_0),p_1^{k}(t_0)) &
d_{arc}(r_2(t_0),p_2^{k}(t_0)) &
\cdots &
0 &
d_{arc}(r_{k+1}(t_0),p_{k+1}^{k}(t_0)) &
\cdots &
d_{arc}(r_n(t_0),p_n^{k}(t_0))
\end{array}
\right]
$}

\vspace{-0.7cm}

\[
\hspace{1.3cm}\underbrace{\hspace{0.46\textwidth}}_{k\times k}
\hspace{0.9cm}
\underbrace{\hspace{0.35\textwidth}}_{k\times(n-k)}
\]

\noindent  Define 

\[
\mathfrak{d}_{ij}(t_0)
=
\begin{cases}
d_{arc}\!\big(r_j(t_0),p_j^{\,i}(t_0)\big),
& \text{if a free path exists for } d_{ij}(t_0),\\
\infty, & \text{otherwise}.
\end{cases}
\]

\[
\mathfrak{d}_i(t_0)
=
\min_{\substack{j=1,\ldots,n\\ j\neq i}}
\mathfrak{d}_{ij}(t_0)
\quad
\text{(See \textbf{Figure}~\ref{Lem-8(2)})}.
\]

\noindent The value of $\mathfrak{d}_{i}(t_0)$ may be either unique or may occur for multiple such $j's$. If
it is unique, then it corresponds to a single \textit{extremal} robot $r_i(t_0)$. But, if it attains multiple times, then it corresponds either to one \textit{extremal} robot or more than one \textit{extremal} robot.

\begin{figure}[!ht]
    \centering
    \vspace{-0.3cm}
    
    \includegraphics[width=0.5\textwidth]{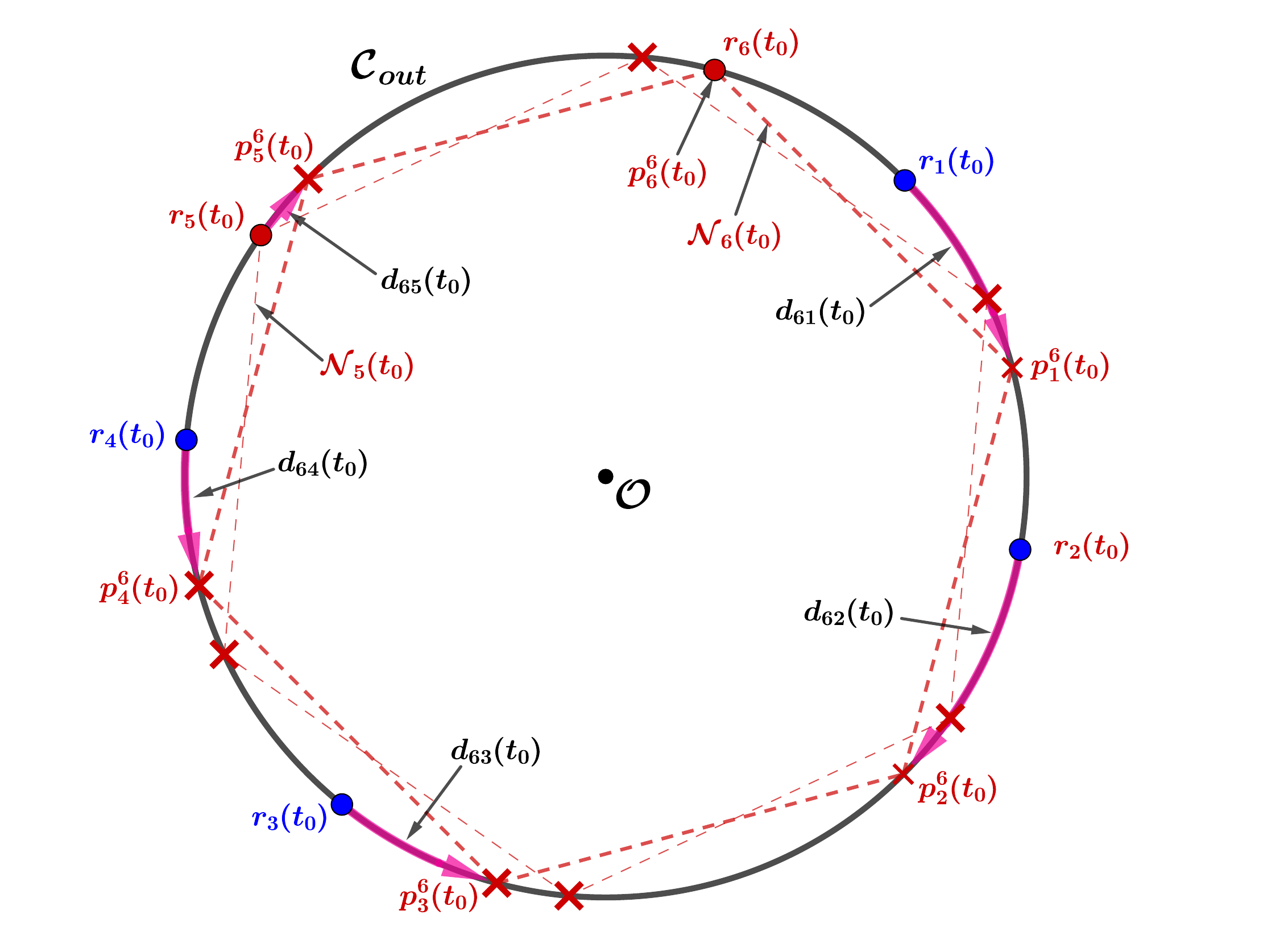}
   \caption{ \textit{An illustration of the initial configuration of six robots $\mathcal{R}(t_0)=\{r_1(t_0),r_2(t_0),\ldots,r_6(t_0)\}$ with the \textit{extremal} robot set $\mathcal E'(t_0)=\{r_5(t_0),r_6(t_0)\}$. Using \textbf{Observation}~\ref{obs:asymmetric-ordering} on $\mathcal{R}(t_0)$, the \textit{extremal} robot $r_6(t_0)$ (red dot) is selected. The destination set determined by fixing $r_6(t_0)$ is $\mathcal P^6(t_0)=\{p_1^6(t_0),p_2^6(t_0),\ldots,p_6^6(t_0)\}$ (red crosses). The arc paths $d_{61}(t_0), d_{62}(t_0), d_{63}(t_0),d_{64}(t_0)$ and $d_{65}(t_0)$ (pink arcs) are free. Hence, $\mathfrak{d}(t_0)=\min\{d_{61}(t_0),\, d_{62}(t_0),\, d_{63}(t_0),\, d_{64}(t_0),\, d_{65}(t_0)\}=d_{65}(t_0)$, and the \textit{candidate} robot is $r_5(t_0)$.}}

     \label{Lem-8(2)}

\end{figure}

\begin{itemize}
    \item[\textbf{Case-I:}] \noindent Suppose that the value of $\mathfrak{d}_i(t_0)$ is unique for a single \textit{extremal} robot, say $r_i(t_0)\in \mathcal E(t_0)$. Let $r_c(t_0)\in\mathcal R(t_0)$ be the \textit{candidate} robot such that 
    \[
    \mathfrak{d}_i(t_0)=d_{ic}(t_0)=d_{arc}\Big(r_c(t_0),p_c^{i}(t_0)\Big).
    \]
    In this situation, the \textit{candidate} robot $r_c(t_0)$ may be either an \textit{extremal} or a \textit{non-\textit{extremal}} robot.

    \begin{itemize}
        \item \textbf{Subcase-IA:} Consider the \textit{extremal} robot $r_m(t_0)$, and suppose that the \textit{candidate} robot $r_c(t_0)$ is also an \textit{extremal} robot for which $\mathfrak{d}_i(t_0)$ is unique. Then
        \begin{equation}
            \mathfrak{d}_i(t_0)=d_{mc}(t_0)=d_{arc}\Big(r_c(t_0),p_c^{m}(t_0)\Big).
            \label{eq2}
        \end{equation}
        Suppose that at time $t_1>t_0$, the robot $r_c(t_0)$ reaches a point $p_c(t_1)$ on the circle $\mathcal C$, before arriving at $p_c^{m}(t_0)$. Define $x=d_{arc}\Big(r_c(t_0),p_c(t_1)\Big)>0$. Then
        \begin{equation}
            d_{mc}(t_1)=d_{mc}(t_0)-x=\mathfrak{d}_i(t_1)\ \text{(say)}.
            \label{eq3}
        \end{equation}
        
        Since $r_c(t_0)$ is an \textit{extremal} robot, we have
\[
P_{r_c}(t_0)
=
\begin{bmatrix}
d_{c1}(t_0) & d_{c2}(t_0) & \cdots & d_{cn}(t_0)
\end{bmatrix}.
\]

  Let 

        \[
\mathfrak{d}_{cj}(t_0)
=
\begin{cases}
d_{arc}\!\big(r_j(t_0),p_j^{\,c}(t_0)\big),
& \text{if a free path exists for } d_{cj}(t_0),\\
\infty, & \text{otherwise}.
\end{cases}
\]

\[
\mathfrak{d}_c(t_0)
=
\min_{\substack{j=1,\ldots,n\\ j\neq c}}
\mathfrak{d}_{cj}(t_0).
\]

        By assumption, $\mathfrak{d}_i(t_0)$ is unique, hence
        \begin{equation}
            \mathfrak{d}_i(t_0)<\mathfrak{d}_c(t_0).
            \label{eq4}
        \end{equation}
        
        At time $t_1$, the robot $r_c$ is positioned at $p_c(t_1)$, and the distance set becomes
\[
P_{r_c}(t_1)
=
\begin{bmatrix}
d_{c1}(t_1) & d_{c2}(t_1) & \cdots & d_{cn}(t_1)
\end{bmatrix}.
\]

Then, for $j=1,2,\ldots,n$, the following relations hold, depending on whether the motion is clockwise or counterclockwise:
\begin{equation}
            d_{cj}(t_1)=\mathfrak{d}_c(t_0)-x, \qquad j\in \{1,2,\ldots,n\},
            \label{eq5}
\end{equation}
\begin{equation}
            d_{ch}(t_1)=\mathfrak{d}_c(t_0)+x, \qquad h\in \{1,2,\ldots,n\}\setminus\{j\}.
            \label{eq6}
        \end{equation}

\noindent
At time $t_1$, we define the free-path arc distance associated with robot $r_c$ as
\[
\mathfrak{d}_{cm}(t_1)
=
\begin{cases}
d_{cm}(t_1),
& \text{if } d_{cm}(t_1)\text{ admits a free path},\\[4pt]
\infty, & \text{otherwise},
\end{cases}
\qquad m=1,2,\ldots,n.
\]
Accordingly, the minimum free-path arc distance of robot $r_c$ at time $t_1$ is
\[
\mathfrak{d}_c(t_1)
=
\min_{\substack{m=1,\ldots,n\\ m\neq c}}
\mathfrak{d}_{cj}(t_1)
=
\min_{m=j\cup h}
\{d_{cj}(t_1),\, d_{ch}(t_1)\}
\]


        

        which implies
        \begin{equation}
            \mathfrak{d}_c(t_1)=\mathfrak{d}_c(t_0)-x.
            \label{eq7}
        \end{equation}
        
        Our claim is that
        \begin{equation}
            d_{ic}(t_1)=\mathfrak{d}_i(t_1)<\mathfrak{d}_c(t_1).
            \label{eq8}
        \end{equation}
        We will prove this result using the method of contradiction. Thus, we assume that $\mathfrak{d}_i(t_1)>\mathfrak{d}_c(t_1)$. Since $x>0$, it follows that
        \begin{equation}
           \mathfrak{d}_i(t_1)+x>\mathfrak{d}_c(t_1)+x.
           \label{eq9}
        \end{equation}
        Substituting from \eqref{eq3} and \eqref{eq7}, equation \eqref{eq9} yields
        \begin{equation}
           \mathfrak{d}_i(t_0)>\mathfrak{d}_c(t_0).
           \label{eq10}
        \end{equation}
        which contradicts \eqref{eq4}. Hence, $\mathfrak{d}_i(t_1)<\mathfrak{d}_c(t_1)$, as claimed. 
        
        Since only the \textit{candidate} robot $r_c(t_0)\in \mathcal{E}'(t_0)$ moves towards the point $p_c(t_1)$ relative to the \textit{extremal} robot $r_m$, while all other $(n-1)$ robots remain stationary, it follows that at time $t_1$ the pair $(r_c$,  $r_m)$ is still uniquely selected. This selection continues to hold for each of the remaining $(k-1)$ \textit{extremal} robots and $(n-(k-2))$ \textit{non-\textit{extremal}} robots, until the candidate robot $r_c$ eventually reaches its destination position $p_c^{m}(t_0)$.

    \end{itemize}
    \begin{itemize}
    \item \textbf{Subcase-IB:} Let $r_m(t_0)$ be an \textit{extremal} robot and $r_c(t_0)$ be a candidate robot which is \textit{non-\textit{extremal}}, i.e., 
    \[
        r_c(t_0)\in \mathcal R(t_0)\setminus \mathcal E(t_0),
    \] and 
    for which $\mathfrak{d}_i(t_0)$ is uniquely defined. Then,  
    \begin{equation}
        \mathfrak{d}_i(t_0)=d_{mc}(t_0)=d_{arc}\Big(r_c(t_0),p_c^{m}(t_0)\Big).
        \label{eq12}
    \end{equation}

    Initially, $r_c(t_0)$ is not an \textit{extremal} robot. From \textbf{Lemma}~\ref{l-3}, it follows that at any time instant $t_1$, robot $r_c$ at position $p_c(t_1)$ will also not become an \textit{extremal} robot. Hence, the distance 
    \[
        d_{arc}\Big(p_c(t_1),p_c^{m}(t_0)\Big)
    \]
    remains minimum, and the pair $(r_m, r_c)$ continues to be selected uniquely.  

    Since we assumed that $d_{arc}\Big(r_c(t_0),p_c^{m}(t_0)\Big)$ is minimum for the \textit{extremal} robot $r_m(t_0)$ and that $r_c$ moves towards $p_c^{m}(t_0)$ while all other robots remain stationary, at time $t_1$ the pair $(r_m, r_c)$ is again selected uniquely. Hence, for the remaining $(k-1)$ \textit{extremal} robots, the pair $(r_m, r_c)$ is again uniquely determined.

\noindent Therefore, if $\mathfrak{d}_i(t_0)$ is unique for a particular \textit{extremal} robot, the optimal assignment of that \textit{extremal} robot can always be identified with respect to all the other robots.
\end{itemize}

    \begin{figure}[!ht]
    \centering
    \vspace{-0.3cm}
    
    \includegraphics[width=0.5\textwidth]{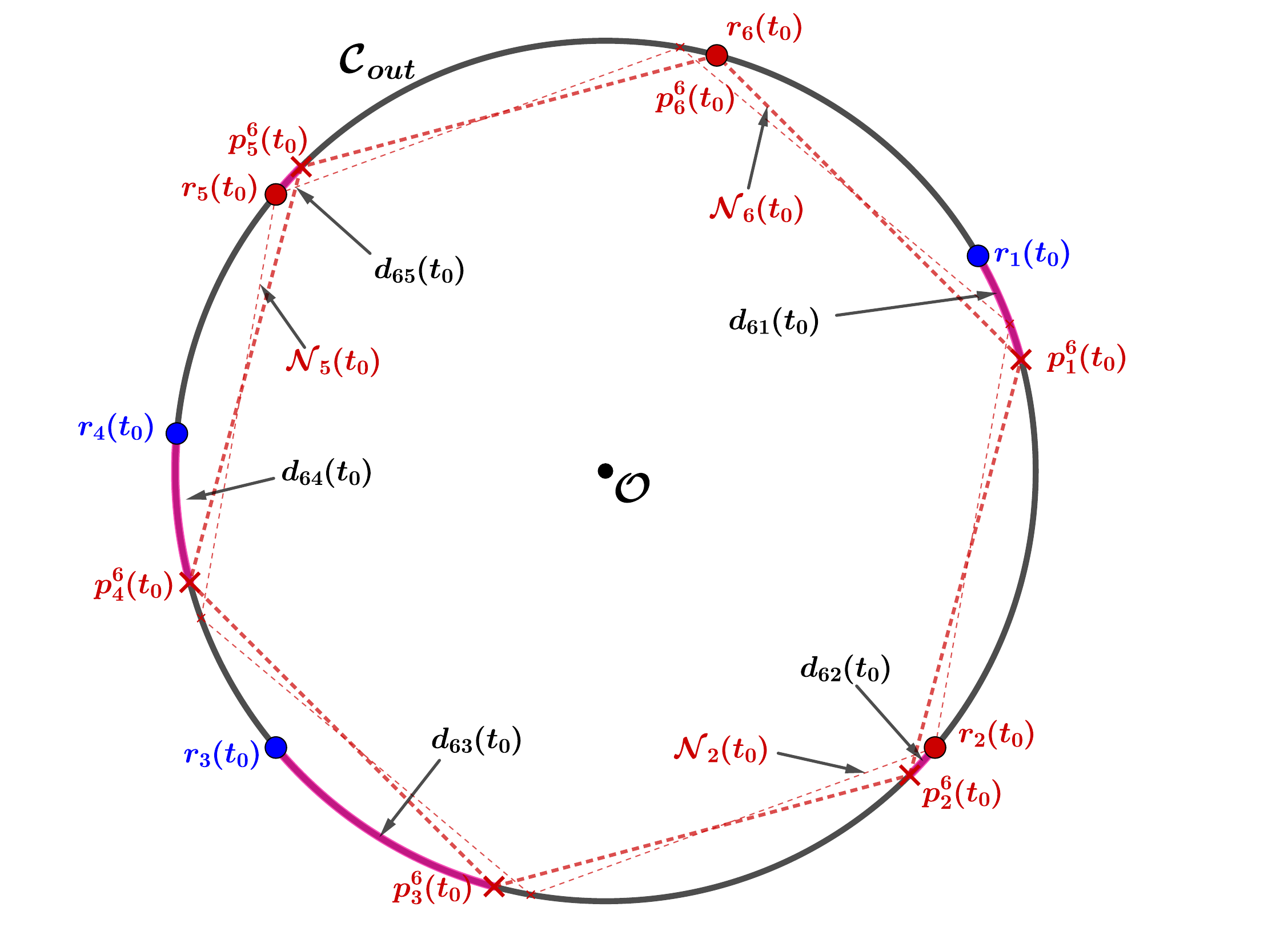}
   \caption{ \textit{An illustration of the initial configuration of six robots $\mathcal{R}(t_0)=\{r_1(t_0),r_2(t_0),\ldots,r_6(t_0)\}$ with the \textit{extremal} robot set $\mathcal E'(t_0)=\{r_2(t_0),r_5(t_0),r_6(t_0)\}$. Using \textbf{Observation}~\ref{obs:asymmetric-ordering} on $\mathcal{R}(t_0)$, the \textit{extremal} robot $r_6(t_0)$ (red dot) is selected. The destination set determined by fixing $r_6(t_0)$ is $\mathcal P^6(t_0)=\{p_1^2(t_0),p_2^2(t_0),\ldots,p_6^2(t_0)\}$ (red crosses). Here, $\mathfrak{d}(t_0)=\min\{d_{61}(t_0),\, d_{62}(t_0),\, d_{63}(t_0),\,d_{64}(t_0),\,d_{65}(t_0)\}=\{d_{62}(t_0),\,d_{65}(t_0)\}$, and according to the \textbf{Observation}~\ref{obs:asymmetric-ordering}, the \textit{candidate} robot is $r_2$.}}
    
     \label{Lem-8(3)}

\end{figure}

\item[\textbf{Case-II:}] Now consider the case when the value of $\mathfrak{d}_i(t_0)$ is not unique.  In such a scenario, $\mathfrak{d}_i(t_0)$ may correspond to multiple values,  either for a single \textit{extremal} robot or for multiple \textit{extremal} robots.

\begin{itemize}
    \item \textbf{Subcase-II A:} Suppose the value of $\mathfrak{d}_i(t_0)$ is not unique for a single \textit{extremal} robot, say $r_i$.  Consider the corresponding distance set with respect to the \textit{extremal} robot $r_i$ at time $t_0$: 

\[
\resizebox{\dimexpr\paperwidth-5.8cm\relax}{!}{$
P_{r_i}(t_0)
=
\begin{bmatrix}
d_{i1}(t_0) & d_{i2}(t_0) & \cdots & d_{in}(t_0)
\end{bmatrix}
=
\begin{bmatrix}
d_{arc}\!\big(r_1(t_0), p_1^{i}(t_0)\big) &
d_{arc}\!\big(r_2(t_0), p_2^{i}(t_0)\big) &
\cdots &
d_{arc}\!\big(r_n(t_0), p_n^{i}(t_0)\big)
\end{bmatrix}
$}
\]

  In this case, we have
\[
\mathfrak{d}_i(t_0)
=
d_{ij}(t_0)
=
d_{ik}(t_0),
\quad
\text{for some \textit{candidate} robots } r_j \text{ and } r_k
\ \text{(See \textbf{Figure}~\ref{Lem-8(3)})}.
\]


       The robot configuration $\mathcal R(t_0)$ is asymmetric and using \textbf{Observation}~\ref{obs:asymmetric-ordering} for the robots in $\mathcal R(t_0)$.  Suppose the \textit{candidate} robot $r_j$, as determined by \textbf{Observation}~\ref{obs:asymmetric-ordering}, moves to a point $p_i(t_1)$ on the circle $\mathcal C$.  In this case, the value of $\mathfrak{d}_i(t_0)$ decreases by the arc length $x$ and becomes unique at time $t_1$. Following the same reasoning as in \textbf{Subcase-IA}, it can be shown that the optimal assignment of the \textit{extremal} robot $r_i$ can always be determined with respect to all other robots.

     \item \textbf{Subcase-II B:} Now suppose that the value of $\mathfrak{d}_i(t_0)$ is the same for several \textit{extremal} robots. In this case, the configuration $\mathcal R(t_0)$ is asymmetric and therefore admits a fixed ordering of the robots using \textbf{Observation}~\ref{obs:asymmetric-ordering}. Using this ordering at time $t_0$, one of the \textit{extremal} robots is selected, say $r_j$. As in \textbf{Subcase-IA}, once $r_j$ is fixed, its optimal assignment with respect to all other robots can be uniquely determined.
\end{itemize}

\end{itemize}

\noindent Hence, in an asymmetric robot configuration $\mathcal R(t_0)$ with multiple optimal assignments, the unique optimal assignment remains fixed with respect to the \textit{extremal} robot selected according to the \textbf{Observation}~\ref{obs:asymmetric-ordering}.

\end{proof}

  \begin{observation}
\label{obs-02}
The algorithm \textit{AsymM1dMinSumC()} achieves collision-free optimal placement of robots for the min-sum problem on the circle $\mathcal C$.
\end{observation}

\subsection*{4.2.3~~Rotationally Symmetric Configurations $\mathcal I_4$: Multiple Optimal Assignments}

This subsection studies robot configurations that admit rotational symmetry but no line of symmetry. In such configurations, the rotational symmetry of the initial placement leads to multiple assignments achieving the same minimum total movement cost. Since no line of symmetry is present, these configurations exhibit a specific structural property that must be explicitly identified in order to design a correct
algorithm. To formalize this structure, we begin by defining the order of rotational symmetry of a configuration and rotational equivalence Classes.

\begin{definition}[Order of Rotational Symmetry]
\label{def-rot-order-w}
A robot configuration $\mathcal{R}(t_0)$ on the circumference $\mathcal C_{out}$
is said to admit \emph{rotational symmetry of order $w \ge 2$} if $w$ is the largest integer such that a rotation of angle $2\pi/w$ about the center of $\mathcal C$ maps the configuration onto itself. That is, for every robot position $r_i(t_0)\in\mathcal{R}(t_0)$, there exists a robot $r_j(t_0)\in\mathcal{R}(t_0)$ such that
\[
r_j(t_0) = \rho(r_i(t_0)),
\]
where $\rho$ denotes the rotation by angle $2\pi/w$.
\end{definition}

\begin{definition}[Rotational Equivalence Classes]
\label{def-rot-equiv}
Let $\mathcal{R}(t_0)$ be a robot configuration that admits rotational symmetry of order $w \ge 2$, as defined above. Let
$\rho : \mathcal C_{out} \rightarrow \mathcal C_{out}$ denote the rotation by angle $2\pi/w$ about the center of $\mathcal C$.

\noindent For $k \ge 0$, let $\rho^{k}$ denotes the $k$-fold composition of $\rho$, i.e.,
$\rho^{k} = \underbrace{\rho \circ \rho \circ \cdots \circ \rho}_{k\text{ times}}$, with $\rho^{0}$ being the identity mapping.  
The \emph{rotational equivalence class} of a robot $r_i(t_0)$ is defined as
\[
\mathcal R(r_i)
=
\big\{
\rho^{k}(r_i(t_0)) \;\big|\; k = 0,1,\ldots,w-1
\big\}.
\]

\noindent The set $\mathcal{R}(t_0)$ is thus partitioned into $w$ disjoint rotational equivalence classes $\mathcal{R}_0, \mathcal{R}_1, \ldots, \mathcal{R}_{w-1}$, where each class is mapped to another under the action of $\rho$.
\end{definition}

\begin{lemma}
\label{lemma-extremal-rot}
If an initial configuration admits rotational symmetry of order $w \ge 2$ with no line of symmetry, then all \textit{extremal} robots belong to a single rotational equivalence class. All other rotational classes contain only non-\textit{extremal} robots.
\end{lemma}

\begin{proof}\normalfont
Consider an initial configuration $\mathcal{R}(t_0) \in \mathcal I_4$ placed on a circle $\mathcal{C}$ (See \textbf{Figure}~\ref{Rot}). The rotational symmetry partitions the robots into $w$ disjoint rotational equivalence classes, where robots in the same class are mapped to each other by rotations of angle $2\pi/w$.

\noindent Assume that the \textit{extremal} robots belong to more than one equivalence class. Due to rotational symmetry, the presence of an \textit{extremal} robot in an equivalence class implies that every robot in that class is \textit{extremal}, as all robots within the class have identical views. Consequently, there would exist at least two distinct regular $n$-gons determined by different sets of \textit{extremal} robots. Since all \textit{extremal} robots lie on the same circle $\mathcal{C}$, the existence of more than one such regular $n$-gon implies that the configuration admits a reflection mapping from one set of \textit{extremal} robots to another. This reflection induces a line of symmetry in the configuration, which contradicts the assumption that the configuration admits no line of symmetry. Therefore, all \textit{extremal} robots must belong to a single rotational equivalence class. It follows that the remaining $w-1$ rotational classes contain only non-\textit{extremal} robots.
\end{proof}

\noindent It is important to note that, although the problem is solvable in this case, the presence of rotational symmetry implies that more than one optimal assignment can exist. This property is captured in the following observation.

\begin{observation}
\label{obs-rot-no-unique}
If an initial robot configuration $\mathcal R(t_0)\in \mathcal I_4$, then the optimal assignment for uniform circle formation is not unique.
\end{observation}

\noindent \textbf{Lemma}~\ref{lemma-extremal-rot} shows that in a rotationally symmetric configuration with no line of symmetry, all \textit{extremal} robots belong to one rotational class, and all other robots are non-\textit{extremal}. Hence, the \textit{extremal} robots together determine the same destination point set. Based on this property, we now describe the algorithm designed for such configurations.

\subsubsection*{4.2.3.1~~Overview of the Algorithm \it{RotSymM1dMinSumC()}}

\noindent This algorithm is designed for initial robot configurations that admit rotational symmetry. For configurations with rotational symmetry of order $w \ge 2$ and no line of symmetry, \textbf{Lemma}~\ref{lemma-extremal-rot} guarantees that all \textit{extremal} robots belong to a single rotational equivalence class. This uniquely fixes the destination points as an regular $n$-gon, since all the \textit{extremal} robots have identical views. The algorithm allows movement only for \textit{non-extremal} robots. Each robot computes the same destination point set and locally checks whether it is a \textit{candidate} robot, i.e., whether it has a free arc path to its assigned destination. The ties can be broken using the concept of views.  \textit{Candidate} robots are selected only from non-\textit{extremal} rotational classes, ensuring that \textit{extremal} robots remain fixed. If the \textit{candidate} robots from the same rotational class move concurrently, the symmetry is preserved and the total arc distance strictly decreases. If, due to asynchrony, there may exist pending moves. As a result, the symmetry may be broken, and the configuration may transform into an asymmetric one. Thus, the execution switches to \textit{AsymM1dMinSumC()}. In all scenarios, the robots travel exclusively along free arc paths, assignments are preserved, and finite-time convergence to the destinations is guaranteed under the \textit{ASYNC} model.

\begin{algorithm}[!ht]
\caption{\it{: RotSymM1dMinSumC()}}
\label{alg:rotsym-minsum}
\begin{algorithmic}[1]
\Require Initial configuration $\mathcal{R}(t) \in \mathcal{I}_4$ on $\mathcal{C}_{out}$
\Ensure \textit{min-sum uniform coverage} on a circle

\State Compute the \textit{extremal} robot set $\mathcal{E}'(t)$
\State Compute the destination point set $\mathcal{P}(t)$ obtained by $\mathcal{E}'(t)$

\ForAll{robots $r_i \in \mathcal{R}(t)$ \textbf{in parallel}}
    \If{$r_i \in \mathcal{E}'(t)$}
        \State $r_i$ stays idle
    \Else
        \If{$r_i$ is a \textit{candidate} robot and has a free arc path to its destination}
            \State $r_i$ moves along $\mathcal{C}_{out}$ toward its assigned destination
        \Else
            \State $r_i$ stays idle
        \EndIf
    \EndIf
\EndFor

\If{rotational symmetry is broken during execution due to pending moves}
    \State Execute \textit{AsymM1dMinSumC()}
\EndIf

\end{algorithmic}
\end{algorithm}

\subsubsection*{4.2.3.2~~Correctness of the Algorithm \it{RotSymM1dMinSumC()}}


\noindent We prove the correctness of the algorithm \textit{RotSymM1dMinSumC()} for robot configurations in the class $\mathcal{I}_4$. Let $\mathcal{R}(t_0)$ be a robot configuration that admits rotational symmetry with no line of symmetry and has multiple optimal assignments. Under the \textit{ASYNC} model, the algorithm \textit{RotSymM1dMinSumC()} fixes the target points once they are computed; therefore, the total distance to the targets remains optimal throughout the execution. By \textbf{Lemma}~\ref{lemma-extremal-rot}, the robots in a rotational equivalence class, denoted by $\mathscr{R}_{\mathrm{candidate}}$, are selected as candidate robots. As long as their movements preserve rotational symmetry, all robots in $\mathscr{R}_{\mathrm{candidate}}$ identify the same set of \textit{extremal} robots in the class $\mathscr{E}_{\mathrm{extremal}}$, and the target points do not change. If the configuration becomes asymmetric due to pending asynchronous moves, \textbf{Lemma}~\ref{lem:rot-invariance} ensures that the \textit{extremal} robots in $\mathscr{E}_{\mathrm{extremal}}$ remain invariant. Hence, the assigned destinations remain unchanged and the total distance stays optimal. Finally, \textbf{Lemma}~\ref{l-9} and \textbf{Lemma}~\ref{lem-correct-rotsym} show that under the \textit{ASYNC} model, the algorithm terminates in finite time, avoids collisions, and places all robots at positions that achieve the optimal total distance.

\begin{lemma}[Invariance of Multiple Optimal Assignments in $\mathcal I_4$]
\label{lem:rot-invariance}
Let $\mathcal{R}(t_0)$ be a robot configuration in $\mathcal{I}_4$ admitting rotational symmetry at time $t_0 \ge 0$. Let $\mathscr{E}_{\mathrm{extremal}}$ denote the rotational equivalence class of \textit{extremal} robots, and let $\mathscr{R}_{\mathrm{candidate}}$ denote a rotational equivalence class of candidate robots. Fix an \textit{extremal} robot $r_e$ in the class $\mathscr{E}_{\mathrm{extremal}}$, and consider the destination points assigned at time $t_0$ to the robots in $\mathscr{R}_{\mathrm{candidate}}$ by fixing $r_e$. If, at time $t_1$, the candidate robots in $\mathscr{R}_{\mathrm{candidate}}$ reach intermediate positions (see \textbf{Definition}~\ref{Intermediate} on page~\pageref{Intermediate}) on $\mathcal{C}_{out}$ before reaching their respective destinations, then the class $\mathscr{E}_{\mathrm{extremal}}$ remains invariant throughout the time interval $[t_0,t_1]$.
\end{lemma}

\begin{proof}
    The proof of this lemma follows the same arguments as those of \textbf{Lemma}~\ref{lem:n-robots}.
\end{proof}

\begin{lemma} 
\label{l-9}
If an initial robot configuration $\mathcal{R}(t_0)\in \mathcal I_4$, then the \textit{min-sum uniform coverage} on a circle problem is solvable for a set of asynchronous robots.
\end{lemma}

\begin{proof}\normalfont
Let $\mathcal{R}(t_0)=\{r_1(t_0), r_2(t_0), \ldots, r_n(t_0)\}\in \mathcal I_4$ be an initial configuration of $n$ robots placed on the circumference $\mathcal C_{out}$ of the circle $\mathcal C$. The rotational symmetry partitions $\mathcal{R}(t_0)$ into $w$ disjoint rotational equivalence classes $\mathcal{R}_0, \mathcal{R}_1, \ldots, \mathcal{R}_{w-1}$, where each class is mapped to another by a rotation of angle $2\pi/w$. By \textbf{Lemma}~\ref{lemma-extremal-rot}, exactly one of these classes, say $\mathcal R_e$, contains all the \textit{extremal} robots. These \textit{extremal} robots jointly determine the same regular $n$-gon and hence uniquely fix the destination point set. All the remaining $w-1$ rotational classes contain only \textit{non-extremal} robots.

\noindent The \textbf{Algorithm} \textit{RotSymM1dMinSumC()} allows movement only for \textit{candidate} robots selected from the non-\textit{extremal} classes. Whenever a robot from a \textit{non-extremal} class is identified as a candidate, all robots in the same rotational equivalence class are candidates at the same time due to the fact that they have identical views. Thus, movement always involves an entire rotational class of robots.

\noindent If all robots belonging to a candidate rotational class move concurrently, the rotational symmetry of the configuration is preserved. In this case, each moving robot reduces the same arc distance toward its assigned destination, and the total sum of arc distances strictly decreases. Due to asynchrony, it is possible that some robots from a candidate class move while others experience pending moves. In such a situation, rotational symmetry is broken, and the configuration transforms into an asymmetric one. From that instant of time, the execution follows the behavior of the \textbf{Algorithm} \textit{AsymM1dMinSumC()}, which guarantees progress toward a unique optimal assignment. During both the rotationally symmetric and asymmetric phases, robots follow only free-arc trajectories. Thus, the destination assignments are preserved, and every moving robot makes strict progress by decreasing its arc distance to the assigned destination. Since all the arc distances are finite and decrease monotonically, all robots reach their destinations in finite time under the \textit{ASYNC} model. In the final configuration, all robots occupy distinct vertices of a regular $n$-gon on $\mathcal C_{out}$, thereby achieving the \textit{min-sum uniform coverage} on a circle. Hence, the problem is solvable for asynchronous robots when the initial configuration admits rotational symmetry but no line of symmetry.
\end{proof}

\begin{lemma}
\label{lem-correct-rotsym}
Let $\mathcal{R}(t_0) \in \mathcal I_4$. Then the algorithm
\textit{RotSymM1dMinSumC()} terminates in finite time under the \textit{ASYNC} model and moves all robots to distinct vertices of a regular $n$-gon on
$\mathcal{C}_{out}$, achieving collision-free \textit{min-sum uniform coverage} on a circle.
\end{lemma}

 \begin{figure}[!ht]
    \centering
    \vspace{-0.3cm}
    
    \includegraphics[width=0.5\textwidth]{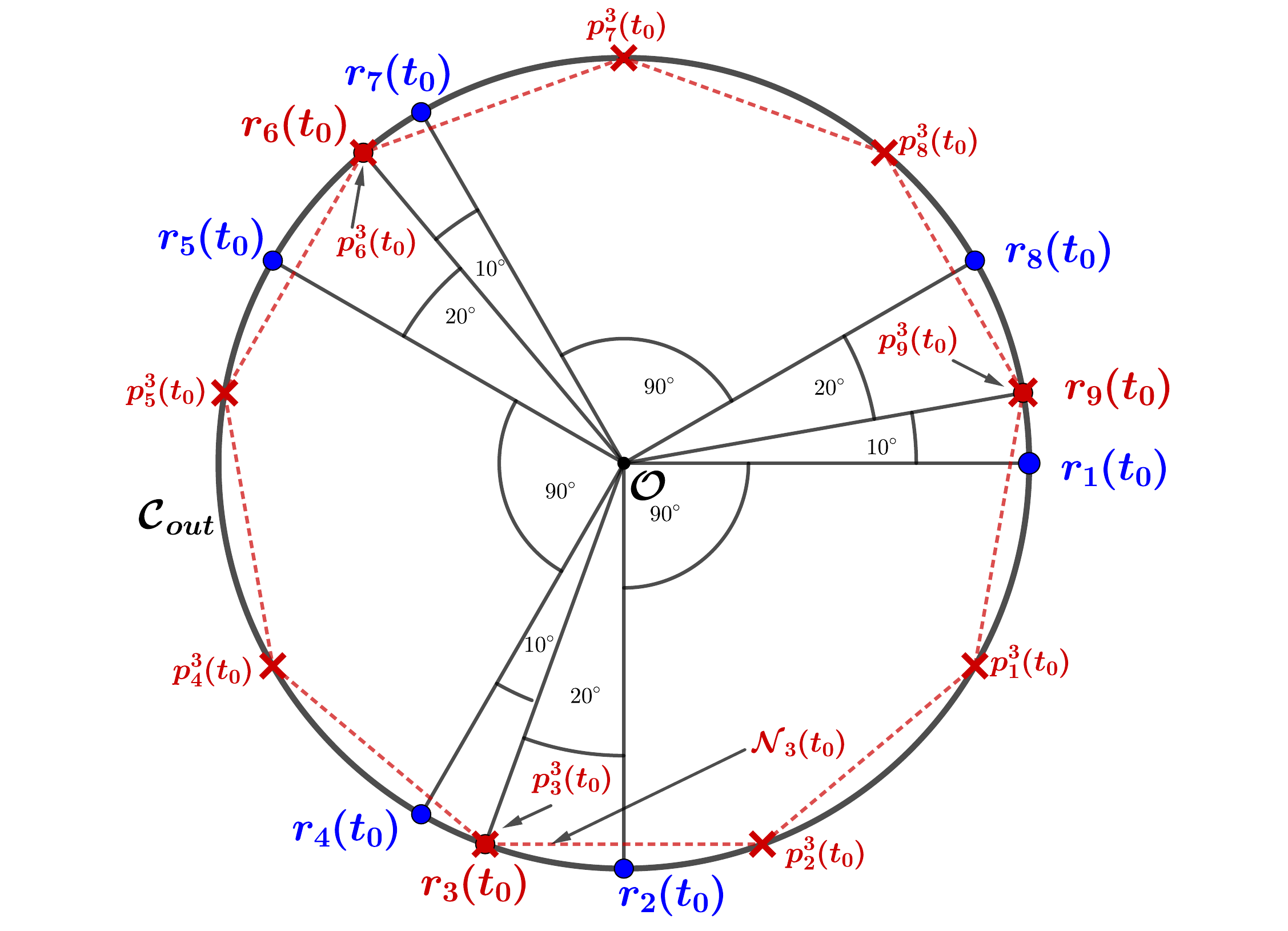}
   \caption{ \textit{An illustration of an initial configuration of nine robots $\mathcal{R}(t_0)=\{r_1(t_0), r_2(t_0), \ldots, r_9(t_0)\}$ admitting rotational symmetry of order $w=3$ (triangular symmetry). The robots are partitioned into three rotational equivalence classes, namely $\mathcal R_0=\{r_3(t_0), r_6(t_0), r_9(t_0)\}$, $\mathcal R_1=\{r_1(t_0), r_4(t_0), r_7(t_0)\}$, and $\mathcal R_2=\{r_2(t_0), r_5(t_0), r_8(t_0)\}$. The \textit{extremal} robot set is $\mathcal R_e=\mathcal R_0$, which uniquely determines the regular $9$-gon $\mathcal{N}_3(t_0)$. The corresponding destination point set satisfies $\mathcal{P}^3(t_0)=\mathcal{P}^6(t_0)=\mathcal{P}^9(t_0)$. Only the destination point set $\mathcal{P}^3(t_0)=\{p_1^3(t_0), p_2^3(t_0), \ldots, p_9^3(t_0)\}$ is shown in the figure.}}

     \label{Rot}

\end{figure}
\begin{proof}\normalfont
Let $\mathcal{R}(t_0)\in \mathcal I_4$ be an initial robot configuration on $\mathcal{C}_{out}$. By \textbf{Lemma}~\ref{lemma-extremal-rot}, all the \textit{extremal} robots belong to a single rotational equivalence class $\mathcal R_e$ and remain fixed throughout the execution of the \textbf{Algorithm} \textit{RotSymM1dMinSumC()}. These robots collectively identify a single destination point set, which is computed identically by all robots and remains fixed throughout the execution of the algorithm. The algorithm permits movement only for \textit{candidate} robots chosen from the \textit{non-extremal} rotational equivalence classes. A robot is selected as a candidate only if its arc path to the assigned destination is free. Hence, whenever a robot moves, it strictly decreases its arc distance to the destination, while the robots with blocked paths remain stationary. Since robots move exclusively along the free arc paths on $\mathcal{C}_{out}$, their arc paths never intersect. Therefore, collisions cannot occur during the execution of the algorithm. Due to the presence of rotational symmetry and the absence of any line of symmetry, the candidate robots always belong to a complete rotational equivalence class. If all robots of such a class move concurrently, the rotational symmetry of the configuration is preserved. If, due to asynchrony, only a subset of the robots in the class move, the symmetry may be broken, and the configuration becomes asymmetric. From that point onward, the execution proceeds with \textit{AsymM1dMinSumC()}, which also allows collision-free robot movement along the free arc paths. In each case, every robot that moves strictly decreases its arc distance. Since all arc distances are finite and monotonically decreasing, only a finite number of movements are possible. Consequently, all robots reach their assigned destinations in finite time under the \textit{ASYNC} model. Thus, the algorithm \textit{RotSymM1dMinSumC()} terminates in finite time and achieves collision-free \textit{min-sum uniform coverage} on a circle.
\end{proof}

\begin{theorem}
\label{thm-main}
The \textit{min-sum uniform coverage} on a circle problem by \textit{oblivious} and
\textit{silent} robots under the \textit{ASYNC} model is deterministically solvable for
all initial configurations except those belonging to classes
$\mathcal I_3$ and $\mathcal I_6$.
\end{theorem}

\begin{proof}\normalfont
Let $\mathcal{R}(t_0)$ be an arbitrary initial robot configuration on the circumference $\mathcal C_{out}$ of the given circle $\mathcal C$. By the configuration classification introduced in \textbf{Section}~\ref{class} (See Page~\pageref{class}), every initial configuration belongs to exactly one of the disjoint classes $\mathcal I_1, \mathcal I_2, \mathcal I_3, \mathcal I_4, \mathcal I_5, \mathcal I_6$. We analyze each class separately.

\medskip
\noindent\textbf{Case–1 ($\mathcal I_1$ — Asymmetric configurations):}
The asymmetric configurations are divided into two subclasses: $\mathcal I'_1$, which admit a unique optimal assignment, and $\mathcal I''_1$, which admit multiple optimal assignments. For configurations in $\mathcal I'_1$, the absence of symmetry ensures that there exists an ordering of the robots and the optimal assignment is uniquely determined. As shown in \textbf{Section}~4.1.1, Algorithm \textit{AsymU1dMinSumC()} fixes the \textit{extremal} robot, assigns destinations accordingly, and allows only robots with a free arc path to move. Each movement strictly decreases the total arc distance, collisions are avoided, and all robots reach distinct vertices of a regular $n$-gon in finite time. For configurations in $\mathcal I''_1$, although multiple optimal assignments exist, the configuration remains asymmetric. As established in \textbf{Section}~4.2.2, by imposing a deterministic ordering by using \textbf{Observation}~\ref{obs:asymmetric-ordering} on the robots, the Algorithm \textit{AsymM1dMinSumC()} selects a unique \textit{extremal} robot and fixes an optimal assignment. The selected assignment remains invariant during the execution, and the robots move only along unobstructed arc paths.
Moreover, exactly one candidate robot advances in each iteration, guaranteeing collision-free convergence in finite time. Hence, the \textit{min-sum uniform coverage} on a circle problem is solvable for all configurations in $\mathcal I_1$.

\medskip
\noindent\textbf{Case–2 ($\mathcal I_2$ — Single line of symmetry with a unique optimal assignment):}
Let $\mathcal R(t_0)\in\mathcal I_2$ be an initial configuration admitting exactly one line of symmetry $\mathcal L$. By \textbf{Lemma}~\ref{l-6}, the robot located at the intersection of $\mathcal L$ with the circumference $\mathcal C_{out}$ is the unique \textit{extremal} robot and uniquely determines the optimal destination point set. As shown in Section~4.1.2, the Algorithm \textit{SymU1dMinSumC()} fixes this \textit{extremal} robot and allows only \textit{non-extremal} robots to move in mirror pairs with respect to $\mathcal L$ along free arc paths. Each movement strictly decreases the total sum of arc distances, and collisions are avoided. Hence, the min-sum uniform
circle formation problem is solvable for all configurations in $\mathcal I_2$.

\medskip
\noindent\textbf{Case–3 ($\mathcal I_3$ — Single line of symmetry with no \textit{extremal} robot on the line):} In this case, multiple optimal assignments with equal total cost exist. Due to symmetry, robots cannot deterministically select one of these
assignments. As shown in Section~4.1.2, this ambiguity prevents deterministic convergence. Hence, the problem is not solvable for configurations in $\mathcal I_3$.

\medskip
\noindent\textbf{Case–4 ($\mathcal I_4$ — Rotational symmetry with no line of symmetry):} Let $\mathcal R(t_0)\in\mathcal I_4$ be an initial robot configuration that admits rotational symmetry of order $w\ge 2$ and no line of symmetry. By
\textbf{Lemma}~\ref{lemma-extremal-rot}, all the \textit{extremal} robots belong to a single rotational
equivalence class, and these robots jointly determine the regular $n$-gon that fixes the destination point set. As shown in Section~4.2.3, Algorithm \textit{RotSymM1dMinSumC()} allows movement only for non-\textit{extremal} robots that have free arc paths to their assigned destinations. If rotational symmetry is preserved, the symmetric robots move concurrently, and the total sum of arc distances
strictly decreases. If the symmetry is broken due to asynchrony, the execution reduces to the asymmetric case, which is solved by \textit{AsymM1dMinSumC()}. In all
executions, the robots move only along free arc paths, no collisions occur, and all robots reach their destinations in finite time. Hence, the \textit{min-sum uniform coverage} on a circle problem is solvable for all configurations in $\mathcal I_4$.

\medskip
\noindent\textbf{Case–5 ($\mathcal I_5$ — Multiple lines of symmetry with \textit{extremal} robots on the lines):} By \textbf{Observation}~\ref{obs-multi-los}, all robots are already positioned on the vertices of a regular $n$-gon at time $t_0$. Hence, a uniform circular formation has already been achieved, and no further movement of robots is required. The problem is trivially solvable for configurations in $\mathcal I_5$.

\medskip
\noindent\textbf{Case–6 ($\mathcal I_6$ — Multiple lines of symmetry with no \textit{extremal} robot on any line):} As shown in Section~4.2.2, the presence of multiple symmetries leads to several optimal assignments with equal cost that cannot be distinguished by oblivious and silent robots. Consequently, the problem is unsolvable for configurations in $\mathcal I_6$.

\medskip
\noindent
Since the problem is solvable for configurations in
$\mathcal I_1 \cup \mathcal I_2 \cup \mathcal I_4 \cup \mathcal I_5$ and unsolvable for configurations in $\mathcal I_3 \cup \mathcal I_6$, the theorem follows.
\end{proof}

\section{Conclusion and Future Work}

\label{Conclusion and Future Work}
\noindent \textit{Conclusion:}
In this paper, we studied the \textit{min-sum uniform coverage} problem for autonomous mobile robots constrained to one-dimensional geometric environments. We considered two fundamental settings: robots initially deployed on a finite line segment and robots positioned on the circumference of a circle. The robots operate under severe limitations, being anonymous, autonomous, homogeneous, and oblivious, and without relying on additional capabilities such as memory, communication, chirality, lights, or multiplicity detection.
We provided a formal characterization of the problem and identified all classes of initial configurations for which \textit{min-sum uniform coverage} is unsolvable due to inherent symmetry or indistinguishability. For all remaining configurations, we proposed deterministic distributed algorithms that guarantee convergence to a uniformly spaced configuration while minimizing the total distance traveled by all robots. The correctness and optimality of the proposed solutions were established through rigorous proofs. These results demonstrate that global min-sum optimality can be achieved even under strong robotic constraints by exploiting geometric properties of the underlying environment.

\noindent \textit{Future Work:}
Beyond the theoretical interest, the techniques developed in this work can serve as foundational building blocks for more complex formation and coverage tasks in higher-dimensional settings. Several directions remain open for future investigation. A natural extension is to consider robots initially located in the interior of a disk, rather than restricted to a line segment or the boundary of a circle. Determining whether exact min-sum solutions exist in such settings remains an open problem. Another direction is the development of approximation algorithms for cases where exact optimality may be computationally infeasible or impossible to achieve under additional constraints.
Further extensions include considering environmental constraints such as obstacles or restricted motion paths, as well as analyzing the time and computational complexity of the proposed algorithms under fully asynchronous distributed execution models. Addressing these challenges would further enhance the applicability of \textit{min-sum uniform coverage} strategies in practical robotic systems.

\bibliographystyle{plain} 
\bibliography{Bibliography}

@String{Computing = "Computing" }

@String{Computer = "{IEEE} Computer" }

@String{Springer = "Springer-Verlag" }

@article{das2022k,
  title={k-Circle formation by disoriented asynchronous robots},
  author={Das, Bibhuti and Chakraborty, Abhinav and Bhagat, Subhash and Mukhopadhyaya, Krishnendu},
  journal={Theoretical Computer Science},
  volume={916},
  pages={40--61},
  year={2022},
  publisher={Elsevier}
}

@ARTICLE{Suzuki1999,
    author = {Ichiro Suzuki  and Masafumi Yamashita.},
    title = {Distributed anonymous mobile robots: Formation of geometric patterns},
    journal = {SIAM Journal on Computing},
    year = {1999},
    volume = {28},
    pages = {pages 1347-1363}
}

@INPROCEEDINGS{vigl2016,
    author = {Marcello Mamino and Giovanni Viglietta},
    title = {Square Formation by Asynchronous Oblivious Robots},
    booktitle = {In Proceedings of the 28th Canadian Conference on Computational Geometry (CCCG)},
    year = {2016},
    pages = {1-6},
    publisher = {}
}

@inproceedings{feletti2018uniform,
  title={Uniform circle formation for swarms of opaque robots with lights},
  author={Feletti, Caterina and Mereghetti, Carlo and Palano, Beatrice},
  booktitle={International Symposium on Stabilizing, Safety, and Security of Distributed Systems},
  pages={317--332},
  year={2018},
  organization={Springer}
}

@incollection{bhagat2018optimum,
  title={Optimum circle formation by autonomous robots},
  author={Bhagat, Subhash and Mukhopadhyaya, Krishnendu},
  booktitle={Advanced Computing and Systems for Security: Volume Five},
  pages={153--165},
  year={2018},
  publisher={Springer}
}

@article{chen2015optimal,
  title={Optimal point movement for covering circular regions},
  author={Chen, Danny Z and Tan, Xuehou and Wang, Haitao and Wu, Gangshan},
  journal={Algorithmica},
  volume={72},
  number={2},
  pages={379--399},
  year={2015},
  publisher={Springer}
}

@article{bhattacharya2009optimal,
  title={Optimal movement of mobile sensors for barrier coverage of a planar region},
  author={Bhattacharya, Binay and Burmester, Mike and Hu, Yuzhuang and Kranakis, Evangelos and Shi, Qiaosheng and Wiese, Andreas},
  journal={Theoretical Computer Science},
  volume={410},
  number={52},
  pages={5515--5528},
  year={2009},
  publisher={Elsevier}
}

@book{santorobook2,
  editor    = {Paola Flocchini and
               Giuseppe Prencipe and
               Nicola Santoro},
  title     = {Distributed Computing by Mobile Entities, Current Research in Moving and Computing},
   publisher = {Springer},
  year      = {2019},
  }

@inproceedings{DBLP:conf/isaac/FlocchiniPSW99,
  author       = {Paola Flocchini and
                  Giuseppe Prencipe and
                  Nicola Santoro and
                  Peter Widmayer},
  editor       = {},
  title        = {Hard Tasks for Weak Robots: The Role of Common Knowledge in Pattern
                  Formation by Autonomous Mobile Robots},
  booktitle    = {Algorithms and Computation, 10th International Symposium, {ISAAC}
                  },
  series       = {Lecture Notes in Computer Science},
  volume       = {},
  pages        = {93--102},
  publisher    = {Springer},
  year         = {1999},
  url          = {https://doi.org/10.1007/3-540-46632-0\_10},
  doi          = {10.1007/3-540-46632-0\_10},
  bibsource    = {dblp computer science bibliography, https://dblp.org}
}

@article{DBLP:journals/siamcomp/SuzukiY99,
  author       = {Ichiro Suzuki and
                  Masafumi Yamashita},
  title        = {Distributed Anonymous Mobile Robots: Formation of Geometric Patterns},
  journal      = {{SIAM} J. Comput.},
  volume       = {28},
  number       = {4},
  pages        = {1347--1363},
  year         = {1999},
  url          = {https://doi.org/10.1137/S009753979628292X},
  doi          = {10.1137/S009753979628292X},
  bibsource    = {dblp computer science bibliography, https://dblp.org}
}

@book{DBLP:series/synthesis/2012Flocchini,
  author       = {Paola Flocchini and
                  Giuseppe Prencipe and
                  Nicola Santoro},
  title        = {Distributed Computing by Oblivious Mobile Robots},
  series       = {Synthesis Lectures on Distributed Computing Theory},
  publisher    = {Morgan {\&} Claypool Publishers},
  year         = {2012},
  url          = {https://doi.org/10.2200/S00440ED1V01Y201208DCT010},
  doi          = {10.2200/S00440ED1V01Y201208DCT010},
  isbn         = {978-3-031-00880-1},
  bibsource    = {dblp computer science bibliography, https://dblp.org}
}

@inproceedings{tan2010new,
  title={New algorithms for barrier coverage with mobile sensors},
  author={Tan, Xuehou and Wu, Gangshan},
  booktitle={International Workshop on Frontiers in Algorithmics},
  pages={327--338},
  year={2010},
  organization={Springer}
}

@inproceedings{DBLP:conf/icdcit/DuttaCDM12,
  author       = {Ayan Dutta and
                  Sruti Gan Chaudhuri and
                  Suparno Datta and
                  Krishnendu Mukhopadhyaya},
  editor       = {},
  title        = {Circle Formation by Asynchronous Fat Robots with Limited Visibility},
  booktitle    = {Distributed Computing and Internet Technology - 8th International
                  Conference, {ICDCIT} 2012, 
                  Proceedings},
  series       = {Lecture Notes in Computer Science},
  volume       = {},
  pages        = {},
  publisher    = {},
  year         = {2012},
  url          = {https://doi.org/10.1007/978-3-642-28073-3\_8},
  doi          = {10.1007/978-3-642-28073-3\_8},
  bibsource    = {dblp computer science bibliography, https://dblp.org}
}

@article{DBLP:journals/dc/FlocchiniPSV17,
  author       = {Paola Flocchini and
                  Giuseppe Prencipe and
                  Nicola Santoro and
                  Giovanni Viglietta},
  title        = {Distributed computing by mobile robots: uniform circle formation},
  journal      = {Distributed Comput.},
  volume       = {30},
  number       = {6},
  pages        = {413--457},
  year         = {2017},
  url          = {https://doi.org/10.1007/s00446-016-0291-x},
  doi          = {10.1007/S00446-016-0291-X},
  bibsource    = {dblp computer science bibliography, https://dblp.org}
}

@Article{app13137991,
AUTHOR = {Feletti, Caterina and Mereghetti, Carlo and Palano, Beatrice},
TITLE = {Uniform Circle Formation for Fully, Semi-, and Asynchronous Opaque Robots with Lights},
JOURNAL = {Applied Sciences},
VOLUME = {13},
YEAR = {2023},
NUMBER = {13},
ARTICLE-NUMBER = {7991},
URL = {https://www.mdpi.com/2076-3417/13/13/7991},
ISSN = {2076-3417},
DOI = {10.3390/app13137991}
}

@INPROCEEDINGS{10579120,
  author={Pattanayak, Debasish and Sharma, Gokarna},
  booktitle={2024 IEEE International Parallel and Distributed Processing Symposium (IPDPS)}, 
  title={Time-Color Tradeoff on Uniform Circle Formation by Asynchronous Robots}, 
  year={2024},
  volume={},
  number={},
  pages={987-997},
  doi={10.1109/IPDPS57955.2024.00092}}

@INPROCEEDINGS{128452,
  author={Sugihara, K. and Suzuki, I.},
  booktitle={Proceedings. 5th IEEE International Symposium on Intelligent Control 1990}, 
  title={Distributed motion coordination of multiple mobile robots}, 
  year={1990},
  volume={},
  number={},
  pages={138-143 vol.1},
  doi={10.1109/ISIC.1990.128452}}

@article{DBLP:journals/cacm/Debest95,
  author       = {Xavier A. Debest},
  title        = {Remark About Self-Stabilizing Systems},
  journal      = {Commun. {ACM}},
  volume       = {38},
  number       = {2},
  pages        = {115--117},
  year         = {1995},
  bibsource    = {dblp computer science bibliography, https://dblp.org}
}

@inproceedings{DBLP:conf/pomc/DefagoK02,
  author       = {Xavier D{\'{e}}fago and
                  Akihiko Konagaya},
  title        = {Circle formation for oblivious anonymous mobile robots with no common
                  sense of orientation},
  booktitle    = {Proceedings of the 2002 Workshop on Principles of Mobile Computing,
                  {POMC} 2002, October 30-31, 2002, Toulouse, France},
  pages        = {97--104},
  publisher    = {{ACM}},
  year         = {2002},
  url          = {https://doi.org/10.1145/584490.584509},
  doi          = {10.1145/584490.584509},
  bibsource    = {dblp computer science bibliography, https://dblp.org}
}

@article{DBLP:journals/tcs/DefagoS08,
  author       = {Xavier D{\'{e}}fago and
                  Samia Souissi},
  title        = {Non-uniform circle formation algorithm for oblivious mobile robots
                  with convergence toward uniformity},
  journal      = {Theor. Comput. Sci.},
  volume       = {396},
  number       = {1-3},
  pages        = {97--112},
  year         = {2008},
  url          = {https://doi.org/10.1016/j.tcs.2008.01.050},
  doi          = {10.1016/J.TCS.2008.01.050},
  bibsource    = {dblp computer science bibliography, https://dblp.org}
}

@inproceedings{DBLP:conf/opodis/FlocchiniPSV14,
  author       = {Paola Flocchini and
                  Giuseppe Prencipe and
                  Nicola Santoro and
                  Giovanni Viglietta},
  editor       = {Marcos K. Aguilera and
                  Leonardo Querzoni and
                  Marc Shapiro},
  title        = {Distributed Computing by Mobile Robots: Solving the Uniform Circle
                  Formation Problem},
  booktitle    = {Principles of Distributed Systems - 18th International Conference,
                  {OPODIS} 2014, Cortina d'Ampezzo, Italy, December 16-19, 2014. Proceedings},
  series       = {Lecture Notes in Computer Science},
  volume       = {8878},
  pages        = {217--232},
  publisher    = {Springer},
  year         = {2014},
  url          = {https://doi.org/10.1007/978-3-319-14472-6\_15},
  doi          = {10.1007/978-3-319-14472-6\_15},
  bibsource    = {dblp computer science bibliography, https://dblp.org}
}

@inproceedings{DBLP:conf/icdcn/MondalC18,
  author       = {Moumita Mondal and
                  Sruti Gan Chaudhuri},
  editor       = {Doina Bein},
  title        = {Uniform circle formation by mobile robots},
  booktitle    = {Proceedings of the Workshop Program of the 19th International Conference
                  on Distributed Computing and Networking, Varanasi, India, January
                  04-07, 2018},
  pages        = {20:1--20:2},
  publisher    = {{ACM}},
  year         = {2018},
  url          = {https://doi.org/10.1145/3170521.3170541},
  doi          = {10.1145/3170521.3170541},
  bibsource    = {dblp computer science bibliography, https://dblp.org}
}

@inproceedings{DBLP:conf/icdcit/MondalC20,
  author       = {Moumita Mondal and
                  Sruti Gan Chaudhuri},
  editor       = {Dang Van Hung and
                  Meenakshi D'Souza},
  title        = {Uniform Circle Formation by Swarm Robots Under Limited Visibility},
  booktitle    = {Distributed Computing and Internet Technology - 16th International
                  Conference, {ICDCIT} 2020, Bhubaneswar, India, January 9-12, 2020,
                  Proceedings},
  series       = {Lecture Notes in Computer Science},
  volume       = {11969},
  pages        = {420--428},
  publisher    = {Springer},
  year         = {2020},
  url          = {https://doi.org/10.1007/978-3-030-36987-3\_28},
  doi          = {10.1007/978-3-030-36987-3\_28},
  bibsource    = {dblp computer science bibliography, https://dblp.org}
}

@article{DBLP:journals/aghcs/AdhikaryKS21,
  author       = {Ranendu Adhikary and
                  Manash Kumar Kundu and
                  Buddhadeb Sau},
  title        = {Circle formation by asynchronous opaque robots on infinite grid},
  journal      = {Comput. Sci.},
  volume       = {22},
  number       = {1},
  year         = {2021},
  url          = {https://doi.org/10.7494/csci.2021.22.1.3840},
  doi          = {10.7494/CSCI.2021.22.1.3840},
  bibsource    = {dblp computer science bibliography, https://dblp.org}
}

@InProceedings{feletti_et_al:LIPIcs.DISC.2024.46,
  author =	{Feletti, Caterina and Pattanayak, Debasish and Sharma, Gokarna},
  title =	{{Brief Announcement: Optimal Uniform Circle Formation by Asynchronous Luminous Robots}},
  booktitle =	{38th International Symposium on Distributed Computing (DISC 2024)},
  pages =	{46:1--46:7},
  series =	{Leibniz International Proceedings in Informatics (LIPIcs)},
  ISBN =	{978-3-95977-352-2},
  ISSN =	{1868-8969},
  year =	{2024},
  volume =	{319},
  editor =	{Alistarh, Dan},
  publisher =	{Schloss Dagstuhl -- Leibniz-Zentrum f{\"u}r Informatik},
  address =	{Dagstuhl, Germany},
  URL =		{https://drops.dagstuhl.de/entities/document/10.4230/LIPIcs.DISC.2024.46},
  URN =		{urn:nbn:de:0030-drops-212748},
  doi =		{10.4230/LIPIcs.DISC.2024.46}
}

@INPROCEEDINGS{1500631,
  author={Stormont, D.P.},
  booktitle={CIHSPS 2005. Proceedings of the 2005 IEEE International Conference on Computational Intelligence for Homeland Security and Personal Safety, 2005.}, 
  title={Autonomous rescue robot swarms for first responders}, 
  year={2005},
  volume={},
  number={},
  pages={151-157},
  doi={10.1109/CIHSPS.2005.1500631}}

@article{TAN201318,
title = {Research Advance in Swarm Robotics},
journal = {Defence Technology},
volume = {9},
number = {1},
pages = {18-39},
year = {2013},
issn = {2214-9147},
doi = {https://doi.org/10.1016/j.dt.2013.03.001},
url = {https://www.sciencedirect.com/science/article/pii/S221491471300024X},
author = {Ying Tan and Zhong-yang Zheng}
}

@article{PATTANAYAK2019145,
title = {Gathering of mobile robots with weak multiplicity detection in presence of crash-faults},
journal = {Journal of Parallel and Distributed Computing},
volume = {123},
pages = {145-155},
year = {2019},
issn = {0743-7315},
doi = {https://doi.org/10.1016/j.jpdc.2018.09.015},
url = {https://www.sciencedirect.com/science/article/pii/S074373151830697X},
author = {Debasish Pattanayak and Kaushik Mondal and Ramesh H. and Partha Sarathi Mandal}}

@article{DBLP:journals/dc/CiceroneSN19,
  author       = {Serafino Cicerone and
                  Gabriele Di Stefano and
                  Alfredo Navarra},
  title        = {Asynchronous Arbitrary Pattern Formation: the effects of a rigorous
                  approach},
  journal      = {Distributed Comput.},
  volume       = {32},
  number       = {2},
  pages        = {91--132},
  year         = {2019},
  url          = {https://doi.org/10.1007/s00446-018-0325-7},
  doi          = {10.1007/S00446-018-0325-7},
  bibsource    = {dblp computer science bibliography, https://dblp.org}
}

@article{DBLP:journals/siamcomp/CohenP05,
  author    = {Reuven Cohen and
               David Peleg},
  title     = {Convergence Properties of the Gravitational Algorithm in Asynchronous
               Robot Systems},
  journal   = {{SIAM} J. Comput.},
  volume    = {34},
  number    = {6},
  pages     = {1516--1528},
  year      = {2005},
  url       = {https://doi.org/10.1137/S0097539704446475},
  doi       = {10.1137/S0097539704446475},
  bibsource = {dblp computer science bibliography, https://dblp.org}
}

@inproceedings{DBLP:conf/icalp/CieliebakFPS03,
  author    = {Mark Cieliebak and
               Paola Flocchini and
               Giuseppe Prencipe and
               Nicola Santoro},
  editor    = {Jos C. M. Baeten and
               Jan Karel Lenstra and
               Joachim Parrow and
               Gerhard J. Woeginger},
  title     = {Solving the Robots Gathering Problem},
  booktitle = {Automata, Languages and Programming, 30th International Colloquium,
               {ICALP} 2003, Eindhoven, The Netherlands, June 30 - July 4, 2003.
               Proceedings},
  series    = {Lecture Notes in Computer Science},
  volume    = {2719},
  pages     = {1181--1196},
  publisher = {Springer},
  year      = {2003},
  url       = {https://doi.org/10.1007/3-540-45061-0\_90},
  doi       = {10.1007/3-540-45061-0\_90},
  bibsource = {dblp computer science bibliography, https://dblp.org}
}

@article{DBLP:journals/tcs/FlocchiniPSW05,
  author    = {Paola Flocchini and
               Giuseppe Prencipe and
               Nicola Santoro and
               Peter Widmayer},
  title     = {Gathering of asynchronous robots with limited visibility},
  journal   = {Theor. Comput. Sci.},
  volume    = {337},
  number    = {1-3},
  pages     = {147--168},
  year      = {2005},
  url       = {https://doi.org/10.1016/j.tcs.2005.01.001},
  doi       = {10.1016/j.tcs.2005.01.001},
  bibsource = {dblp computer science bibliography, https://dblp.org}
}

@inproceedings{DBLP:conf/sirocco/SuzukiY96,
  author    = {Ichiro Suzuki and
               Masafumi Yamashita},
  editor    = {Nicola Santoro and
               Paul G. Spirakis},
  title     = {Distributed Anonymous Mobile Robots},
  booktitle = {SIROCCO'96, The 3rd International Colloquium on Structural Information
               {\&} Communication Complexity, Siena, Italy, June 6-8, 1996},
  pages     = {313--330},
  publisher = {Carleton Scientific},
  year      = {1996},
  bibsource = {dblp computer science bibliography, https://dblp.org}
}

@article{DBLP:journals/tcs/FlocchiniPSW08,
  author    = {Paola Flocchini and
               Giuseppe Prencipe and
               Nicola Santoro and
               Peter Widmayer},
  title     = {Arbitrary pattern formation by asynchronous, anonymous, oblivious
               robots},
  journal   = {Theor. Comput. Sci.},
  volume    = {407},
  number    = {1-3},
  pages     = {412--447},
  year      = {2008},
  url       = {https://doi.org/10.1016/j.tcs.2008.07.026},
  doi       = {10.1016/j.tcs.2008.07.026},
  bibsource = {dblp computer science bibliography, https://dblp.org}
}

\end{document}